\documentclass[]{interact}

\usepackage{epstopdf}% To incorporate .eps illustrations using PDFLaTeX, etc.
\usepackage[caption=false]{subfig}% Support for small, `sub' figures and tables
\usepackage[natbibapa,nodoi]{apacite}% Citation support using apacite.sty. Commands using natbib.sty MUST be deactivated first!
\theoremstyle{plain}% Theorem-like structures provided by amsthm.sty
\newtheorem{theorem}{Theorem}[section]
\newtheorem{lemma}[theorem]{Lemma}

\newtheorem{proposition}[theorem]{Proposition}

\theoremstyle{definition}
\newtheorem{definition}[theorem]{Definition}

\theoremstyle{remark}
\newtheorem{remark}{Remark}

\renewcommand{\phi}{\varphi}

\usepackage{upgreek, comment}
\usepackage{multicol}
\usepackage{bussproofs}
\usepackage[all]{xy}

  \newtheorem{innercustomthm}{Theorem}
\newenvironment{customthm}[1]
  {\renewcommand\theinnercustomthm{#1}\innercustomthm}
  {\endinnercustomthm}

\newcommand{\iffi}{\textit{iff} }

\newcommand{\dfn}{Definition}
\newcommand{\thm}{Theorem}
\newcommand{\lem}{Lemma}
\newcommand{\prp}{Proposition}
\newcommand{\sect}{Section}
\newcommand{\cptr}{Chapter}

\newcommand{\fig}{Figure}
\newcommand{\rmk}{Remark}
\newcommand{\app}{Appendix}

\newcommand{\nipc}{\mathsf{N}_{\ipc}}
\newcommand{\ngd}{\mathsf{N}_{\gd}}
\newcommand{\nnd}{\mathsf{N}_{\nd}}
\newcommand{\ncdl}{\mathsf{N}_{\cdl}}
\newcommand{\nndl}{\mathsf{N}_{\ndl}}
\newcommand{\ncd}{\mathsf{N}_{\cd}}
\newcommand{\ncalc}{\mathsf{N}}
\newcommand{\hqlc}{\mathsf{HQLC}}
\newcommand{\hif}{\mathsf{HIF}}

\newcommand{\X}{\mathsf{X}}
\newcommand{\lab}{\mathrm{Lab}}

\newcommand{\prop}{\Upphi}
\newcommand{\pred}{\Uppsi}
\newcommand{\vars}{\mathrm{Var}}

\newcommand{\parama}{y}

\newcommand{\imp}{\supset}

\newcommand{\langp}{\mathcal{L}_{P}}
\newcommand{\langq}{\mathcal{L}_{Q}}

\newcommand{\ari}[1]{ar(#1)}

\newcommand{\cd}{\textrm{CD}}
\newcommand{\nd}{\textrm{ND}}
\newcommand{\cdl}{\textrm{GC}}
\newcommand{\ndl}{\textrm{GN}}
\newcommand{\gd}{\textrm{GD}}
\newcommand{\ipc}{\textrm{I}}
\newcommand{\logic}{\textrm{L}}

\newcommand{\lin}{(lin)}

\newcommand{\nsa}{\Sigma}
\newcommand{\nsb}{\Gamma}
\newcommand{\nsc}{\Delta}
\newcommand{\nsd}{\Pi}

\newcommand{\ns}{\Sigma}

\newcommand{\nant}{\Gamma}
\newcommand{\ncon}{\Delta}
\newcommand{\nx}{\Sigma}
\newcommand{\ny}{\Gamma}

\newcommand{\nest}[2]{[#1]_{#2}}
\newcommand{\onest}[2]{[#1]_{#2}}

\newcommand{\lb}{[}
\newcommand{\rb}{]}
\newcommand{\hol}{\{}
\newcommand{\hor}{\}}

\newcommand{\inp}{\bullet}
\newcommand{\outp}{\circ}
\newcommand{\io}{\star}
\newcommand{\rable}[2]{#1 \twoheadrightarrow^{*} #2}

\newcommand{\rtable}[2]{#1 \twoheadrightarrow^{+} #2}

\newcommand{\id}{(id)}
\newcommand{\idfo}{(id_{Q})}
\newcommand{\botl}{(\bot^{\inp})}

\newcommand{\botr}{(\bot^{\outp})}

\newcommand{\disr}{(\lor^{\outp})}
\newcommand{\conr}{(\land^{\outp})}

\newcommand{\impl}{(\imp^{\inp})}
\newcommand{\impr}{(\imp^{\outp})}

\newcommand{\existsri}{(\exists^{\outp}_{1})}
\newcommand{\existsrii}{(\exists^{\outp}_{2})}
\newcommand{\allr}{(\forall^{\outp})}
\newcommand{\disl}{(\lor^{\inp})}
\newcommand{\conl}{(\land^{\inp})}
\newcommand{\existsl}{(\exists^{\inp})}

\newcommand{\allli}{(\forall^{\inp}_{1})}
\newcommand{\alllii}{(\forall^{\inp}_{2})}

\newcommand{\cdr}{(cd)}
\newcommand{\ddr}{(dd)}
\newcommand{\ndr}{(nd)}
\newcommand{\doms}{(ds)}

\newcommand{\psub}{(ps)}
\newcommand{\lsub}{(ls)}
\newcommand{\wk}{(wk)}
\newcommand{\wkv}{(wv)}
\newcommand{\ctrv}{(cv)}
\newcommand{\ex}{(ex)}

\newcommand{\ctrr}{(ctr^{\outp})}
\newcommand{\ctrl}{(ctr^{\inp})}

\newcommand{\cut}{(cut)}

\newcommand{\ctr}{(ctr^{\io})}

\newcommand{\mrg}{(mrg)}
\newcommand{\lwr}{(lwr)}
\newcommand{\lft}{(lft)}

\newcommand{\ec}{(ec)}

\newcommand{\nec}{(n)}

\newcommand{\ru}{(r)}
\newcommand{\rone}{(r_{1})}
\newcommand{\rtwo}{(r_{2})}

\newcommand{\prf}{\mathcal{D}}
\newcommand{\va}{\mathrm{X}}
\newcommand{\vb}{\mathrm{Y}}
\newcommand{\vc}{\mathrm{Z}}

\newcommand{\fint}[1]{f(#1)}

\newcommand{\branch}{\mathcal{B}}
\newcommand{\prove}{\mathtt{Prove}}
\newcommand{\success}{\mathtt{True}}

\newcommand{\dlc}{(dlc)}
\newcommand{\dbc}{(dbc)}
\newcommand{\bdtwo}{\mathrm{Bd_{2}}}

\begin{document}

%\articletype{ARTICLE TEMPLATE}% Specify the article type or omit as appropriate

\title{Nested Sequents for Intermediate Logics: The Case of G\"odel-Dummett Logics}

\author{
\name{Tim S. Lyon\thanks{CONTACT Tim S. Lyon. Email: timothy\_stephen.lyon@tu-dresden.de}}
\affil{Institute of Artificial Intelligence, Technische Universit\"at Dresden, Dresden, Germany}
}

\maketitle

\begin{abstract}
We present nested sequent systems for propositional G\"odel-Dummett logic and its first-order extensions with non-constant and constant domains, built atop nested calculi for intuitionistic logics. To obtain nested systems for these G\"odel-Dummett logics, we introduce a new structural rule, called the \emph{linearity rule}, which (bottom-up) operates by linearizing branching structure in a given nested sequent. In addition, an interesting feature of our calculi is the inclusion of \emph{reachability rules}, which are special logical rules that operate by propagating data and/or checking if data exists along certain paths within a nested sequent. Such rules require us to generalize our nested sequents to include \emph{signatures} (i.e. finite collections of variables) in the first-order cases, thus giving rise to a generalization of the usual nested sequent formalism. Our calculi exhibit %a wealth of 
 favorable properties, admitting the height-preserving invertibility of every logical rule and the (height-preserving) admissibility of a large collection of structural and reachability rules. We prove all of our systems sound and cut-free complete, and show that syntactic cut-elimination obtains for the intuitionistic systems. We conclude the paper by discussing possible extensions and modifications, putting forth an array of structural rules that could be used to provide a sizable class of intermediate logics with cut-free nested sequent systems.
\end{abstract}

\begin{keywords}
Admissibility; cut-elimination; first-order; intermediate logic; invertibility; nested sequent: proof theory; reachability rule, signature
\end{keywords}

%The Introduction
\section{Introduction}

 Intermediate logics are fragments of classical logic subsuming intuitionistic logic, and thus exist as logics `intermediate' between the former and the latter. In this paper, we study the proof theory of a set of intermediate logics referred to as \emph{G\"odel-Dummett logics} as well as the intuitionistic logics they are based upon. Such logics have attracted considerable attention in the literature. For instance, propositional G\"odel-Dummett logic was used by \citet{God32} to prove that intuitionistic logic does not have a finite characteristic matrix, \citet{Vis82} used the logic in an analysis of Heyting arithmetic, and \citet{LifPeaVal01} employed a variation of the logic to study the strong equivalence of logic programs. Moreover, G\"odel-Dummett logics have been recognized as blending the paradigms of fuzzy and constructive reasoning~\citep{Avr91,BaaPreZac07,BaaZac00,Haj98,TakTit84}. 
 
 We consider the proof theory of three G\"odel-Dummett logics in this paper, namely,
\begin{description}

\item[(1)] Propositional G\"odel-Dummett logic ($\gd$),

\item[(2)] First-order G\"odel-Dummett logic with non-constant domains ($\ndl$), and 

\item[(3)] First-order G\"odel-Dummett logic with constant domains ($\cdl$).

\end{description}
  We also consider the proof theory of their base intuitionistic logics; in particular, 
\begin{description}

\item[(4)] Propositional intuitionistic logic ($\ipc$), 

\item[(5)] First-order intuitionistic logic with non-constant domains ($\nd$), and 

\item[(6)] First-order intuitionistic logic with constant domains ($\cd$).

\end{description}
  The logics $\gd$, $\ndl$, and $\cdl$ can be obtained by extending the axiomatizations of $\ipc$, $\nd$, and $\cd$, respectively, with the \emph{linearity axiom} $(\phi \imp \psi) \lor (\psi \imp \phi)$ (cf.~\citet{GabSheSkv09}). In semantic terms, the logics $\gd$, $\ndl$, and $\cdl$ can be seen as the set of valid formulae over Kripke frames for $\ipc$, $\nd$, and $\cd$, respectively, which additionally satisfy the \emph{connectivity condition}, i.e. for worlds $w$, $u$, and $v$, if $w$ relates to $u$ and $v$ via the accessibility relation, then either $u$ relates to $v$ or $v$ relates to $u$ (cf.~\citet{GabSheSkv09}). As Kripke frames for intuitionistic logics have an accessibility relation that is a pre-order, this constraint has the effect that Kripke frames for G\"odel-Dummett logics are \emph{linear}. Moreover, in the first-order setting every world of a Kripke frame is associated with a non-empty domain of elements used to interpret quantificational formulae. For the non-constant domain logics $\nd$ and $\ndl$, these domains are permitted to grow along the accessibility relation, whereas for the constant domain logics $\cd$ and $\cdl$ they are held constant at each world. This behavior is reflected in the axiomatizations of $\nd$, $\ndl$, $\cd$, and $\cdl$ whereby the latter two logics include the \emph{quantifier shift axiom} $\forall x (\phi \lor \psi) \imp \forall x \phi \lor \psi$ where $x$ does not occur free in $\psi$~\citep{GabSheSkv09,Grz64}, and the former two logics omit it.

 The central aim of this paper is to provide a uniform and modular proof theory for the above six logics in the style of \emph{nested sequents}. A nested sequent is a formula encoding a \emph{tree} whose nodes are (pairs of) multisets of formulae, and a nested sequent system (or, calculus) is a set of inference rules that operate over such. The paradigm of nested sequents serves as a proper generalization of Gentzen's sequent calculus formalism~\citep{Gen35a,Gen35b} and was initiated by \citet{Bul92} and \citet{Kas94}. The framework of nested sequents was then subsequently expanded upon in a sequence of works by Br\"unnler and Poggiolesi~\citep{Bru06,Bru09,Pog09,Pog09b,Pog10} whereby the authors explored admissibility, invertibility, and cut-elimination properties in the context of modal logics. The introduction of these systems was largely motivated by the search for \emph{analytic} proof systems, which operate by step-wise (de)composing logical formulae. This method of deduction has the effect that proofs generated within such systems typically exhibit the \emph{sub-formula property}, i.e. every formula occurring in a proof is a subformula of the conclusion of the proof. As the well-known cut rule deletes formulae from the premises to the conclusion in an inference, thus violating the sub-formula property, analytic systems are normally \emph{cut-free} and do not require the cut rule for completeness. The analytic quality of nested systems has proven them useful in a variety of cases; for instance, nested sequent calculi have been employed in constructive proofs of interpolation~\citep{FitKuz15,LyoTiuGorClo20}, in writing decision procedures~\citep{Lyo21thesis,TiuIanGor12}, and in establishing complexity-hardness results~\citep{LyoGom22}. 
 
 A characteristic feature of nested calculi is the incorporation of \emph{propagation rules}~\citep{CasCerGasHer97,Fit72,GorPosTiu08} and/or \emph{reachability rules}~\citep{Fit14,Lyo21thesis}. Propagation rules function by propagating formulae along certain paths within the tree structure of a nested sequent, whereas reachability rules have the additional functionality of checking if data exists along certain paths within a nested sequent (see~\citet{Lyo21thesis} for a discussion). The latter class of rules was motivated by the work of \citet{Fit14}, who provided a mechanism (referred to as \emph{availability}) for capturing both non-constant and constant domain variants of first-order intuitionistic logic in a single nested sequent framework (and equivalent prefixed tableaux). In essence, Fitting shows that through the imposition or dismissal of a certain side condition on quantifier rules one can capture \nd \ and \cd, respectively, in a single nested calculus presentation. Thus, enforcing or not enforcing this side condition permits one to `toggle' between the non-constant and constant domain variants of first-order intuitionistic logic.
 
 In this paper, we present a modified version of Fitting's account to capture both non-constant and constant domain reasoning within a single (analytic and cut-free) nested sequent framework. In particular, we generalize the structure of nested sequents to include a multiset of variables (called a \emph{signature}) at every node in the tree encoded by a nested sequent, similar to what was done with hypersequents to capture the logic \ndl \ in~\citet{Tiu11}. We then define reachability rules that operate by searching for variables through %select 
 paths of a nested sequent, instantiating quantificational formulae with such terms when applied bottom-up. We note that our use of signatures in nested sequents is helpful for proving cut-free completeness, namely, in extracting counter-models from failed proof-search. %We note that our use of signatures in nested sequents is crucial for our results and begets distinct advantages over Fitting's account (which neglects the use of signatures in nested sequents). As a case in point, we show that if we disregard the use of signatures in our reachability rules, then certain rules of our nested systems become non-invertible, which obstructs our proof of syntactic cut-elimination for \ipc, \nd, and \cd. 
 As a consequence of our formulation, we show that we obtain nested systems with desirable proof-theoretic properties, such as the (height-preserving) admissibility of a sizable class of structural rules and the (height-preserving) invertibility of all logical rules. Similar properties were first shown to hold for nested systems in the context of modal logics~\citep{Bru09,Pog09}, and typically endow nested systems with certain advantages over calculi built within other proof-theoretic formalisms. For example, the hypersequent systems of \citet{BaaZac00} and \cite{Tiu11} for $\cdl$ and $\ndl$, respectively, include non-invertible logical rules, which obfuscates the extraction of counter-models from failed proof-search. By contrast, our nested systems circumvent this issue as all logical rules are invertible.
 
 Beyond generalizing nested sequents with signatures and formulating a new class of reachability rules, we also introduce a novel structural rule, called the \emph{linearity rule} $\lin$, which lets us pass from nested systems for intuitionistic logics to nested systems for G\"odel-Dummett logics. The linearity rule offers a unique functionality, which \emph{linearizes} nested sequents when applied bottom-up, and differs from the various rules given in the literature on G\"odel-Dummett logics, being distinct from the communication rule used in hypersequents~\citep{Avr91}, the connected rule used in labeled sequents~\citep{DycNeg12}, and the $\rightarrow_{R}^{2}$ rule used in linear nested sequents~\citep{KuzLel18}. We show that our nested systems for \ipc, \nd, and \cd \ become sound and complete for \gd, \ndl, and \cdl, respectively, when extended with the $\lin$ rule.
 
 To summarize, our paper accomplishes the following: (1) We generalize the formalism of nested sequents to include signatures, permitting us to define a new class of reachability rules suitable for toggling between non-constant and constant domain reasoning. (2) We provide a uniform, modular, and analytic nested sequent presentation of the above six logics showing all systems sound and cut-free complete. (3) We introduce a novel structural rule capturing linear reasoning, whose presence or omission lets us pass between nested calculi for intuitionistic and G\"odel-Dummett logics. (4) We show that useful structural and reachability rules are (height-preserving) admissible in our nested calculi, that all logical rules are (height-preserving) invertible, and show that the nested calculi for \ipc, \nd, and \cd \ satisfy a syntactic cut-elimination theorem. (5) We define a novel class of nested structural rules, conjecturing that a wide array of intermediate logics can be captured by extending the nested calculi for \ipc, \nd, and \cd \ with such rules, thus making progress toward the development of a general theory of nested sequents for intermediate logics.

The paper is organized as follows: In \sect~\ref{sec:log-prelims-I}, we give the semantics and axiomatizations for the six logics mentioned above. In the subsequent section (\sect~\ref{sec:nested-calculi}), we define our nested sequent calculi showing them sound and cut-free complete. In \sect~\ref{sec:properties}, we establish (height-preserving) admissibility and invertibility results for our nested systems, and in \sect~\ref{sec:cut-elim}, we prove syntactic cut-elimination for the \ipc, \nd, and \cd \ systems. In the final section (\sect~\ref{sec:conclusion}) we discuss possible extension of our framework to capture other intermediate logics with nested sequents, and conclude.

%Introducing Logics, Semantics, Axiomatizations
\section{Logical Preliminaries}
\label{sec:log-prelims-I}

 We now introduce the intermediate logics that will be discussed throughout the paper. In the first subsection, we explain the semantics and give the axiomatizations for two logics: intuitionistic propositional logic and G\"odel-Dummett logic. In the second subsection, we extend propositional intuitionistic logic and G\"odel-Dummett logic to the first-order case, giving the semantics and axiomatizations for the non-constant and constant domain versions.

\subsection{Intuitionistic and G\"odel-Dummett Logic}

 We let $\prop = \{p, q, r, \ldots\}$ be a set of denumerably many \emph{propositional variables} and we define our language $\langp$ to be the set of all formulae generated from the following grammar in BNF:
$$
\phi ::= p \ \vert \ \bot \ \vert \ (\phi \lor \phi) \ \vert \ (\phi \land \phi)  \ \vert \ (\phi \imp \phi)
$$
 where $p$ ranges over $\prop$. We use lower-case Greek letters $\phi$, $\psi$, $\chi$, $\ldots$ to denote formulae and define $\top = p \imp p$ for a fixed propositional variable $p$. The complexity of a formula $\phi$, written $\vert\phi\vert$, is recursively defined as follows: (i) $\vert p \vert = \vert \bot \vert := 0$ and (ii) $\vert \phi \star \psi \vert := \vert \phi \vert + \vert \psi \vert + 1$ for $\star \in \{\lor, \land, \imp\}$. We present a Kripke-style semantics for our logics (cf.~\citet{GabSheSkv09}), defining the frames and models used first, and explaining how formulae are evaluated over them second.

\begin{definition}[Frame, Model]\label{def:frame-model} We define two types of frames with the second type extending the properties of the first:
\begin{itemize}

\item An \emph{\ipc-frame} is a pair $F = (W,\leq)$ such that $W$ is a non-empty set $\{w, u, v, \ldots\}$ of \emph{worlds} and $\leq \ \subseteq W \times W$ is a reflexive and transitive binary relation on $W$.\footnote{The properties imposed on $\leq$ are defined as follows: (reflexivity) for all $w \in W$, $w \leq w$, and (transitivity) for all $w, u, v \in W$, if $w \leq v$ and $v \leq u$, then $w \leq u$.}

\item A \emph{\gd-frame} is an \ipc-frame that also satisfies the following \emph{connectivity condition}: if $w \leq u$ and $w \leq v$, then either $u \leq v$ or $v \leq u$.
\end{itemize}

 We define an \emph{\ipc-model} and \emph{\gd-model} $M$ to be an ordered pair $(F,V)$ where $F$ is an \ipc-frame or \gd-frame, respectively, and where $V$ is a \emph{valuation function} such that $V(p,w) \subseteq \{()\}$, where $()$ is the empty tuple, meaning $V(p,w) = \{()\}$ or $V(p,w) = \emptyset$, and which satisfies the \emph{monotonicity condition}: (M) If $w \leq v$, then $V(p,w) \subseteq V(p,v)$.\footnote{As specified in \dfn~\ref{def:semantic-clauses}, we interpret $() \in V(p,w)$ to mean that $p$ holds at $w$, and $() \not\in V(p,w)$ to mean that $p$ does not hold at $w$. We define the valuation function in the manner described as it easily generalizes to the first-order case.}
\end{definition}

 %When interpreting first-order formulae from $\lang$, we define the satisfaction of \emph{$\langd$-sentences} at worlds, rather than formulae from $\lang$, as with first-order intuitionistic logics~\cite{GabSheSkv09}. Given that a domain $D(w) = \{\parama, \paramb, \paramc, \ldots\}$ is a non-empty set of constants, we define $\langd$ to be the language $\lang$ extended with the constants from $D(w)$, and we define an $\langd$-sentence to be a sentence in the language $\langd$, that is, a formula in $\langd$ that does not contain any free variables. When interpreting propositional formulae from $\langp$, we need not make use of $\langd$-sentences as no formula in $\langp$ contains free variables. We note that the following definition uniformly describes how formulae are interpreted from both $\langp$ and $\lang$.

 We remark that the connectivity condition imposed on \gd-frames/models implies that the $\leq$ relation is \emph{linear}, i.e. for any two worlds $w$ and $u$, either $w \leq u$ or $u \leq w$.

%\cite{GabSheSkv09}
\begin{definition}[Semantic Clauses]
\label{def:semantic-clauses} Let $M$ be an \ipc- or \gd-model with $w \in W$. We interpret formulae by means of the following clauses:

\begin{itemize}

%\item $M,w \Vdash \top$;

%\item $M,w \not\Vdash \bot$;

\item $M,w \Vdash p$ \iffi $() \in V(p,w)$;

\item $M,w \not\Vdash \bot$;

\item $M,w \Vdash \phi \lor \psi$ \iffi $M,w \Vdash \phi$ or $M,w \Vdash \psi$;

\item $M,w \Vdash \phi \land \psi$ \iffi $M,w \Vdash \phi$ and $M,w \Vdash \psi$;

\item $M,w \Vdash \phi \imp \psi$ \iffi for all $u \in W$, if $w \leq u$ and $M,u \Vdash \phi$, then $M,u \Vdash \psi$;

\item $M \Vdash \phi$ \iffi $M,u \Vdash \phi$ for worlds $u \in W$ of $M$.

\end{itemize}
A formula $\phi \in \langp$ is \emph{\ipc-valid} or \emph{\gd-valid} \iffi $M \Vdash \phi$ for all \ipc-models or \gd-models $M$, respectively.
\end{definition}

 The following generalized version of the monotonicity property holds on \ipc- and \gd-models, and can be shown by induction on the complexity of $\phi$~\citep{GabSheSkv09}.

\begin{proposition}\label{prop:monotonicity} Let $M$ be an \ipc- or \gd-model. For any formula $\phi \in \langp$, if $M,w \Vdash \phi$ and $w \leq v$, then $M,v \Vdash \phi$.
\end{proposition}

%\begin{proof}
%By induction on the complexity of $\phi$.
%\end{proof}

\begin{definition}[Axioms]\label{def:axioms} We define the propositional logics by means of the following set of axioms:
\begin{multicols}{2}
\begin{description}

\item[A0] $\phi \supset (\psi \supset \phi)$

\item[A1] $(\phi \supset (\psi {\supset} \chi)) \supset ((\phi {\supset} \psi) \supset (\phi {\supset} \chi))$

\item[A2] $\phi \supset (\psi \supset (\phi \land \psi))$

\item[A3] $(\phi \land \psi) \supset \phi$

\item[A4] $(\phi \land \psi) \supset \psi$

\item[A5] $\phi \supset (\phi \lor \psi)$

\item[A6] $\psi \supset (\phi \lor \psi)$

\item[A7] $(\phi {\supset} \chi) \supset ((\psi {\supset} \chi) \supset ((\phi {\lor} \psi) \supset \chi))$

\item[A8] $\bot \imp \phi$

\item[A9] $(\phi \imp \psi) \lor (\psi \imp \phi)$

\item[R0] \AxiomC{$\phi$}\AxiomC{$\phi \imp \psi$}\RightLabel{mp}\BinaryInfC{$\psi$}\DisplayProof

\end{description}
\end{multicols}
 We define \emph{intuitionistic propositional logic} \ipc \ to be the smallest set of formulae from $\langp$ closed under substitutions of the axioms A0--A8 and applications of the inference rule R0. We define \emph{G\"odel-Dummett logic} \gd \ to be the smallest set of formulae from $\langp$ closed under the axioms A0--A9 and applications of the inference rule R0. We refer to axiom A9 as the \emph{linearity axiom}. For $\logic \in \{\ipc, \gd\}$, we write $\vdash_{\logic} \phi$ to denote that $\phi$ is an element, or \emph{theorem}, of $\logic$.
\end{definition}

% The following soundness and completeness results are well-known and follow from the work by~\citet{GabSheSkv09}.

 The following soundness and completeness results are well-known; cf.~\citet{GabSheSkv09}.

\begin{theorem}[Soundness and Completeness]\label{thm:sound-complete-logics}  For $\phi \in \langp$, $\vdash_{\logic} \phi$ \iffi $\phi$ is \logic-valid with $\logic \in \{\ipc,\gd\}$.
\end{theorem}

\subsection{First-order Intuitionistic and G\"odel-Dummett Logics}

 We let $\vars := \{x, y, z, \ldots\}$ be a denumerable set of \emph{variables}. Our first-order language includes \emph{atomic formulae} of form $p(x_{1}, \ldots, x_{n})$, which are obtained by prefixing an $n$-ary predicate $p$ from a set $\pred := \{p, q, r, \ldots\}$ of denumerably many predicates of each arity $n \in \mathbb{N}$ to a tuple of variables of length $n$. We let $\ari{p}$ denote the arity of a predicate and refer to predicates of arity $0$ as \emph{propositional variables}. We will often write a list of variables $x_{1}, \ldots, x_{n}$ as $\vec{x}$, and similarly, will write atomic formulae of the form $p(x_{1}, \ldots, x_{n})$ as $p(\vec{x})$. The first-order language $\langq$ is defined to be the set of all formulae generated from the following grammar in BNF:
$$
\phi ::= p(\vec{x}) \ \vert \ \bot \ \vert \ (\phi \lor \phi) \ \vert \ (\phi \land \phi)  \ \vert \ (\phi \imp \phi) \ \vert \ (\exists x \phi) \ \vert \ (\forall x \phi)
$$
 where $p$ ranges over $\pred$, and the variables $\vec{x} = x_{1}, \ldots, x_{n}$ and $x$ range over the set $\vars$. We use lower-case Greek letters $\phi$, $\psi$, $\chi$, $\ldots$ to denote formulae.

 As usual, we say that the occurrence of a variable $x$ in $\phi$ is \emph{free} given that $x$ does not occur within the scope of a quantifier. %Also, we will occasionally make use of \emph{substitutions} on formulae, letting $\phi(t_{1}/x_{1},\ldots, t_{k}/x_{k})$ denote the substitution of the terms $t_{1}, \ldots, t_{k}$ for each free occurrence of the variables $x_{1}, \ldots, x_{k}$ in $\phi$, respectively. For simplicity, we assume w.l.o.g. that $\phi(t_{1}/x_{1},\ldots, t_{k}/x_{k})$ does not contain any quantified occurrences of $x_{1}, \ldots, x_{k}$. In addition, we let $\phi(y/x)$ denote the substitution of the variable $y$ for all free occurrences of the variable $x$ in $\phi$, and we say that \emph{$y$ is free for $x$ in $\phi$} when substituting $y$ for $x$ in $\phi$ does not cause $y$ to become bound by a quantifier; e.g. $y$ is free for $x$ in $\forall z (p(x) \land q(z))$, but not in $\forall y (p(x) \land q(y))$. 
 We say that \emph{$y$ is free for $x$ in $\phi$} when substituting $y$ for $x$ in $\phi$ does not cause $y$ to become bound by a quantifier; e.g. $y$ is free for $x$ in $\forall z (p(x) \land q(z))$, but not in $\forall y (p(x) \land q(y))$. In addition, we let $\phi(y/x)$ denote the substitution of the variable $y$ for all free occurrences of the variable $x$ in $\phi$, possibly renaming bound variables to ensure that $y$ is free for $x$ in $\phi$. We extend the definition of the \emph{complexity} of a formula from the previous section with the following case: $\vert Q x \phi \vert := \vert \phi \vert + 1$ for $Q \in \{\forall, \exists\}$.

 As before, we follow the work of ~\citet{GabSheSkv09}, and define a Kripke-style semantics for our first-order logics.

\begin{definition}[Frame, Model]\label{def:frame-model-fo} We define four types of frames:
\begin{itemize}

\item An \emph{\nd-frame} is a triple $F = (W,\leq,D)$ such that $(W,\leq)$ is an \ipc-frame and $D$ is a \emph{domain function} mapping a world $w \in W$ to a non-empty set $D(w)$ satisfying the \emph{nested domain condition}: (ND) If $d \in D(w)$ and $w \leq v$, then $d \in D(v)$. 

\item A \emph{\cd-frame} is an \emph{\nd-frame} that additionally satisfies the \emph{constant domain condition}: (CD) If $w,u \in W$, then $D(w) = D(u)$. %(CD) If $d \in D(v)$ and $w \leq v$, then $d \in D(w)$. 

\item An \emph{\ndl-frame} is a triple $F = (W,\leq,D)$ such that $(W,\leq)$ is a \gd-frame and $D$ is a \emph{domain function} mapping a world $w \in W$ to a non-empty set $D(w)$ satisfying the nested domain condition (ND). % If $d \in D(w)$ and $w \leq v$, then $d \in D(v)$. 

\item A \emph{\cdl-frame} is an \emph{\ndl-frame} that additionally satisfies the constant domain condition (CD).\footnote{Note that the (ND) condition becomes redundant in the presence of the (CD) condition.} % If $d \in D(v)$ and $w \leq v$, then $d \in D(w)$. 
\end{itemize}
 %We define a \emph{frame} to be an \ipc-frame, \gd-frame, \nd-frame, or \cd-frame.
 For $\logic \in \{\nd,\cd,\ndl,\cdl\}$, we define an \emph{\logic-model} $M$ to be an ordered pair $(F,V)$ where $F$ is an \logic-frame, and where $V$ is a \emph{valuation function} such that $V(p,w) \subseteq D(w)^{n}$ with $n \in \mathbb{N}$, which satisfies the following \emph{monotonicity condition}: (M) If $w \leq v$, then $V(p,w) \subseteq V(p,v)$. We make the simplifying assumption that for each world $w \in W$, $D(w)^{0} = \{()\}$, where $()$ is the empty tuple, meaning $V(p,w) = \{()\}$ or $V(p,w) = \emptyset$, for any propositional variable $p$ (as in \dfn~\ref{def:frame-model}). Thus, the first-order semantics extends the propositional semantics.
 
 Given an \logic-model $M = (W,\leq,D,V)$ with $w \in W$ for $\logic \in \{\nd,\cd,\ndl,\cdl\}$, we define an \emph{$M$-assignment} $\mu : \vars \to D(W)$ to be a function mapping variables to elements of $D(W) = \bigcup_{w \in W} D(w)$. We let $\mu[d/x]$ be the same as $\mu$, but where the variable $x$ is mapped to the element $d \in D(W)$.
\end{definition}

 %When interpreting first-order formulae from $\lang$, we define the satisfaction of \emph{$\langd$-sentences} at worlds, rather than formulae from $\lang$, as with first-order intuitionistic logics~\cite{GabSheSkv09}. Given that a domain $D(w) = \{\parama, \paramb, \paramc, \ldots\}$ is a non-empty set of constants, we define $\langd$ to be the language $\lang$ extended with the constants from $D(w)$, and we define an $\langd$-sentence to be a sentence in the language $\langd$, that is, a formula in $\langd$ that does not contain any free variables. When interpreting propositional formulae from $\langp$, we need not make use of $\langd$-sentences as no formula in $\langp$ contains free variables. We note that the following definition uniformly describes how formulae are interpreted from both $\langp$ and $\lang$.

%\cite{GabSheSkv09}
\begin{definition}[Semantic Clauses]
\label{def:semantic-clauses-fo} Let $\logic \in \{\nd,\cd,\ndl,\cdl\}$ and $M$ be an \logic-model with $w \in W$. We interpret formulae by means of the following clauses:
\begin{itemize}

\item $M,w,\mu \Vdash p(x_{1}, \ldots, x_{n})$ \iffi $(\mu(x_{1}), \ldots, \mu(x_{n})) \in V(p,w)$; % with $\ari{p} = n$;

\item $M,w,\mu \not\Vdash \bot$;

\item $M,w,\mu \Vdash \phi \lor \psi$ \iffi $M,w,\mu \Vdash \phi$ or $M,w,\mu \Vdash \psi$;

\item $M,w,\mu \Vdash \phi \land \psi$ \iffi $M,w,\mu \Vdash \phi$ and $M,w,\mu \Vdash \psi$;

\item $M,w,\mu \Vdash \phi \imp \psi$ \iffi for all $u \in W$, if $w \leq u$ and $M,u,\mu \Vdash \phi$, then $M,u,\mu \Vdash \psi$;

\item $M,w,\mu \Vdash \exists x \phi$ \iffi there exists a $d \in D(w)$ such that $M,w,\mu[d/x] \Vdash \phi$;

\item $M,w,\mu \Vdash \forall x \phi$ \iffi for all $u \in W$ and $d \in D(u)$, if $w \leq u$, then $M, u , \mu[d/x] \Vdash \phi$;

\item $M,w \Vdash \phi$ \iffi $M,w,\mu \Vdash \phi$ for all $M$-assignments $\mu$;

\item $M \Vdash \phi$ \iffi $M,u \Vdash \phi$ for worlds $u \in W$ of $M$.

\end{itemize}
 A formula $\phi \in \langq$ is \emph{\logic-valid} \iffi $M \Vdash \phi$ for all \logic-models $M$. %A formula $\phi \in \langq$ is \emph{\logic-valid} \iffi $M \Vdash \forall \vec{x} \phi$ for all \logic-models, where $\forall \vec{x} \phi$ is the universal closure of $\phi$.
\end{definition}

 As in the propositional setting, a generalized form of monotonicity holds and may be proven by induction on the complexity of $\phi$~\citep{GabSheSkv09}.

\begin{proposition} \label{prop:monotonicity-fo} Let $\logic \in \{\nd,\cd,\ndl,\cdl\}$ with $M$ an \logic-model. For any formula $\phi \in \langq$, if $M,w,\mu \Vdash \phi$ and $w \leq v$, then $M,v,\mu \Vdash \phi$.
\end{proposition}

\begin{proof}
By induction on the complexity of $\phi$.
\end{proof}

\begin{definition}[Axioms]\label{def:axioms-fo} We extend the axiomatizations in \dfn~\ref{def:axioms} to provide axiomatizations for the four first-order (intermediate) logics we consider.
 \begin{multicols}{2}
\begin{description}

\item[A10] $\forall x \phi \supset \phi(y/x)~[\textit{y free for x}]$

\item[A11] $\phi(y/x) \supset \exists x \phi~[\textit{y free for x}]$

\item[A12] $\forall x (\psi \imp \phi(x)) \imp (\psi \imp \forall x \phi(x))$

\item[A13] $\forall x (\phi(x) \imp \psi) \imp (\exists x \phi(x) \imp \psi)$

\item[A14] $\forall x (\phi(x) {\vee} \psi) \imp \forall x \phi(x) {\vee} \psi~[x \not\in \psi]$

\item[R1] \AxiomC{$\phi$}\RightLabel{gen}\UnaryInfC{$\forall x \phi$}\DisplayProof

\end{description}
\end{multicols}
\noindent
 We provide syntactic definitions of each first-order logic accordingly:
\begin{itemize}

\item We define \emph{first-order intuitionistic logic with non-constant domains} \nd \ to be the smallest set of formulae from $\langq$ closed under substitutions of the axioms A0--A8 and A10--A13, and applications of the inference rules R0 and R1.

\item We define \emph{first-order intuitionistic logic with constant domains} \cd \ to be the smallest set of formulae from $\langq$ closed under substitutions of the axioms A0--A8 and A10--A14, and applications of the inference rules R0 and R1.

\item We define \emph{first-order G\"odel-Dummett logic with non-constant domains} \ndl \ to be the smallest set of formulae from $\langq$ closed under substitutions of the axioms A0--A13 and applications of the inference rules R0 and R1.

\item We define \emph{first-order G\"odel-Dummett logic with constant domains} \cdl \ to be the smallest set of formulae from $\langq$ closed under substitutions of the axioms A0--A14 and applications of the inference rules R0 and R1.

\end{itemize}
 We note that axioms A10 and A11 are subject to the side condition (shown in brackets) that $y$ must be free for $x$ and A14 is subject to the side condition (also shown in brackets) that $x$ does not occur free in $\psi$. We refer to axiom A14 as the \emph{constant domain axiom}. For $\logic \in \{\nd, \cd, \ndl, \cdl\}$, we write $\vdash_{\logic} \phi$ to denote that $\phi$ is an element, or \emph{theorem}, of $\logic$.
\end{definition}

 The soundness and completeness of the above logics is well-known~\citep{GabSheSkv09}.

\begin{theorem}[Soundness and Completeness]\label{thm:sound-complete-logics} Let $\logic \in \{\nd,\cd,\ndl,\cdl\}$. For $\phi \in \langq$, $\vdash_{\logic} \phi$ \iffi $\phi$ is \logic-valid.
\end{theorem}

%Introducing Nested Sequent Formalism
\section{Nested Sequent Systems}
\label{sec:nested-calculi}

 We present the nested sequent systems for propositional intuitionistic and G\"odel-Dummett logic first, and then show how these nested calculi can be extended to cover the first-order cases. 

\subsection{Systems for Intuitionistic and G\"odel-Dummett Logic}\label{subsec:prop-nested-systems}

 Motivated by the notation and terminology of~\citet{Str13}, we define $\phi^{\inp}$ to be an \emph{input formula} and $\phi^{\outp}$ to be an \emph{output formula}, for $\phi \in \langp$. We refer to a formula $\phi^{\io}$ with $\io \in \{\inp,\outp\}$ as a \emph{polarized formula} more generally. We call a finite (potentially empty) multiset $\nsb$ of polarized formulae a \emph{flat sequent} and we sometimes write $\nsb^{\io}$ to denote a flat sequent whose polarized formulae are of polarity $\io \in \{\inp,\outp\}$. Last, we let $\lab = \{w_{i} \ \vert \ i \in \mathbb{N} \setminus \{0\}\}$ be a denumerable set of \emph{labels}, and we recursively define a \emph{nested sequent} $\ns$ as follows:
\begin{itemize}

\item Each flat sequent is a nested sequent, and

\item Any object of the form $\nsb, [\nsc_{1}]_{w_{1}}, \ldots, [\nsc_{n}]_{w_{n}}$, where $\nsb$ is a flat sequent and $\nsc_{i}$ is a nested sequent for $1 \leq i \leq n$, is a nested sequent.

\end{itemize}
 We will often use $w$, $u$, $v$, $\ldots$ (occasionally annotated) to denote labels and we make the simplifying assumption that every occurrence of a label in a nested sequent is unique. Note that the incorporation of labels in our nested systems is useful as it simplifies the presentation of our reachability rules below. We use upper-case Greek letters $\nsa$, $\nsb$, $\nsc$, $\ldots$ (occasionally annotated) to denote nested sequents.

 A nice feature of nested sequents is that such objects normally permit a \emph{formula interpretation}~\citep{Bru09,Bul92,Kas94,Pog09}, i.e. each nested sequent may be read as an equivalent formula in the language of the logic. We may utilize this property to lift the semantics from the language $\langp$ to our nested sequents, which proves useful in establishing soundness (\thm~\ref{thm:soundness} below).

\begin{definition}[Formula Interpretation]\label{def:formula-interpretation} Let $\nsb := \nsc^{\inp},\nsd^{\outp}$ with $\nsc^{\inp}$ and $\nsd^{\outp}$ multisets of input and output formulae, respectively. The \emph{formula interpretation} $\fint{\nsa}$ of a nested sequent $\nsa$ of the form $\nsb, [\nsc_{1}]_{w_{1}}, \ldots, [\nsc_{n}]_{w_{n}}$ is recursively defined as follows:
\begin{itemize}

\item $\fint{\nsb} := \displaystyle{\bigwedge \nsc \imp \bigvee \nsd}$

\item $\fint{\nsb, [\nsc_{1}]_{w_{1}}, \ldots, [\nsc_{n}]_{w_{n}}} := \displaystyle{\bigwedge \nsb \imp \bigvee \nsc \lor (\fint{\nsc_{1}} \lor \cdots \lor \fint{\nsc_{n}})}$

\end{itemize}
 We use $\bigwedge$ and $\bigvee$ to denote a conjunction and disjunction of all formulae in a multiset, respectively. As is conventional, we define $\bigwedge \emptyset = \top$ and $\bigvee \emptyset = \bot$. We define a nested sequent $\nsa$ to be \ipc-valid or \gd-valid \iffi $\fint{\nsa}$ is \ipc-valid or \gd-valid, respectively.
\end{definition}

 As witnessed in the definition above, input formulae serve the same purpose as the antecedent of a (traditional) sequent~\citep{Gen35a,Gen35b} and output formulae serve the same purpose as a consequent, that is, the flat sequent $\phi_{1}^{\inp}, \ldots, \phi_{n}^{\inp}, \psi_{1}^{\outp}, \ldots, \psi_{k}^{\outp}$ is simply a mutliset representation of the (traditional) sequent $\phi_{1}, \ldots, \phi_{n} \vdash \psi_{1}, \ldots, \psi_{k}$; we make use of polarized formulae however as it simplifies our presentation and is consistent with notation employed in the literature~\citep{Str13,Lyo21b}. Nested sequents are multisets encoding trees whose nodes are multisets of polarized formulae (i.e. flat sequents) as recognized in the subsequent definition (cf.~\citet{Bru09,Bul92,Kas94,Pog09}).
 
% , which encode trees whose nodes are multisets of polarized formulae (cf.~\citet{Bru09,Bul92,Kas94,Pog09}) as recognized in the subsequent definition.
  
  %We note that input formulae serve the same purpose as the antecedent of a (traditional) sequent and output formulae serve the same purpose as a consequent, that is, $\phi_{1}^{\inp}, \ldots, \phi_{n}^{\inp}, \psi_{1}^{\outp}, \ldots, \psi_{k}^{\outp}$ is simply a mutliset representation of the (traditional) sequent $\phi_{1}, \ldots, \phi_{n} \vdash \psi_{1}, \ldots, \psi_{k}$. We make use of polarized formulae however as it simplifies our presentation.

 %Nested sequents are multisets encoding trees whose nodes are multisets of polarized formulae (cf.~\citet{Bru09,Bul92,Kas94,Pog09}) as recognized in the subsequent definition.
 
\begin{definition}[Tree of a Nested Sequent]\label{def:tree-of-nest-seq} Let $\nsa = \nsb, \nest{\nsc_{1}}{w_{1}}, \ldots, \nest{\nsc_{n}}{w_{n}}$ be a nested sequent. We define the \emph{tree of $\nsa$}, denoted $tr(\nsa) = (V,E)$, recursively on the structure of $\nsa$ as follows:
$$
V = \{(w_{0},\nsb)\} \cup \bigcup_{1 \leq i \leq n} V_{i}
\qquad
E = \{(w_{0},w_{i}) \ | \ 1 \leq i \leq n\} \cup \bigcup_{1 \leq i \leq n} E_{i}
$$
 where $tr(\nsc_{i}) = (V_{i},E_{i})$ for $1 \leq i \leq n$.
\end{definition}

  Given a nested sequent $\nsa = \nsb, \nest{\nsc_{1}}{w_{1}}, \ldots, \nest{\nsc_{n}}{w_{n}}$, we can graphically depict the tree $tr(\nsa)$ of the nested sequent as shown below.
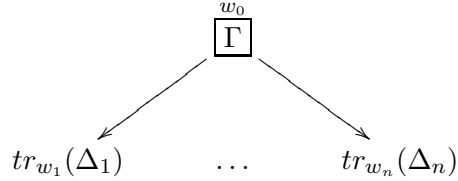
\begin{center}
%\resizebox{\columnwidth}{!}{
\begin{tabular}{c c c}
\xymatrix@C=1em{
%\xymatrix{
 & & \overset{w_{0}}{\boxed{\nsb}}\ar@{->}[dll]\ar@{->}[drr] &  &   		\\
 tr_{w_{1}}(\nsc_{1}) & & \hdots  & & tr_{w_{n}}(\nsc_{n})
}
\end{tabular}
%}
\end{center}
We refer to a flat sequent $\nsc_{i}$ (i.e. a node in the tree above) as a \emph{$w_{i}$-component}, or as a \emph{component} more generally if we do not wish to specify its label. We note that the \emph{root} $\nsb$ is always assumed to be associated with the label $w_{0}$, e.g. $\nsb$ is the $w_{0}$-component in the tree above. Moreover, we use the notation $\nsa\{\nsb_{1}\}_{w_{1}}\cdots\{\nsb_{n}\}_{w_{n}}$ to denote a nested sequent $\nsa$ such that in $tr(\nsa)$ the data $\nsb_{1}, \ldots, \nsb_{n}$ is rooted at $w_{1}, \ldots, w_{n}$, respectively. For example, if $\nsa = p^{\inp},[q^{\inp},[\emptyset]_{u}]_{w},[p \imp q^{\inp}, r \lor \bot^{\outp}]_{v}$, then $\nsa\{p^{\inp}\}_{w_{0}}$, $\nsa\{q^{\inp},[\emptyset]_{u}\}_{w}\{p \imp q^{\inp}\}_{v}$, and $\nsa\{q^{\inp}\}_{w}\{p \imp q^{\inp}\}_{v}$ are all correct representations of $\nsa$ in our notation. In other words, the notation $\nsa\{\nsb_{1}\}_{w_{1}}\cdots\{\nsb_{n}\}_{w_{n}}$ lets us specify data rooted at $w_{i}$-components of a nested sequent. %Notice that with this notation we need to specify \emph{completely} what nested sequent is rooted at a node, as witnessed by the $\nsa\{q,[\emptyset]_{i}\}_{u}\{p \imp q\}_{v}$ and $\nsa\{q\}_{u}\{p \imp q\}_{v}$ examples with respect to the $u$-component. Rather, if we write $\nsa\{\phi^{\inp}\}_{w}$ or $\nsa\{\phi,[\nsb]_{u}\}_{w}$, for instance, we think of this notation is expressing a window that tells us $\phi^{\inp}$ occurs in the $w$-component in the first case, 
 %We note that this notation can be inserted into itself; for instance, $\nsa\{\nsb\{\emptyset\}_{u}\}_{w}$ is a correct representation of $\nsa$ with the nested sequent $\nsb = \emptyset$. 
  We also define a \emph{reachability relation $\rable{}{}$} and \emph{strict reachability relation $\rtable{}{}$} on nested sequents by means of the trees they encode:

\begin{definition}[$\rable{}{}$,$\rtable{}{}$]\label{def:reachability-relation} Let $\nsa$ be a nested sequent with $tr(\ns) = (V,E)$. For two labels $w$ and $u$ occurring in $\nsa$, we say that $u$ is \emph{reachable} from $w$ (written $\rable{w}{u}$) \iffi $w = u$ or there exists a path $(w,v_{1}), \ldots, (v_{n},u) \in E$ from $w$ to $u$. We define $\rtable{w}{u}$ \iffi $\rable{w}{u}$ and $w \neq u$.
%of forward edges from $w$ to $u$ in $tr(\ns)$. We define $\rtable{w}{u}$ \iffi $\rable{w}{u}$ and $w \neq u$.
\end{definition}

 The nested calculi $\nipc$ and $\ngd$ for \ipc \ and \gd \ are displayed in \fig~\ref{fig:nested-calculi} and consist of the \emph{initial rules} $\id$ and $\botl$. With the exception of $\lin$, which we refer to as a \emph{structural rule} (as it only affects the structure of nested sequents), all other rules are \emph{logical rules}. The two nested calculi are defined as collections of these rules:
 
\begin{definition}[$\nipc$, $\ngd$] We define $\nipc$ to be the set consisting of the $\id$, $\botl$, $\disl$, $\disr$, $\conl$, $\conr$, $\impl$, and $\impr$ rules from \fig~\ref{fig:nested-calculi}. We define $\ngd$ to be the set $\nipc \cup \{\lin\}$.
\end{definition}

 In~\citet[\cptr~5]{Lyo21thesis}, nested sequent calculi (referred to as $\mathsf{DIntQ}$ and $\mathsf{DIntQC}$) were provided for first-order intuitionistic logics with non-constant and constant domains. The calculus $\nipc$ serves as the propositional fragment of these systems, and as exhibited in~\citet[\sect~5]{Lyo21thesis}, possesses favorable properties (which will also be discussed in \sect~\ref{sec:properties}). A unique feature of $\nipc$ (and its extension $\ngd$) is the incorporation of the reachability rule $\id$ and the propagation rule $\impl$.\footnote{These rules are referred to as $(id_{*})$ and $(Pr_{\imp})$, respectively, in~\citet[\fig~5.8]{Lyo21thesis}.} Both rules are applicable only if $\rable{w}{u}$ holds, i.e. $\id$ checks if $u$ is reachable from $w$ in a nested sequent with $p^{\inp}$ occurring in the $w$-component and $p^{\outp}$ occurring in the $u$-component, while $\impl$ propagates $\phi^{\outp}$ and $\psi^{\inp}$ along reachable paths when applied bottom-up. Such rules endow our systems with a degree of modularity as changing this side condition yields a nested calculus for another logic. For instance, if we stipulate that $\impl$ is applicable only if $\rtable{w}{u}$ holds (rather than $\rable{w}{u}$), then $\nipc$ becomes a calculus for a sub-intuitionistic logic (cf.~\citet{Res94}).

\begin{figure}[t]
%\noindent\hrule

\begin{center}
\begin{tabular}{c c}
\AxiomC{}
\RightLabel{$\id^{\dag}$}
\UnaryInfC{$\ns \hol p^{\inp} \hor_{w} \hol p^{\outp} \hor_{u}$}
\DisplayProof

&

\AxiomC{}
\RightLabel{$\botl$}
\UnaryInfC{$\nsa \hol  \bot^{\inp} \hor_{w} $}
\DisplayProof
\end{tabular}
\end{center}

%\medskip

\begin{center}
\begin{tabular}{c c c}
\AxiomC{$\nsa \hol \phi^{\inp} \hor_{w} $}
\AxiomC{$\nsa \hol \psi^{\inp} \hor_{w} $}
\RightLabel{$\disl$}
\BinaryInfC{$\nsa \hol \phi \lor \psi^{\inp} \hor_{w} $}
\DisplayProof

&

\AxiomC{$\nsa \hol  \phi^{\outp}, \psi^{\outp} \hor_{w}  $}
\RightLabel{$\disr$}
\UnaryInfC{$\nsa \hol  \phi \lor \psi^{\outp} \hor_{w} $}
\DisplayProof

&

\AxiomC{$\nsa \hol \phi^{\inp}, \psi^{\inp} \hor_{w}  $}
\RightLabel{$\conl$}
\UnaryInfC{$\nsa \hol \phi \land \psi^{\inp} \hor_{w} $}
\DisplayProof
\end{tabular}
\end{center}

%\medskip

\begin{center}
\begin{tabular}{c c}
\AxiomC{$\nsa \hol \phi^{\outp} \hor_{w} $}
\AxiomC{$\nsa \hol \psi^{\outp} \hor_{w} $}
\RightLabel{$\conr$}
\BinaryInfC{$\nsa \hol \phi \land \psi^{\outp} \hor_{w} $}
\DisplayProof

&

\AxiomC{$\nsa \hol \Gamma, \nest{\phi^{\inp}, \psi^{\outp}}{u} \hor_{w} $}
\RightLabel{$\impr$} %^{\dag_{2}(\X)}$}
\UnaryInfC{$\nsa \hol \Gamma, \phi \imp \psi^{\outp} \hor_{w} $}
\DisplayProof
\end{tabular}
\end{center}

%\medskip

\begin{center}
\begin{tabular}{c}
\AxiomC{$\nsa \hol \phi \imp \psi^{\inp} \hor_{w} \hol \ny, \phi^{\outp} \hor _{u}$}
\AxiomC{$\nsa \hol \phi \imp \psi^{\inp} \hor_{w} \hol \ny, \psi^{\inp} \hor _{u}$}
\RightLabel{$\impl^{\dag}$}
\BinaryInfC{$\nsa \hol \phi \imp \psi^{\inp} \hor_{w} \hol \ny \hor _{u}$}
\DisplayProof
\end{tabular}
\end{center}

%\medskip

\begin{center}
\AxiomC{$\nsa\{[\nsb,[\nsc]_{u}]_{v}\}_{w}$}
\AxiomC{$\nsa\{[\nsc,[\nsb]_{v}]_{u}\}_{w}$}
\RightLabel{$\lin$}
\BinaryInfC{$\nsa\{[\nsb]_{v},[\nsc]_{u}\}_{w}$}
\DisplayProof
\end{center}

\noindent
\textbf{Side conditions:}\\
The side condition $\dag$ stipulates that the rule is applicable only if $\rable{w}{u}$.

\caption{Inference rules for $\nipc$ and $\ngd$.} %The side condition $\dag$ stipulates that the rule is applicable only if $\rable{w}{u}$.
\label{fig:nested-calculi}
\end{figure}

 Proofs (or, derivations) are constructed in $\nipc$ and $\ngd$ in the traditional manner by successively applying inference rules starting from the initial rules. The \emph{height} of a derivation is defined in the usual fashion as the number of nested sequents occurring in a longest branch of a proof from the conclusion to an initial rule. We define the \emph{active} formulae of an inference to be those formulae that the rule operates on, and we define the \emph{active} components to be those components that the rule operates within. We define the \emph{principal} formula of a logical rule to be the (polarized) logical formula displayed in the conclusion. For example, the $\phi^{\outp}$, $\psi^{\outp}$, and $\phi \land \psi^{\outp}$ formulae are active in $\conr$ with $\phi \land \psi^{\outp}$ principal, and the $w$-, $v$-, and $u$-components are active in $\lin$. We consider a formula $\phi$ derivable in $\nipc$ or $\ngd$ \iffi $\phi^{\outp}$ is derivable in $\nipc$ or $\ngd$, respectively.

 A unique feature of $\ngd$ is the incorporation of the $\lin$ rule. This rule operates by (bottom-up) \emph{linearizing} branching structure in a nested sequent and corresponds to the fact that models for \gd \ are connected and linear (as discussed in \sect~\ref{sec:log-prelims-I}). More precisely, if we have a branching structure in our nested sequent as shown below left, then this structure can be linearized in two possible ways as shown below middle (corresponding to the left premise of $\lin$) and as shown below right (corresponding to the right premise of $\lin$). 
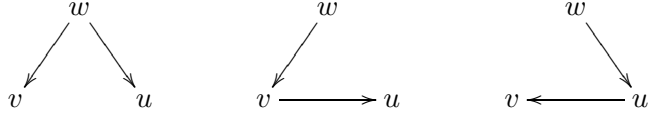
\begin{center}
%\resizebox{\columnwidth}{!}{
\begin{tabular}{c @{\hskip 3em} c @{\hskip 3em} c}
\xymatrix@C=1em{
%\xymatrix{
 &  w\ar@{->}[dl]\ar@{->}[dr]   & \\
 v &     & u 
}

&

\xymatrix@C=1em{
%\xymatrix{
 &  w\ar@{->}[dl]  & \\
 v\ar@{->}[rr] &    & u 
}

&

\xymatrix@C=1em{
%\xymatrix{
 &  w\ar@{->}[dr]  &\\
 v  &   & u\ar@{->}[ll] 
}
\end{tabular}
%}
\end{center}
 This new rule differs from other rules given in the literature for G\"odel-Dummett logics; e.g. the communication rule in hypersequents~\citep{Avr91}, the connected rule in labeled sequents~\citep{DycNeg12}, and the $\rightarrow_{R}^{2}$ rule in the context of linear nested sequents~\citep{KuzLel18}.\footnote{We remark that in certain contexts axioms and their corresponding frame properties may be straightforwardly encoded in logical rules; e.g.~\citet{Bru09,Pog09,GorPosTiu11,Lyo21b}. Nevertheless, it is unclear what (if any) logical rule captures the linearity axiom A9 in our setting.} Also, we note that preliminary algorithmic approaches to transforming axioms into inference rules have been put forth in the context of nested calculi for propositional intermediate logics~\citep{CiaTesStr22}. Nevertheless, such approaches---as currently stated---fail to generate the $\lin$ rule. As shown below, the linearity axiom $(\phi \imp \psi) \lor (\psi \imp \phi)$ can be derived by means of this rule, by using \lem~\ref{lem:general-id}, which is proven below.
 
\begin{center}
\AxiomC{}
\RightLabel{\lem~\ref{lem:general-id}}
\UnaryInfC{$[\phi^{\inp}, \psi^{\outp},[\psi^{\inp}, \phi^{\outp}]_{u}]_{v}$}

\AxiomC{}
\RightLabel{\lem~\ref{lem:general-id}}
\UnaryInfC{$[\psi^{\inp}, \phi^{\outp},[\phi^{\inp}, \psi^{\outp}]_{v}]_{u}$}

\RightLabel{$\lin$}
\BinaryInfC{$[\phi^{\inp}, \psi^{\outp}]_{u},[\psi^{\inp}, \phi^{\outp}]_{v}$}
\RightLabel{$\impr \times 2$}
\UnaryInfC{$(\phi \imp \psi)^{\outp},(\psi \imp \phi)^{\outp}$}
\RightLabel{$\disr$}
\UnaryInfC{$(\phi \imp \psi) \lor (\psi \imp \phi)^{\outp}$}
\DisplayProof
\end{center}

\begin{theorem}[Soundness]\label{thm:soundness} Let $\ns$ be a nested sequent. (1) If $\ns$ is derivable in $\nipc$, then $\ns$ is \ipc-valid; (2) If $\ns$ is derivable in $\ngd$, then $\ns$ is \gd-valid.
\end{theorem}

\begin{proof} Both claims are shown by induction on the height of the given derivation. The first claim follows by \thm~29 and \thm~34 in~\citet[\cptr~5]{Lyo21thesis}, which establishes the soundness of each rule beside $\lin$ in $\ngd$. %Let us therefore confirm the soundness of $\lin$ in $\ngd$. 
 We therefore focus on the second claim, showing the soundness of $\id$ in the base case and the soundness of $\lin$ in the inductive step.
 
\textit{Base case.} Suppose for a contradiction that an instance $\nsa\{p^{\inp}\}_{w}\{p^{\outp}\}_{u}$ of $\id$ is not \gd-valid. Then, there exists a \gd-model $M = (W,\leq,V)$ and a world $o$ such that $M,o \not\Vdash \fint{\nsa\{p^{\inp}\}_{w}\{p^{\outp}\}_{u}}$, implying the existence of a world $o_{1}$ (corresponding to the $w$-component of $\nsa$) and a world $o_{2}$ (corresponding to the $u$-component of $\nsa$) such that $M, o_{1} \Vdash p$ and $M, o_{2} \not\Vdash p$. By the side condition imposed on $\id$, we know that $\rable{w}{u}$, implying that $o_{1} \leq o_{2}$. Hence, by the monotonicity condition (M) (see \dfn~\ref{def:frame-model}), it folows that $M, o_{2} \Vdash p$, which is a contradiction. %, and so, every instance of $\id$ must be \gd-valid.

\textit{Inductive step.} To show the soundness of $\lin$, we assume that the conclusion $\nsa\{[\nsb]_{v},[\nsc]_{u}\}_{w}$ of the rule is not \gd-valid and argue that at least one premise is not \gd-valid. We let $\nsb = \nsb_{1}^{\inp},\nsb_{2}^{\outp},\nsb_{3}$ and $\nsc = \nsc_{1}^{\inp},\nsc_{2}^{\outp},\nsc_{3}$ be the partitions of $\nsb$ and $\nsc$ such that $\nsb_{1}^{\inp}$ and $\nsc_{1}^{\inp}$ contain all input formulae, $\nsb_{2}^{\outp}$ and $\nsc_{2}^{\outp}$ contain all output formulae, and $\nsb_{3}$ and $\nsc_{3}$ contain all nestings at the $v$- and $u$-components, respectively. By our assumption, there exists a \gd-model $M = (W,\leq,V)$ and a world $o$ such that $M,o \not\Vdash \fint{\nsa\{[\nsb]_{v},[\nsc]_{u}\}_{w}}$. In particular, by \dfn~\ref{def:formula-interpretation}, we know that there exist worlds $o_{1}$ and $o_{2}$ such that $o \leq o_{1}$ and $o \leq o_{2}$ in $M$ with $M, o_{1} \Vdash \bigwedge \nsb_{1}$, $M, o_{2} \Vdash \bigwedge \nsc_{2}$, $M, o_{1} \not\Vdash \bigvee \nsb_{2}$, $M, o_{2} \not\Vdash \bigvee \nsc_{2}$. By the connectivity condition imposed on \gd-frames (\dfn~\ref{def:frame-model}), we know that either $o_{1} \leq o_{2}$ or $o_{2} \leq o_{1}$. The first case falsifies the left premise of $\lin$ and the second case falsifies the right premise of $\lin$.
\end{proof}

\begin{lemma}\label{lem:general-id}
Let $\ncalc \in \{\nipc,\ngd\}$. For any formula $\phi \in \langp$, the nested sequent $\nsa\{\phi^{\inp}\}_{w}\{\phi^{\outp}\}_{u}$ is derivable in $\ncalc$ with $\rable{w}{u}$.
\end{lemma}

\begin{proof} We prove the result by induction on the complexity of $\phi$. The base case is trivial as $\ns\{p^{\inp}\}_{w}\{p^{\outp}\}_{u}$ is an instance of $\id$ and $\ns\{\bot^{\inp}\}_{w}\{\bot^{\outp}\}_{u}$ is an instance of $\botl$. We therefore argue the inductive step.

If we suppose that $\phi$ is of the form $\psi \lor \chi$, then the desired nested sequent is derivable as shown below. 
\begin{center}
\AxiomC{$\ns\{\psi^{\inp}\}_{w},\{ \psi^{\outp}, \chi^{\outp}\}_{u}$}
\AxiomC{$\ns\{\chi^{\inp}\}_{w},\{\psi^{\outp}, \chi^{\outp}\}_{u}$}
\RightLabel{$\disl$}
\BinaryInfC{$\ns\{\psi \lor \chi^{\inp}\}_{w},\{ \psi^{\outp}, \chi^{\outp}\}_{u}$}
\RightLabel{$\disr$}
\UnaryInfC{$\ns\{\psi \lor \chi^{\inp}\}_{w},\{ \psi \lor \chi^{\outp}\}_{u}$}
\DisplayProof
\end{center}
 Since the case where $\phi$ is of the form $\psi \land \chi$ is similar to the proof above, we omit it. If $\phi$ is of the form $\psi \imp \chi$, then the desired nested sequent is derivable as shown below:
\begin{center}
\AxiomC{$\ns\{\psi \imp \chi^{\inp}\}_{w}\{ \nest{\emptyset;\psi^{\inp}, \chi^{\inp}, \chi^{\outp}}{v}\}_{u}$}
\AxiomC{$\ns\{\psi \imp \chi^{\inp}\}_{w}\{ \nest{\emptyset;\psi^{\inp}, \psi^{\outp}, \chi^{\outp}}{v}\}_{u}$}
\RightLabel{$\impl$}
\BinaryInfC{$\ns\{\psi \imp \chi^{\inp}\}_{w}\{ \nest{\emptyset; \psi^{\inp}, \chi^{\outp}}{v}\}_{u}$}
\RightLabel{$\impr$}
\UnaryInfC{$\ns\{\psi \imp \chi^{\inp}\}_{w}\{ \psi \imp \chi^{\outp}\}_{u}$}
\DisplayProof
\end{center}
\end{proof}

\subsection{Systems for the First-order Logics}

 In the first-order case, we make use of polarized formulae over the language $\langq$ in our nested sequents, for which the same terminology and notation used in the previous section applies. Let $\nsb$ be a (potentially empty) multiset of these polarized formulae, $\lab$ be a set of labels as before, and let $\va \subset \vars$ be a multiset of variables, referred to as a \emph{signature}; we recursively define a \emph{nested sequent} $\ns$ in the first-order setting as follows:
\begin{itemize}

\item Each object of the form $\va; \nsb$ (referred to more specifically as a \emph{flat sequent}) is a nested sequent, and

\item Any object of the form $\va; \nsb, [\nsc_{1}]_{w_{1}}, \ldots, [\nsc_{n}]_{w_{n}}$, where $\nsc_{i}$ is a nested sequent for $1 \leq i \leq n$, is a nested sequent.

\end{itemize}
 We use the same notation as in the previous section to denote labels and nested sequents, and use $\va, \vb, \vc, \ldots$ to denote signatures occurring in nested sequents. The incorporation of variables into the syntax of our nested sequents assists us in capturing all considered first-order logics within a single formalism, as explained below.  This feature is reminiscent of the hypersequent calculus $\hqlc$ for first-order G\"odel-Dummett logic with non-constant domains~\citep{Tiu11}, which likewise incorporates terms into the syntax of sequents. 
  
 As in the propositional setting, nested sequents encode trees. We define the tree $tr(\nsa)$ of a nested sequent $\nsa$ in the first-order setting as in \dfn~\ref{def:tree-of-nest-seq}, albeit with one difference: in the first-order setting signatures must also be taken into account, and thus, nested sequents are taken to be multisets denoting trees with nodes that are pairs of the form $(w,\va; \nsb)$ such that $w$ is a label, $\va$ is a signature, and $\nsb$ is a multiset of polarized formulae. 
%\begin{definition} Let $\nsa = \va; \nsb, \nest{\nsc_{1}}{w_{1}}, \ldots, \nest{\nsc_{n}}{w_{n}}$ be a nested sequent. We define the \emph{tree of $\nsa$}, denoted $tr(\nsa) = (V,E)$, recursively on the structure of $\nsa$ as follows:
%$$
%V = \{(w_{0},\va; \nsb)\} \cup \bigcup_{1 \leq i \leq n} V_{i}
%\qquad
%E = \{(w_{0},w_{i}) \ | \ 1 \leq i \leq n\} \cup \bigcup_{1 \leq i \leq n} E_{i}
%$$
% where $tr(\nsc_{i}) = (V_{i},E_{i})$ for $1 \leq i \leq n$.
%\end{definition}
 For instance, the nested sequent $\nsa = \va; \nsb, \nest{\nsc_{1}}{w_{1}}, \ldots, \nest{\nsc_{n}}{w_{n}}$ corresponds to the tree $tr(\nsa)$ shown below.
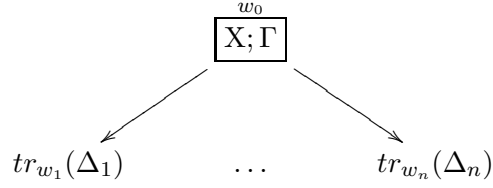
\begin{center}
%\resizebox{\columnwidth}{!}{
\begin{tabular}{c c c}
\xymatrix@C=1em{
%\xymatrix{
 & & \overset{w_{0}}{\boxed{\va; \nsb}}\ar@{->}[dll]\ar@{->}[drr] &  &   		\\
 tr_{w_{1}}(\nsc_{1}) & & \hdots  & & tr_{w_{n}}(\nsc_{n})
}
\end{tabular}
%}
\end{center}
 Similar to the propositional case, we use the notation $\nsa\{\va_{1}; \nsb_{1}\}_{w_{1}}\cdots\{\va_{n}; \nsb_{n}\}_{w_{n}}$ to denote a nested sequent $\nsa$ such that the data $\va_{1}; \nsb_{1}, \ldots, \va_{n}; \nsb_{n}$ is rooted at $w_{1}, \ldots, w_{n}$, respectively, in $tr(\nsa)$. To simplify notation in certain cases, we sometimes disregard the presentation of signatures; for example, if $\nsa = x; p^{\inp},[x,y; q^{\inp},[\emptyset; p \imp q^{\outp}]_{u}]_{w}$, then $\nsa\{x; p^{\inp}\}_{w_{0}}$, $\nsa\{p^{\inp}\}_{w_{0}}$, $\nsa\{q^{\inp},[\emptyset; p \imp q^{\outp}]_{u}\}_{w}$, and $\nsa\{x,y; q^{\inp}\}_{w}\{p \imp q^{\outp}\}_{v}$ are all correct representations of $\nsa$ in our notation. Likewise, we define the relations $\rable{}{}$ and $\rtable{}{}$ over first-order nested sequents in the same manner as in \dfn~\ref{def:reachability-relation}. While in the propositional setting the reachability relation $\rable{}{}$ is sufficient to formalize our propagation and reachability rules, in the first-order setting we must also make use of the notion of \emph{availability} (cf.~\citet{Fit14,Lyo21a,Lyo21thesis}), a tool utilized to formalize the quantifier rules in our first-order nested calculi.
 
\begin{definition}[Available]\label{def:available} We say that a variable $x$ is \emph{available} for a label $w$ in a nested sequent $\ns$ \iffi there exists a label $u$ in $\ns$ such that (1) $\rable{u}{w}$ and (2) $x$ occurs in the signature $\va$ in the $u$-component of $\ns$.
\end{definition}

 To provide further intuition concerning the notion of availability, we give an example, and let $\ns = y; p(y)^{\inp}, \lb x; \exists y \forall x p(x,y)^{\outp}, \lb z; r(y) \imp p(z)^{\outp} \rb_{w_{2}} \rb_{w_{1}}$. Then, we have that $y$ is available for $w_{0}$, $w_{1}$, and $w_{2}$ (recall that the root, which is $y; p(y)$ in this example, is always associated with the fixed label $w_{0}$), while $x$ is only available for $w_{1}$ and $w_{2}$, and $z$ is only available for $w_{2}$.

%\begin{center}
%\resizebox{\columnwidth}{!}{
%\begin{tabular}{c c c}
%\xymatrix@C=1em{
%\xymatrix{
% & & \overset{w_{0}}{\boxed{y; p(y)}}\ar@{->}[dll]\ar@{->}[drr] &  &   		\\
% \overset{w_{1}}{\boxed{\emptyset; \exists y \forall x p(x,y)}} & &  & & \overset{w_{2}}{\boxed{x,z; r(y) \imp p(z)}}
%}
%\end{tabular}
%}
%\end{center}
 
 The nested calculi $\nnd$, $\ncd$, $\nndl$, and $\ncdl$ for first-order intuitionistic and G\"odel-Dummett logics are defined in \dfn~\ref{def:nns-ncd} below, and extend the nested calculi $\nipc$ and $\ngd$ by employing first-order nested sequents in rules while including the first-order rules shown in \fig~\ref{fig:nested-calculi-fo}.\footnote{We note that $\nnd$ and $\ncd$ differ from the nested sequent systems given in~\citet{Fit14} and~\citet{Lyo21thesis} for first-order intuitionistic logics. Specifically, $\nnd$ and $\ncd$ employ signatures in nested sequents with reachability rules that rely on such. The use of signatures is helpful in extracting counter-models from failed proof-search, which is crucial for the proof of cut-free completeness (see \thm~\ref{thm:cut-free-comp} below).} The rules in \fig~\ref{fig:nested-calculi-fo} are subject to a variety of side conditions which rely on the reachability relation $\rable{}{}$, the notion of availability, and the notion of \emph{freshness}. In particular, we say that a variable $x$ (or, label $w$) is \emph{fresh} in a rule \iffi it does not occur in the conclusion of the rule. When we refer to a variable or label as \emph{fresh} we mean that it is fresh in the context of the rule where it appears. %; thus, a variable (or, label) is \emph{not} fresh \iffi it occurs in the conclusion of a rule. 
 
%\footnote{We note that $\nnd$ and $\ncd$ differ from the nested sequent systems given in~\citet{Fit14} and~\citet{Lyo21thesis} for first-order intuitionistic logics. Specifically, $\nnd$ and $\ncd$ employ signatures in nested sequents with reachability rules that rely on such. As will be discussed in \rmk~\ref{rmk:invert-fititng} of \sect~\ref{sec:properties}, this leads to $\nnd$ and $\ncd$ possessing a wider array of favorable proof-theoretic properties.}
 
\begin{figure}[t]
%\noindent\hrule

\begin{center}
\begin{tabular}{c c}
\AxiomC{}
\RightLabel{$\idfo^{\dag_{1}(\X)}$}
\UnaryInfC{$\nsa \hol p(\vec{x})^{\inp} \hor_{w} \hol p(\vec{x})^{\outp} \hor_{u}$}
\DisplayProof

&

\AxiomC{$\nsa \hol \va, \vec{x}; p(\vec{x})^{\inp} \hor_{w}$}
\RightLabel{$\doms$}
\UnaryInfC{$\nsa \hol \va; p(\vec{x})^{\inp} \hor_{w}$}
\DisplayProof
\end{tabular}
\end{center}

%\medskip

\begin{center}
\begin{tabular}{c c}
\AxiomC{$\nsa \hol  \phi(y/x)^{\outp}, \exists x \phi^{\outp}  \hor _{w}$}
\RightLabel{$\existsri^{\dag_{3}(\X)}$}
\UnaryInfC{$\nsa \hol \exists x \phi^{\outp} \hor _{w}$}
\DisplayProof

&

\AxiomC{$\nsa \hol \va, y; \nsb \hor_{w} \hol \phi(y/x)^{\outp}, \exists x \phi^{\outp}  \hor _{u}$}
\RightLabel{$\existsrii^{\dag_{2}(\X)}$}
\UnaryInfC{$\nsa \hol \va; \nsb \hor_{w} \hol \exists x \phi^{\outp}  \hor _{u}$}
\DisplayProof
\end{tabular}
\end{center}

%\medskip

\begin{center}
\begin{tabular}{c c}
\AxiomC{$\nsa \hol  \va, y;  \phi(y/x)^{\inp}  \hor_{w} $}
\RightLabel{$\existsl^{\dag_{4}(\X)}$}
\UnaryInfC{$\nsa \hol  \va; \exists x \phi^{\inp}\hor_{w} $}
\DisplayProof

&

\AxiomC{$\nsa \hol \nsb, \nest{y ; \phi(y/x)^{\outp}}{u} \hor_{w} $}
\RightLabel{$\allr^{\dag_{4}(\X)}$}
\UnaryInfC{$\nsa \hol \nsb, \forall x \phi^{\outp} \hor_{w} $}
\DisplayProof
\end{tabular}
\end{center}

%\medskip

\begin{center}
\begin{tabular}{c c}
\AxiomC{$\nsa \hol \forall x \phi^{\inp} \hor_{w} \hol  \nsb, \phi(y/x)^{\inp} \hor_{u}$}
\RightLabel{$\allli^{\dag_{5}(\X)}$}
\UnaryInfC{$\nsa \hol \forall x \phi^{\inp} \hor_{w} \hol  \nsb \hor_{u}$}
\DisplayProof

&

\AxiomC{$\nsa \hol \va,y; \nsb \hor_{w} \hol \forall x \phi^{\inp} \hor_{u} \hol \nsc, \phi(y/x)^{\inp} \hor_{v}$}
\RightLabel{$\alllii^{\dag_{6}(\X)}$}
\UnaryInfC{$\nsa \hol \va; \nsb \hor_{w} \hol \forall x \phi^{\inp} \hor_{u} \hol \nsc \hor_{v}$}
\DisplayProof
\end{tabular}
\end{center}

\noindent
\textbf{\nd \ side conditions}:
\begin{center}
\begin{minipage}{.425\textwidth}
\begin{itemize}

\item[] $\dag_{1}(\nd) :=$ $\rable{w}{u}$

\item[] $\dag_{2}(\nd) := $ $\rable{w}{u}$ \& $y$ is fresh

\item[] $\dag_{3}(\nd) := $ $y$ is available for $w$ 

\end{itemize}
\end{minipage}
\begin{minipage}{.55\textwidth}
\begin{itemize}

\item[] $\dag_{4}(\nd) := $ $y$ is fresh 

\item[] $\dag_{5}(\nd) := $ $\rable{w}{u}$ \& $y$ is available for $u$ 

\item[] $\dag_{6}(\nd) := $ $\rable{w}{u}\rable{}{v}$ \& $y$ is fresh

\end{itemize}
\end{minipage}
\end{center}

\noindent
\textbf{\cd \ side conditions}:
\begin{center}
\begin{minipage}{.425\textwidth}
\begin{itemize}

\item[] $\dag_{1}(\cd) :=$ $\rable{w}{u}$

\item[] $\dag_{2}(\cd) := $ $y$ is fresh

\item[] $\dag_{3}(\cd) := $ $y \in \vars$

\end{itemize}
\end{minipage}
\begin{minipage}{.55\textwidth}
\begin{itemize}

\item[] $\dag_{4}(\cd) := $ $y$ is fresh

\item[] $\dag_{5}(\cd) := $ $\rable{w}{u}$ \& $y \in \vars$

\item[] $\dag_{6}(\cd) := $ $\rable{u}{v}$ \& $y$ is fresh

\end{itemize}
\end{minipage}
\end{center}

%\hrule
\caption{First-order inference rules. When $\X$ is $\nd$, the rule is subject to the associated ND side condition, and when $\X$ is $\cd$, the rule is subject to the associated CD side condition. The first-order calculi that employ such rules are specified in \dfn~\ref{def:nns-ncd}.}
\label{fig:nested-calculi-fo}
\end{figure}

\begin{definition}[$\nnd$, $\ncd$, $\nndl$, $\ncdl$]\label{def:nns-ncd} We define both $\nnd$ and $\ncd$ to be the set consisting of the rules $\idfo$, $\botl$, $\doms$, $\disl$, $\disr$, $\conl$, $\conr$, $\impl$, $\impr$, $\existsl$, $\existsri$, $\existsrii$, $\allli$, $\alllii$, and $\allr$, but where the first-order rules in $\nnd$ are subject to the \emph{\nd \ side conditions}, and where the first-order rules in $\ncd$ are subject to the \emph{\cd \ side conditions} (see \fig~\ref{fig:nested-calculi-fo} for a description of the \nd \ and \cd \ side conditions). We define $\nndl = \nnd \cup \{\lin\}$ and $\ncdl = \ncd \cup \{\lin\}$.
\end{definition}

\begin{remark}\label{rmk:fo-over-prop}
 The $\id$ rule is a special case of the $\idfo$ rule where the principal formulae are propositional atoms. Thus, $\nnd$ and $\ncd$ can be seen as properly extending $\nipc$, and $\nndl$ and $\ncdl$ can be see as properly extending $\ngd$. We may therefore view $\nipc$ as a fragment of both $\nnd$ and $\ncd$, and $\ngd$ as a fragment of $\nndl$ and $\ncdl$, where the signature of each nested sequent is set to $\emptyset$.
\end{remark}

 The \emph{initial rules} of our first-order nested calculi are $\idfo$ and $\botl$, while $\doms$ and $\lin$ are the only \emph{structural rules} (with the latter rule occurring only in $\nndl$ and $\ncdl$), and all remaining rules are \emph{logical rules} as they bottom-up decompose complex logical formulae. The \emph{domain shift rule} $\doms$ encodes the semantic property that $V(p,w) \subseteq D(w)^{n}$ in any model, while the $\existsrii$ and $\alllii$ rules encode the fact that domains are non-empty as both rules bottom-up introduce fresh variables to signatures. Such rules are required for completeness in the first-order setting as shown in \app~\ref{app:completeness}. The notion of a \emph{proof} (or, \emph{derivation}), its \emph{height}, and \emph{active} and \emph{principal formulae/components} are defined as in the previous section. We consider a formula $\phi$ derivable in a first-order nested calculus \iffi $\va; \phi^{\outp}$ is derivable, where $\va$ is the set of free variables occurring in $\phi$.
 
 A distinctive feature of our first-order nested calculi is the inclusion of \emph{(first-order) reachability rules}. This class of rules (introduced in~\citet{Lyo21thesis}) serves as a generalization of the well-known class of \emph{propagation rules} (cf.~\citet{CasCerGasHer97,Fit72,GorPosTiu11}). Whereas propagation rules may propagate formulae throughout (the tree of) a nested sequent when applied bottom-up, reachability rules can additionally check to see if data exists along certain paths within (the tree of) a nested sequent. The rules $\idfo$, $\impl$, $\existsri$, $\existsrii$, $\allli$, and $\alllii$ serve as our reachability rules, though it should be noted that $\impl$ qualifies as a \emph{proper} propagation rule as it omits any search for data. To demonstrate the functionality of such rules, we provide a proof of the constant domain axiom (A14 in \dfn~\ref{def:axioms-fo}) in $\ncd$.

 We let $\nsa_{1}$ and $\nsa_{2}$ be the following two nested sequents, both of which are derivable by \lem~\ref{lem:general-id-fo} (proven below) since $\phi(y)^{\inp}$ and $\phi(y)^{\outp}$ occur in $\nsa_{1}$ with $\rable{w}{u}$, and $\psi^{\inp}$ and $\psi^{\outp}$ occur in $\nsa_{2}$.
\begin{itemize}
\item $\nsa_{1} = \emptyset; [\emptyset; \forall x (\phi(x) \lor \psi)^{\inp}, \phi(y)^{\inp}, \psi^{\outp}, [y; \phi(y)^{\outp}]_{u}]_{w}$

\item $\nsa_{2} = \emptyset; [\emptyset; \forall x (\phi(x) \lor \psi)^{\inp}, \psi^{\inp}, \psi^{\outp}, [y; \phi(y)^{\outp}]_{u}]_{w}$
\end{itemize}
 The derivability of the above nested sequents implies the derivability of the constant domain axiom, as shown below.\\

\begin{center}
\AxiomC{$\nsa_{1}$}

\AxiomC{$\nsa_{2}$}

\RightLabel{$\disl$}
\BinaryInfC{$\emptyset; [\emptyset; \forall x (\phi(x) \lor \psi)^{\inp}, \phi(y) \lor \psi^{\inp}, \psi^{\outp}, [y; \phi(y)^{\outp}]_{u}]_{w}$}
\RightLabel{$\allli$}
\UnaryInfC{$\emptyset; [\emptyset; \forall x (\phi(x) \lor \psi)^{\inp}, \psi^{\outp}, [y; \phi(y)^{\outp}]_{u}]_{w}$}
\RightLabel{$\allr$}
\UnaryInfC{$\emptyset; [\emptyset; \forall x (\phi(x) \lor \psi)^{\inp}, \forall x \phi(x)^{\outp}, \psi^{\outp}]_{w}$}
\RightLabel{$\disr$}
\UnaryInfC{$\emptyset; [\emptyset; \forall x (\phi(x) \lor \psi)^{\inp}, \forall x \phi(x) \lor \psi^{\outp}]_{w}$}
\RightLabel{$\impr$}
\UnaryInfC{$\emptyset; \forall x (\phi(x) \lor \psi) \imp \forall x \phi(x) \lor \psi^{\outp}$}
\DisplayProof
\end{center}

%\medskip

 We observe that the constant domain axiom is \emph{not} derivable in $\nnd$ as the side condition $\dag_{5}(\nd)$ imposed on $\allli$ is not satisfied above, i.e. $y$ is not available for $w$ in the conclusion of $\allli$. However, $\allli$ is applicable in the above proof with the calculus $\ncd$ as $\dag_{5}(\cd)$ is satisfied. %since $y$ is not fresh in the conclusion of $\allli$.
 
 As mentioned previously, alternative structural proof systems have been introduced for first-order G\"odel-Dummett logics, most notably, the hypersequent calculus $\hif$ introduced by \citet{BaaZac00} for $\cdl$, and the hypersequent calculus $\hqlc$ introduced by \citet{Tiu11} for $\ndl$.\footnote{The logic $\cdl$ is referred to as \emph{intuitionistic fuzzy logic} ($\mathsf{IF}$) in \citet{BaaZac00}, and $\ndl$ is referred to as \emph{G\"odel-Dummett logic} ($\mathsf{QLC}$) in \citet{Tiu11}.} There are a few notable differences between the nested calculi presented here and the aforementioned hypersequent systems. First, both $\hif$ and $\hqlc$ employ hypersequents of the shape $\Phi_{1} \vdash \Psi_{1} \ | \ \cdots \ | \ \Phi_{n} \vdash \Psi_{n}$ such that each $\Phi_{i}$ is a finite multiset of formulae and/or terms and each $\Psi_{i}$ is either empty or contains a single formula, i.e. the sequents $\Phi_{i} \vdash \Psi_{i}$ occurring within the hypersequents of $\hif$ and $\hqlc$ are \emph{single-conclusioned}. This is distinct from the nested sequents employed in our nested calculi, which are \emph{multi-conclusioned}, that is, they allow for multiple output formulae to occur within a component. (NB. As discussed at the beginning of \sect~\ref{subsec:prop-nested-systems}, output formulae are analogous to the consequent of a traditional, Gentzen-style sequent.) This provides our nested calculi with a distinct advantage over their hypersequent counterparts as the multi-conclusioned nature of sequents necessitates the invertibility of all logical rules (as detailed in \sect~\ref{sec:properties} below), having the effect that counter-models may be readily extracted from failed proof-search (as detailed in \app~\ref{app:completeness}). The invertibility of all logical rules is a property which fails to hold for both $\hif$ and $\hqlc$. %Second, unlike $\hif$, the nested calculus $\ncdl$ is free of structural rules for weakenings and contractions as such rules are admissible in the system (\sect~\ref{sec:properties}). 
 Furthermore, despite that fact that both $\hif$ and $\hqlc$ are hypersequent systems, $\hqlc$ employs a much richer syntax than $\hif$ and the intuitionistic nature of such systems is obtained from the use of single-conclusioned sequents rather than explicitly encoding intuitionistic model-theoretic properties (e.g. monotonicity of atomic formulae; see \dfn~\ref{def:frame-model}) into the functionality of rules. Consequently, our nested systems for G\"odel-Dummett logics enjoy a higher degree of uniformity and modularity as we may pass from a nested system for one logic to another by simply modifying the side conditions imposed on reachability rules. In fact, it is the use of reachability rules that permits us to capture the nested domain condition (ND) for \nd \ and \ndl, and the constant domain condition (CD) for \cd \ and \cdl \ within a single formalism (these conditions are given in \dfn~\ref{def:frame-model-fo}). %By modifying the side conditions of such rules, we may pass from a nested calculus for one logic to a nested calculus for another. Furthermore, 
 We could even modify the side conditions of our reachability rules to obtain logics beyond \nd, \cd, \ndl, and \cdl \ (see \citet{Lyo21thesis} for a discussion). %As an example, the rules $\existsrii$ and $\alllii$ capture the fact that domains are non-empty in frames, and thus, if we disregard the freshness constraint on $y$ in the side conditions $\dag_{2}(\X)$ and $\dag_{6}(\X)$ for $\X \in \{\nd,\cd\}$, then our nested calculus would capture variants of \nd \ and \cd \ whereby domains are permitted to be empty (see \cite{Lyo21a} for a discussion).
 
 We note that our notion of availability imposed on quantifier rules is strongly connected to Fitting's notion thereof~\citep{Fit14}. Fitting introduced nested calculi for first-order intuitionistic logics employing variants of our existential and universal quantifier rules, which differ from our formulation in at least two significant ways: (1) Fitting omits the use of signatures in nested sequents and (2) a term $y$ is defined to be \emph{available} for a $w$-component in a nested sequent \iffi there exists a $u$-component containing a polarized formula $\phi(y)^{\io}$ such that $\rable{u}{w}$.\footnote{To ease our presentation, we have modified Fitting's formulation of available terms to fit within our nomenclature. See \citet[\dfn~9.1]{Fit14} for the original formulation.} Thus, a term is available in Fitting's framework when it occurs within a \emph{formula} as opposed to a \emph{signature}. If we were to formulate our availability condition and associated quantifier rules in such a manner, this would impede our proof of cut-free completeness (see \thm~\ref{thm:cut-free-comp} below) as we use signatures to extract counter-models from failed proof-search (see \app~\ref{app:completeness}) and it is not clear how this extraction is to take place otherwise. %in the absence of signatures. 
 
 %If we were to formulate our availability condition and associated quantifier rules in such a manner, this would bring about significant disadvantages. We discuss this point in more detail in \rmk~\ref{rmk:invert-fititng} of \sect~\ref{sec:properties}.

 We now establish the soundness of our calculi by interpreting nested sequents over \logic-models, for $\logic \in \{\nd,\cd,\ndl,\cdl\}$. We could provide a formula interpretation of our first-order nested sequents by employing the  \emph{existence predicate} as in \citet{Tiu11}, however, this would require us to consider and define extensions of our first-order logics with such a predicate. We therefore opt to interpret nested sequents directly over models.

\begin{definition}[Nested Semantics]\label{def:sequent-semantics} Let $\logic \in \{\nd,\cd,\ndl,\cdl\}$, $M = (W,\leq,D,V)$ be an \logic-model, and $\ns$ be a nested sequent. We define an \emph{$M$-interpretation} to be a function $\iota$ mapping every label $w \in \lab$ to a world $\iota(w) \in W$. % and (2) every variable $x \in \vars$ to an element $\iota(x) \in D(W)$. 
 Let $\nsb = \nsc^{\inp},\nsd^{\outp}$ be a multiset of polarized formulae and $\mu$ be an $M$-assignment. We recursively define the satisfaction of a nested sequent $\ns$ with $\iota$ and $\mu$, %at a label $w$ on an \nd- or \cd-model $M$ with an interpretation $\iota$ (
 written $M,\iota,\mu,w_{0} \models \ns$, accordingly:
\begin{itemize}

\item if $\nsa = \va; \emptyset$, then $M,\iota,\mu,u \not\models \ns$;

%\item if $\nsa = \va; \phi(x_{1}, \ldots, x_{n})^{\io}$, then $M,\iota,u \models \phi(x_{1}, \ldots, x_{n})^{\io}$ \iffi there exists an assignment $\mu_{\iota(u)}$ such that (1) for every $x \in \va$, $\iota(x) = \mu_{\iota(u)}(x)$ and (2) $M, \iota(u), \mu_{\iota(u)} \Vdash \phi(x_{1}, \ldots, x_{n})$.

\item if $\nsa = \va; \nsb$, then $M,\iota,\mu,u \models \nsa$ \iffi %there exists an assignment $\mu$ such that (1) $\iota(x) = \mu_{\iota(u)}(x)$ for every $x \in \va$, (2) 
 if $\mu(x) \in D(\iota(u))$ for all $x \in \va$ and $M, \iota(u), \mu \Vdash \phi$ for every $\phi^{\inp} \in \nsc^{\inp}$, then $M, \iota(u), \mu \Vdash \psi$ for some $\psi^{\outp} \in \nsd^{\outp}$;

\item if $\nsa = \va; \nsb, [\nsc_{1}]_{u_{1}}, \ldots, [\nsc_{n}]_{u_{n}}$, then $M,\iota,\mu,u \models \ns$ \iffi %there exists an assignment $\mu_{\iota(u)}$ such that  (1) for every $x \in \va$, $\iota(x) = \mu_{\iota(u)}(x)$, (2) for every $i \in \{1,\ldots,n\}$, $\iota(u) \leq \iota(u_{i})$, and (3) 
 if for every $1 \leq i \leq n$, $\iota(u) \leq \iota(u_{i})$, then either $M, \iota,\mu,u \models \va; \nsb$ or $M, \iota,\mu,u_{i} \models \nsc_{i}$.
\end{itemize}
 We write $M,\iota,\mu,w_{0} \not\models \ns$ when a nested sequent $\ns$ is \emph{not satisfied} on $M$ with $\iota$ and $\mu$. We say that a nested sequent $\ns$ is \logic-valid \iffi for every \logic-model $M$, every $M$-interpretation $\iota$, and every $M$-assignment $\mu$, we have $M,\iota,\mu,w_{0} \models \ns$, and we say that $\ns$ is \logic-invalid otherwise.
\end{definition}

\begin{theorem}[Soundness]\label{thm:soundness} Let $\logic \in \{\nd,\cd,\ndl,\cdl\}$ and $\ns$ be a nested sequent. If $\ns$ is derivable in $\mathsf{N}_{\logic}$, then $\ns$ is \logic-valid.
\end{theorem}

\begin{proof} We argue the claim by induction on the height of the given derivation for $\ndl$. The remaining claims are similar.

\textit{Base case.} It is straightforward to show that any instance of $\botl$ is \ndl-valid. Let us therefore argue that any instance of $\id$ is \ndl-valid. We consider an arbitrary instance of $\id$, where $\rable{w}{u}$ holds by the side condition.
\begin{center}
\AxiomC{}
\RightLabel{$\id$}
\UnaryInfC{$\ns \hol p(\vec{x})^{\inp} \hor_{w} \hol p(\vec{x})^{\outp} \hor_{u}$}
\DisplayProof
\end{center}
 Suppose that $\ns \hol p(\vec{x})^{\inp} \hor_{w} \hol p(\vec{x})^{\outp} \hor_{u}$ is \ndl-invalid. Then, there exists an \ndl-model $M = (W, \leq, D, V)$, $M$-interpretation $\iota$, and $M$-assignment $\mu$ such that $\iota(w) \leq \iota(v_{1}), \ldots, \iota(v_{n}) \leq \iota(u)$ with $v_{1}, \ldots, v_{n}$ the labels along the path from $w$ to $u$ in $\ns$. By \dfn~\ref{def:sequent-semantics}, we know that $M,\iota(w),\mu \Vdash p(\vec{x})$ and $M, \iota(u), \mu \not\Vdash p(\vec{x})$. However, this produces a contradiction  by \prp~\ref{prop:monotonicity-fo} as $M, \iota(u), \mu \Vdash p(\vec{x})$ must hold.

\textit{Inductive step.} We prove the inductive step by contraposition, showing that if the conclusion of the rule is \ndl-invalid, then at least one premise is \ndl-invalid. We consider the $\impl$, $\existsri$, $\alllii$, $\allr$, and $\lin$ cases as the remaining cases are simple or similar.

$\impl$. Assume that $\ns \{\phi \imp \psi^{\inp}\}_{w}\{\Gamma\}_{u}$ is \ndl-invalid and that the side condition $\rable{w}{u}$ holds. By our assumption, there exists an \ndl-model $M$, $M$-interpretation $\iota$, and $M$-assignment $\mu$ such that (1) $\iota(w) \leq \iota(v_{1}), \ldots, \iota(v_{n}) \leq \iota(u)$ with $v_{1}, \ldots, v_{n}$ the labels along the path from $w$ to $u$ and (2) $M, \iota(w), \mu \Vdash \phi \imp \psi$. Since $M, \iota(w), \mu \Vdash \phi \imp \psi$, we know that either $M, \iota(u), \mu \not\Vdash \phi$ or $M,\iota(u), \mu \Vdash \psi$. In the first case, the left premise of $\impl$ is invalid, and in the second case, the right premise of $\impl$ is invalid.

$\existsri$. Suppose that $\ns \{\exists x \phi^{\outp}\}_{w}$ is \ndl-invalid. Then, there exists an \ndl-model $M$, $M$-interpretation $\iota$, and $M$-assignment $\mu$ such that $M, \iota(w), \mu \not\Vdash \exists x \phi$. By the side condition on $\existsri$, $y$ is available for $w$, meaning there exists a path $\iota(u) \leq \iota(v_{1}), \ldots, \iota(v_{n}) \leq \iota(w)$ in $M$ such that $\mu(y) \in D(\iota(u))$ by \dfn~\ref{def:sequent-semantics}. By the (ND) condition, we know that $\mu(y) \in D(\iota(w))$, showing that $M, \iota(w), \mu \not\Vdash \phi(y/x)$, which proves the premise of $\existsri$ \ndl-invalid.

$\alllii$. Suppose that $\nsa \hol \va; \nsb \hor_{w} \hol \forall x \phi^{\inp} \hor_{u} \hol \nsc \hor_{v}$ is \ndl-invalid. Then, there exists an \ndl-model $M$, $M$-interpretation $\iota$, and $M$-assignment $\mu$ such that $M, \iota(u), \mu \Vdash \forall x \phi$. By the side condition on $\alllii$, we know there exists a path $\iota(w) \leq \iota(u_{1}), \ldots, \iota(u_{n}) \leq \iota(u)$ and a path $\iota(u) \leq \iota(v_{1}), \ldots, \iota(v_{k}) \leq \iota(v)$ in $M$. Moreover, by the fact that domains in \ndl-models are non-empty, we know there exists an element $d \in D(\iota(w))$. Hence, by the (ND) condition $d \in D(\iota(v))$. Therefore, $M, \iota(v), \mu[d/y] \Vdash \phi(y/x)$, showing that the premise of $\alllii$ is \ndl-invalid as well.

$\allr$. Let us assume that $\nsa \hol \nsb, \forall x \phi^{\outp} \hor_{w} $ is \ndl-invalid. Then, there exists a model $M$, $M$-interpretation $\iota$, and $M$-assignment $\mu$ such that $M, \iota(w), \mu \not\Vdash \forall x \phi$. Thus, there exists a world $u \in W$ such that $\iota(w) \leq u$, $d \in D(u)$, and $M, u, \mu[d/y] \not\Vdash \phi(y/x)$. Let $\iota'(v) = \iota(v)$ if $v \neq u$ and $\iota'(u) = u$ otherwise, and $\mu'(z) = \mu(z)$ if $z \neq y$ and $\mu'(y) = d$ otherwise. Then, the premise of $\allr$ is not satisfied on $M$ with $\iota'$ and $\mu'$, showing it \ndl-invalid.

$\lin$. Let us assume that $\nsa\{[\nsb]_{v},[\nsc]_{u}\}_{w}$ is \ndl-invalid. Then, there exists an \ndl-model $M$, $M$-interpretation $\iota$, and $M$-assignment $\mu$ such that $M,\iota,\mu,v \not\models \nsb$, $M,\iota,\mu,u \not\models \nsc$, $\iota(w) \leq \iota(v)$, and $\iota(w) \leq \iota(u)$. By the connectivity property (see \dfn~\ref{def:frame-model}), we know that either $\iota(v) \leq \iota(u)$ or $\iota(u) \leq \iota(v)$. The first case proves the left premise of $\lin$ \ndl-invalid, and the second case proves the right premise of $\lin$ \ndl-invalid.
\end{proof}

\begin{lemma}\label{lem:general-id-fo}
Let $\ncalc \in \{\nnd,\ncd,\nndl,\ncdl\}$. For any formula $\phi \in \langq$, the nested sequent $\nsa\{\phi^{\inp}\}_{w}\{\phi^{\outp}\}_{u}$ is derivable in $\ncalc$ with $\rable{w}{u}$.
\end{lemma}

\begin{proof} The lemma extends the proof of \lem~\ref{lem:general-id}, being shown by induction on the complexity of $\phi$. We show the case where $\phi$ is of the form $\forall x \psi$.
\begin{center}
\AxiomC{$\ns\{\forall \psi^{\inp}\}_{w}\{ \nest{y; \psi(\parama/x)^{\inp}, \psi(\parama/x)^{\outp}}{v}\}_{u}$}
\RightLabel{$\allli$}
\UnaryInfC{$\ns\{\forall \psi^{\inp}\}_{w}\{\nest{y; \psi(\parama/x)^{\outp}}{v}\}_{u}$}
\RightLabel{$\allr$}
\UnaryInfC{$\ns\{\forall \psi^{\inp}\}_{w}\{ \forall \psi^{\outp}\}_{u}$}
\DisplayProof
\end{center}
 Observe that $y$ is available for $v$ in the $\allli$ rule above since $y$ occurs in the $v$-component, showing that the inference is indeed valid.
\end{proof}

 All of our nested calculi can be shown to be cut-free complete. %, i.e. every \logic-valid nested sequent is derivable in $\mathsf{N}_{\logic}$ with $\mathrm{L} \in \{\ipc, \nd, \cd, \gd, \ndl, \cdl\}$. 
 This result is proven by providing a (potentially non-terminating) proof-search algorithm $\prove$ for each nested system. If $\prove$ terminates, then the input nested sequent has a proof, and if $\prove$ does not terminate, then we show that a counter-model can be constructed witnessing the invalidity of the input nested sequent. In the $\gd$, $\ndl$, and $\cdl$ cases, if $\prove$ does not terminate, then the $\lin$ rule plays a crucial role in the extraction of a counter-model. As the $\lin$ rule bottom-up linearizes branching structure in a nested sequent, the rule effectively imposes a linear order on the components of a nested sequent, ultimately yielding a counter-model with a linear accessibility relation in the case of failed proof-search. As the details of this proof are lengthy and tedious, we defer the proof of cut-free completeness to the appendix (\app~\ref{app:completeness}).

\begin{theorem}[Cut-free Completeness]\label{thm:cut-free-comp}
Let $\mathrm{L} \in \{\ipc, \nd, \cd, \gd, \ndl, \cdl\}$ and $\vec{x}; \phi(\vec{x})$ be a nested sequent with $\vec{x}$ all free variables in $\phi(\vec{x})$. If a nested sequent $\vec{x}; \phi(\vec{x})$ is $\mathrm{L}$-valid, then it is derivable in $\ncalc_{\mathrm{L}}$.
\end{theorem}

\begin{remark} We note that all axioms and inference rules for \nd \ and \cd \ are respectively derivable in $\nnd$ and $\ncd$ \emph{without the $\doms$ rule}, but with the $\cut$ rule (shown in \fig~\ref{fig:lab-struc-rules}). Therefore, due to the cut-elimination theorem (\thm~\ref{thm:cut-elim-int}) in \sect~\ref{sec:cut-elim}, the nested systems $\nnd$ and $\ncd$ are complete relative to \nd \ and \cd, respectively, without the $\doms$ rule, showing that the rule is admissible in these systems.
 %That is, if $\vdash_{\nd} \phi$ or $\vdash_{\cd} \phi$, then $\va; \phi$ is derivable in $\nnd$ or $\ncd$, respectively, where $\va$ is the set of free variables in $\phi$. Nevertheless, the $\doms$ rule is used to establish \emph{strong completeness}, i.e. the cut-free completeness theorem (\thm~\ref{thm:cut-free-comp}) above, as $\doms$ is used to extract a counter-model from failed proof-search.
\end{remark}

%Invertiblity and admissibility properties
\section{Invertibility and Admissibility Properties}\label{sec:properties}

 We now prove that our nested calculi satisfy a broad range of height-preserving admissibility and invertibility properties. We define a rule $(r)$ to be (height-preserving) admissible in a nested calculus $\ncalc$ \iffi if the premises have proofs (of heights $h_{1},\ldots,h_{n}$), then the conclusion has a proof (of height $h \leq \max\{h_{1}, \ldots, h_{n}\}$). A rule $(r)$ is defined to be (height-preserving) invertible \iffi if the conclusion has a proof (of height $h$), then the premises have a proof (of height $h$ or less). We refer to height-preserving admissible rules as \emph{hp-admissible} rules and height-preserving invertible rules as \emph{hp-invertible} rules.
 
 As will be shown, all rules of $\nipc$, $\nnd$, and $\ncd$ (i.e. all logical rules in $\ngd$, $\nndl$, and $\ncdl$) are hp-invertible. The various (hp-)admissible rules are displayed in \fig~\ref{fig:lab-struc-rules}. The $\ndr$, $\ddr$, $\lwr$, and $\lft$ rules are subject to the side condition that $\rtable{w}{u}$ and the rules $\lsub$, $\nec$, and $\ex$ are subject to the side condition that the label $u$ is fresh (recall that at the beginning of \sect~\ref{sec:nested-calculi} we stipulated that all labels occurring in nested sequents must be unique). The $\psub$ rule substitutes a variable $y$ for every occurrence of a free variable $x$ in a nested sequent $\ns$, possibly renaming bound variables to avoid unwanted variable capture; for example, if $\ns$ is the nested sequent $x; p(x)^{\inp},[\emptyset; q(z) \imp r(x)^{\outp}]_{w}$, then $\ns(y/x) =y; p(y)^{\inp},[\emptyset; q(z) \imp r(y)^{\outp}]_{w}$. We remark that although the (hp-)admissibility of these rules is interesting in its own right, such rules serve a practical role, being used in our proof of syntactic cut-elimination for $\nipc$, $\nnd$, and $\ncd$ in the following section.
 
 We will prove most (hp-)admissibility and (hp-)invertibility results for the nested calculus $\nndl$. By \rmk~\ref{rmk:fo-over-prop}, $\nndl$ properly extends $\nipc$, $\ngd$, and $\nnd$. Therefore, by omitting the first-order cases in the proofs below or disregarding the $\lin$ cases, each proof serves as a corresponding (hp-)admissibility or (hp-)invertibility result for $\nipc$, $\ngd$, or $\nnd$. Moreover, as $\nndl$ is a close variant of $\ncdl$, all proofs below can be straightforwardly modified for $\ncd$ and $\ncdl$. 

\begin{figure}[t]
%\noindent\hrule

\begin{center}
\begin{tabular}{c c c}
\AxiomC{$\nsa\{\nsb\}_{w}$}
\RightLabel{$\wk$}
\UnaryInfC{$\nsa\{\nsb, \nsc\}_{w}$}
\DisplayProof

&

\AxiomC{$\nsa\{\va,x,x; \nsb\}_{w}$}
\RightLabel{$\ctrv$}
\UnaryInfC{$\nsa\{\va,x ; \nsb\}_{w}$}
\DisplayProof

&

\AxiomC{$\nsa\{\nsb, \bot^{\outp}\}_{w}$}
\RightLabel{$\botr$}
\UnaryInfC{$\nsa\{\nsb\}_{w}$}
\DisplayProof
\end{tabular}
\end{center}

%\medskip

\begin{center}
\begin{tabular}{c c c}
\AxiomC{$\nsa\{\va; \nsb\}_{w}\{\vb, \vc; \nsc\}_{u}$}
\RightLabel{$\ndr^{\dag_{1}}$}
\UnaryInfC{$\nsa\{\va, \vc; \nsb\}_{w}\{\vb; \nsc\}_{u}$}
\DisplayProof

&

\AxiomC{$\nsa\{\va, \vb; \nsb\}_{w}\{\vc; \nsc\}_{u}$}
\RightLabel{$\ddr^{\dag_{1}}$}
\UnaryInfC{$\nsa\{\va; \nsb\}_{w}\{\vc, \vb; \nsc\}_{u}$}
\DisplayProof

&

\AxiomC{$\nsa\{\va, \vb; \nsb\}_{w}$}
\RightLabel{$\cdr$}
\UnaryInfC{$\nsa\{\va; \nsb\}_{w}$}
\DisplayProof
\end{tabular}
\end{center}

%\medskip

\begin{center}
\begin{tabular}{c c c}
\AxiomC{$\nsa\{\va; \nsb\}_{w}$}
\RightLabel{$\wkv$}
\UnaryInfC{$\nsa\{\va,x ; \nsb\}_{w}$}
\DisplayProof

&

\AxiomC{$\nsa\{\va; \nsb, [\vb; \nsc]_{u}\}_{w}$}
\RightLabel{$\mrg$}
\UnaryInfC{$\nsa\{\va,\vb; \nsb, \nsc\}$}
\DisplayProof

&

\AxiomC{$\nsa\{\nsb, \phi^{\io}, \phi^{\io}\}_{w}$}
\RightLabel{$\ctr$}
\UnaryInfC{$\nsa\{\nsb, \phi^{\io}\}_{w}$}
\DisplayProof
\end{tabular}
\end{center}

%\medskip

\begin{center}
\begin{tabular}{c c c c}
\AxiomC{$\nsa\{\nsb, [\nsc]_{v} \}_{w}$}
\RightLabel{$\ex^{\dag_{2}}$}
\UnaryInfC{$\nsa\{\nsb, [\emptyset; [\nsc]_{v}]_{u} \}_{w}$}
\DisplayProof

&

\AxiomC{$\nsa$}
\RightLabel{$\psub$}
\UnaryInfC{$\nsa(y/x)$}
\DisplayProof

&

\AxiomC{$\nsa$}
\RightLabel{$\lsub^{\dag_{2}}$}
\UnaryInfC{$\nsa(w/u)$}
\DisplayProof

&

\AxiomC{$\nsa$}
\RightLabel{$\nec^{\dag_{2}}$}
\UnaryInfC{$\emptyset; [\nsa]_{u}$}
\DisplayProof
\end{tabular}
\end{center}

%\medskip

\begin{center}
\begin{tabular}{c c}
\AxiomC{$\nsa\{\nsb, \nsc^{\outp}\}_{w}\{\nsd\}_{u}$}
\RightLabel{$\lwr^{\dag_{1}}$}
\UnaryInfC{$\nsa\{\nsb\}_{w}\{\nsd, \nsc^{\outp}\}_{u}$}
\DisplayProof

&

\AxiomC{$\nsa\{\nsb\}_{w}\{\nsc, \nsd^{\inp}\}_{u}$}
\RightLabel{$\lft^{\dag_{1}}$}
\UnaryInfC{$\nsa\{\nsb, \nsd^{\inp}\}_{w}\{\nsc\}_{u}$}
\DisplayProof
\end{tabular}
\end{center}

%\medskip

\begin{center}
\begin{tabular}{c c}
\AxiomC{$\nsa\{[\va; \nsb]_{u}, [\vb; \nsc]_{v}\}_{w}$}
\RightLabel{$\ec$}
\UnaryInfC{$\nsa\{[\va,\vb; \nsb,\nsc]_{u}\}_{w}$}
\DisplayProof

&

\AxiomC{$\nsa\{\nsb, \phi^{\outp}\}_{w}$}
\AxiomC{$\nsa\{\nsb, \phi^{\inp}\}_{w}$}
\RightLabel{$\cut$}
\BinaryInfC{$\nsa\{\nsb\}_{w}$}
\DisplayProof
\end{tabular}
\end{center}

\noindent
\textbf{Side conditions:}\\
$\dag_{1}$ stipulates that the rule is applicable only if $\rtable{w}{u}$.\\
$\dag_{2}$ stipulates that $u$ must be fresh.

\caption{Admissible rules. We let $\io \in \{\inp,\outp\}$ in $\ctr$ and note that $\nsc^{\outp}$ and $\nsd^{\inp}$ represent multisets of output and input formulae, respectively, in the $\lwr$ and $\lft$ rules.}
\label{fig:lab-struc-rules}
\end{figure}

\begin{lemma}\label{lem:sub-admiss}
If $\ncalc \in \{\nipc,\nnd,\ncd,\ngd,\nndl,\ncdl\}$, then the $\psub$ and $\lsub$ rules are hp-admissible in $\ncalc$.
\end{lemma}

\begin{proof} We prove the result by induction on the height of the given derivation for $\nndl$. We argue the $\psub$ case as the $\lsub$ case is trivial.

\textit{Base case.} If $\psub$ is applied to an instance of $\idfo$ or $\botl$, then the conclusion is an instance of each rule resolving the base case.

\textit{Inductive step.} With the exception of the $\existsrii$, $\existsl$, $\alllii$ and $\allr$ rules, which have freshness conditions, the $\psub$ rule freely permutes above every rule of $\nnd$. Hence, these cases are easily resolved. We show how to resolve a non-trivial $\alllii$ case and omit the other cases as they are argued in a similar fashion.

 Suppose we have an application of $\psub$ after an application of $\alllii$ as shown below left, where $\psub$ substitutes in the fresh variable $y$. Then, by applying IH twice to first substitute a fresh variable $z$ for the variable $y$, and then substituting $y$ for $x$, as shown below right, we obtain the desired conclusion.
\begin{flushleft}
\begin{tabular}{c c}
\AxiomC{$\nsa \hol \va,y; \nsb \hor_{w} \hol \forall x \phi^{\inp} \hor_{u} \hol \nsc, \phi(y/x)^{\inp} \hor_{v}$}
\RightLabel{$\alllii$}
\UnaryInfC{$\nsa \hol \va; \nsb \hor_{w} \hol \forall x \phi^{\inp} \hor_{u} \hol \nsc \hor_{v}$}
\RightLabel{$\psub$}
\UnaryInfC{$(\nsa \hol \va; \nsb \hor_{w} \hol \forall x \phi^{\inp} \hor_{u} \hol \nsc \hor_{v})(y/x)$}
\DisplayProof

&

$\leadsto$
\end{tabular}
\end{flushleft}
\begin{flushright}
\AxiomC{$\nsa \hol \va,y; \nsb \hor_{w} \hol \forall x \phi^{\inp} \hor_{u} \hol \nsc, \phi(y/x)^{\inp} \hor_{v}$}
\RightLabel{IH} %{$\psub$}
\UnaryInfC{$\nsa \hol \va,z; \nsb \hor_{w} \hol \forall x \phi^{\inp} \hor_{u} \hol \nsc, \phi(z/x)^{\inp} \hor_{v}$}
\RightLabel{IH} %{$\psub$}
\UnaryInfC{$(\nsa \hol \va,z; \nsb \hor_{w} \hol \forall x \phi^{\inp} \hor_{u} \hol \nsc, \phi(z/x)^{\inp} \hor_{v})(x/y)$}
\RightLabel{$\alllii$}
\UnaryInfC{$(\nsa \hol \va; \nsb \hor_{w} \hol \forall x \phi^{\inp} \hor_{u} \hol \nsc \hor_{v})(x/y)$}
\DisplayProof
\end{flushright}
\end{proof}

\begin{lemma}\label{lem:botl-nec-admiss}
If $\ncalc \in \{\nipc,\nnd,\ncd,\ngd,\nndl,\ncdl\}$, then the $\botr$ and $\nec$ rules are hp-admissible in $\ncalc$.
\end{lemma}

\begin{proof} All results are shown by induction on the height of the given derivation. The base cases are trivial as any application of one of the rules to an instance of $\idfo$ or $\botl$ is another instance of the rule. The inductive steps are straightforward as well since both rules permute above every rule of $\ncalc$.
\end{proof}

\begin{lemma}\label{lem:wkv-ctrv-admiss}
If $\ncalc \in \{\nnd,\ncd,\nndl,\ncdl\}$, then the $\wkv$ and $\ctrv$ rules are hp-admissible in $\ncalc$.
\end{lemma}

\begin{proof} Both rules are proven hp-admissible by induction on the height of the given derivation. The proofs are straightforward in both cases, so we only show the $\existsri$ case for $\ctrv$. If $\ctrv$ is applied on a principal variable $y$ of $\existsri$ as shown below left, then due to the additional copy of the variable $y$, the two rules may be permuted as shown below right.
\begin{flushleft}
\begin{tabular}{c c}
\AxiomC{$\nsa \hol \vb, y,y; \nsc \hor_{v} \hol \va, y; \nsb \hor_{w} \hol \phi(y/x)^{\outp}, \exists x \phi^{\outp}  \hor _{u}$}
\RightLabel{$\existsri$}
\UnaryInfC{$\nsa \hol \vb, y,y; \nsc \hor_{v} \hol \va; \nsb \hor_{w} \hol \phi(y/x)^{\outp}, \exists x \phi^{\outp}  \hor _{u}$}
\RightLabel{$\ctrv$}
\UnaryInfC{$\nsa \hol \vb, y; \nsc \hor_{v} \hol \va; \nsb \hor_{w} \hol \phi(y/x)^{\outp}, \exists x \phi^{\outp}  \hor _{u}$}
\DisplayProof

&

$\leadsto$
\end{tabular}
\end{flushleft}
\begin{flushright}
\AxiomC{$\nsa \hol \vb, y,y; \nsc \hor_{v} \hol \va, y; \nsb \hor_{w} \hol \phi(y/x)^{\outp}, \exists x \phi^{\outp}  \hor _{u}$}
\RightLabel{IH}
\UnaryInfC{$\nsa \hol \vb, y; \nsc \hor_{v} \hol \va, y; \nsb \hor_{w} \hol \phi(y/x)^{\outp}, \exists x \phi^{\outp}  \hor _{u}$}
\RightLabel{$\existsri$}
\UnaryInfC{$\nsa \hol \vb, y; \nsc \hor_{v} \hol \va; \nsb \hor_{w} \hol \phi(y/x)^{\outp}, \exists x \phi^{\outp}  \hor _{u}$}
\DisplayProof
\end{flushright}
\end{proof}

\begin{lemma}\label{lem:wk-rules-admiss}
If $\ncalc \in \{\nipc,\nnd,\ncd,\ngd,\nndl,\ncdl\}$, then the $\wk$ rule is hp-admissible in $\ncalc$.
\end{lemma}

\begin{proof} We prove the lemma by induction on the height of the given derivation for $\nndl$. We note that the base case is trivial as any application of $\wk$ to $\idfo$ or $\botl$ yields another instance of the rule, therefore, we focus on showing the inductive step. The only non-trivial cases of the inductive step occur when $\wk$ is applied to the conclusion of rule with a freshness condition, namely, $\existsl$, $\existsrii$, $\alllii$, or $\allr$. We show how to resolve the non-trivial $\existsl$ case as the remaining cases are similar.

Let us assume that $\wk$ introduces a nested sequent $\nsb$ containing the variable $y$, which is fresh in the $\existsl$ inference. By replacing $y$ with a fresh variable $z$, followed by an application of the hp-admissible rule $\psub$ (see \lem~\ref{lem:sub-admiss} above), and then IH (i.e. $\wk$) and $\existsl$, we obtain the desired conclusion (shown below right).
\begin{center}
\begin{tabular}{c c c}
\AxiomC{$\nsa\{\va, y; \phi(\parama/x)^{\inp}\}$}
\RightLabel{$\existsl$}
\UnaryInfC{$\nsa\{\va; \exists x \phi^{\inp}\}$}
\RightLabel{$\wk$}
\UnaryInfC{$\nsa\{\va; \exists x \phi^{\inp}, \nsb\}$}
\DisplayProof

&

$\leadsto$

&

\AxiomC{$\nsa\{\va, y; \phi(\parama/x)^{\inp}\}$}
\RightLabel{$\psub$}
\UnaryInfC{$\nsa\{\va,z; \phi(z/x)^{\inp}\}$}
\RightLabel{IH}
\UnaryInfC{$\nsa\{\va,z; \phi(z/x)^{\inp}, \nsb\}$}
\RightLabel{$\existsl$}
\UnaryInfC{$\nsa\{\va; \exists x \phi^{\inp}, \nsb\}$}
\DisplayProof
\end{tabular}
\end{center}
\end{proof}

\begin{lemma}\label{lem:nd-cd-admiss} (1) The $\ndr$ rule is hp-admissible in $\nnd$ and $\nndl$; (2) The $\ndr$, $\ddr$, and $\cdr$ rules are hp-admissible in $\ncd$ and $\ncdl$.
\end{lemma}

\begin{proof} We argue claim 1 for $\nndl$ as the other claims are similar. The proof is by induction on the height of the given derivation. As the bases cases are trivial, we only consider the inductive step. Moreover, we note that in the inductive step, the only non-trivial cases occur when $\ndr$ is applied after an application of $\existsri$, $\allli$, or $\lin$. We argue a non-trivial $\allli$ case and omit the $\existsri$ case as it is similar, and show the non-trivial $\lin$ case.

$\allli$. As shown in the inference below left, $\ndr$ shifts the variable $y$ from the $u$-component to the $w$-component. As shown below right, we may resolve the case by first applying $\ndr$, and then since $\rable{w}{u}$ and $\rable{u}{v}$ hold due to the side conditions in the proof below left, it follows that $\rable{w}{v}$, meaning that $\allli$ may be applied after $\ndr$.
\begin{flushleft}
\begin{tabular}{c c}
\AxiomC{$\nsa \hol \va; \nsb, \forall x \phi^{\inp} \hor_{w} \hol \vb,y; \nsc \hor_{u} \hol \nsd, \phi(y/x)^{\inp} \hor_{v}$}
\RightLabel{$\allli$}
\UnaryInfC{$\nsa \hol \va; \nsb, \forall x \phi^{\inp} \hor_{w} \hol \vb,y; \nsc \hor_{u} \hol \nsd \hor_{v}$}
\RightLabel{$\ndr$}
\UnaryInfC{$\nsa \hol \va,y; \nsb, \forall x \phi^{\inp} \hor_{w} \hol \vb; \nsc \hor_{u} \hol \nsd \hor_{v}$}
\DisplayProof

&

$\leadsto$
\end{tabular}
\end{flushleft}
\begin{flushright}
\AxiomC{$\nsa \hol \va; \nsb, \forall x \phi^{\inp} \hor_{w} \hol \vb,y; \nsc \hor_{u} \hol \nsd, \phi(y/x)^{\inp} \hor_{v}$}
\RightLabel{IH}
\UnaryInfC{$\nsa \hol \va,y; \nsb, \forall x \phi^{\inp} \hor_{w} \hol \vb; \nsc \hor_{u} \hol \nsd, \phi(y/x)^{\inp} \hor_{v}$}
\RightLabel{$\allli$}
\UnaryInfC{$\nsa \hol \va,y; \nsb, \forall x \phi^{\inp} \hor_{w} \hol \vb; \nsc \hor_{u} \hol \nsd \hor_{v}$}
\DisplayProof
\end{flushright}

$\lin$. As shown in the proof below, $\ndr$ shifts a collection of variables from the $u$-component to the $w$-component:
\begin{center}
\AxiomC{$\nsa\{\va; \nsb\}_{w}\{[\nsc,[\vb,\vc;\nsd]_{u}]_{v}\}_{i}$}
\AxiomC{$\nsa\{\va; \nsb\}_{w}\{[\vb,\vc;\nsd,[\nsc]_{v}]_{u}\}_{i}$}
\RightLabel{$\lin$}
\BinaryInfC{$\nsa\{\va; \nsb\}_{w}\{[\vb,\vc;\nsd]_{u},[\nsc]_{v}\}_{i}$}
\RightLabel{$\ndr$}
\UnaryInfC{$\nsa\{\va,\vc; \nsb\}_{w}\{[\vb;\nsd]_{u},[\nsc]_{v}\}_{i}$}
\DisplayProof
\end{center}
 By the side condition on the $\ndr$ rule in the proof above, we know that $\rable{w}{i}$. It thus follows that in the premises of $\lin$ (shown above) $\rable{w}{u}$ holds since $\rable{i}{u}$ holds in both premises, implying that $\ndr$ may be applied to each premise of $\lin$, as shown below. A single application of $\lin$ gives the desired conclusion.
\begin{center}
\AxiomC{$\nsa\{\va; \nsb\}_{w}\{[\nsc,[\vb,\vc;\nsd]_{u}]_{v}\}_{i}$}
\RightLabel{IH} %{$\ndr$}
\UnaryInfC{$\nsa\{\va,\vc; \nsb\}_{w}\{[\nsc,[\vb;\nsd]_{u}]_{v}\}_{i}$}

\AxiomC{$\nsa\{\va; \nsb\}_{w}\{[\vb,\vc;\nsd,[\nsc]_{v}]_{u}\}_{i}$}
\RightLabel{IH} %{$\ndr$}
\UnaryInfC{$\nsa\{\va,\vc; \nsb\}_{w}\{[\vb;\nsd,[\nsc]_{v}]_{u}\}_{i}$}
\RightLabel{$\lin$}
\BinaryInfC{$\nsa\{\va,\vc; \nsb\}_{w}\{[\vb;\nsd]_{u},[\nsc]_{v}\}_{i}$}
\DisplayProof
\end{center}
\end{proof}

\begin{lemma}\label{lem:hp-invert-prop}
If $\ncalc \in \{\nipc,\nnd,\ncd,\ngd,\nndl,\ncdl\}$, then the $\disl$, $\disr$, $\conl$, $\conr$, $\impl$, and $\impr$ rules are hp-invertible in $\ncalc$.
\end{lemma}

\begin{proof} The hp-invertibility of $\impl$ follows from the hp-admissibility of $\wk$ (\lem~\ref{lem:wk-rules-admiss}). We argue the hp-admissibility of the $\impr$ rule by induction on the height of the given derivation, and note that the remaining cases are simple or similar.

\textit{Base case.} Suppose that $\phi \imp \psi^{\outp}$ occurs in an instance of $\idfo$, as shown below left, or an instance of $\botl$, as shown below right.
\begin{center}
\begin{tabular}{c c}
\AxiomC{}
\RightLabel{$\idfo$}
\UnaryInfC{$\ns \hol \Gamma, \phi \imp \psi^{\outp} \hor_{v} \hol p(\vec{\parama})^{\inp} \hor_{w} \hol p(\vec{\parama})^{\outp} \hor_{u}$}
\DisplayProof

&

\AxiomC{}
\RightLabel{$\botl$}
\UnaryInfC{$\nx \hol \Gamma, \phi \imp \psi^{\outp} \hor_{v} \hol  \bot^{\inp} \hor_{w} $}
\DisplayProof
\end{tabular}
\end{center}
 It is simple to verify the invertibility of the $\impr$ rule in these cases as witnessed by the $\idfo$ instance below left, and the $\botl$ instance below right.
\begin{center}
\begin{tabular}{c c}
\AxiomC{}
\RightLabel{$\idfo$}
\UnaryInfC{$\ns \hol [\emptyset; \phi^{\inp}, \psi^{\outp}]_{u} \hor_{v} \hol p(\vec{\parama})^{\inp} \hor_{w} \hol p(\vec{\parama})^{\outp} \hor_{u}$}
\DisplayProof

&

\AxiomC{}
\RightLabel{$\botl$}
\UnaryInfC{$\nx \hol \Gamma, [\emptyset; \phi^{\inp}, \psi^{\outp}]_{u} \hor_{v} \hol  \bot^{\inp} \hor_{w} $}
\DisplayProof
\end{tabular}
\end{center}

\textit{Inductive step.} With the exception of the $\impr$ case, all cases are resolved by invoking IH, and then applying the corresponding rule. For example, suppose that the last rule applied in the given derivation is $\allli$ and is of the form shown below left. As shown below right, the case is resolved by invoking IH and then applying $\allli$ as the path $\rable{w}{u}$ is still present after IH has been applied.
\begin{flushleft}
\begin{tabular}{c c}
\AxiomC{$\nsa \hol \forall x \phi^{\inp}, \psi \imp \chi^{\outp} \hor_{w} \hol  \nsb, \phi(y/x)^{\inp} \hor_{u}$}
\RightLabel{$\allli$}
\UnaryInfC{$\nsa \hol \forall x \phi^{\inp}, \psi \imp \chi^{\outp} \hor_{w} \hol  \nsb \hor_{u}$}
\DisplayProof

&

$\leadsto$
\end{tabular}
\end{flushleft}
\begin{flushright}
\AxiomC{$\nsa \hol \forall x \phi^{\inp}, \psi \imp \chi^{\outp} \hor_{w} \hol  \nsb, \phi(y/x)^{\inp} \hor_{u}$}
\RightLabel{IH}
\UnaryInfC{$\nsa \hol \forall x \phi^{\inp}, [\emptyset; \psi^{\inp}, \chi^{\outp}]_{v} \hor_{w} \hol  \nsb, \phi(y/x)^{\inp} \hor_{u}$}
\RightLabel{$\allli$}
\UnaryInfC{$\nsa \hol \forall x \phi^{\inp}, [\emptyset; \psi^{\inp}, \chi^{\outp}]_{v} \hor_{w} \hol  \nsb \hor_{u}$}
\DisplayProof
\end{flushright}
 If the last rule applied in the given derivation is $\impr$, then either the formula we aim to invert is principal, or it is not. In the latter case, we invoke IH, and then apply the corresponding rule, and in the former case, shown below, the desired conclusion is obtained by taking the proof of the premise.
\begin{center}
\begin{tabular}{c c c}
\AxiomC{$\nx \hol \Gamma, \nest{\emptyset;  \phi^{\inp}, \psi^{\outp}}{u} \hor_{w} $}
\RightLabel{$\impr$}
\UnaryInfC{$\nx \hol \Gamma, \phi \imp \psi^{\outp} \hor_{w} $}
\DisplayProof

&

$\leadsto$

&

\AxiomC{$\nx \hol \Gamma, \nest{\emptyset; \phi^{\inp}, \psi^{\outp}}{u} \hor_{w} $}
\DisplayProof
\end{tabular}
\end{center} 
\end{proof}

\begin{lemma}\label{lem:hp-invert-fo}
If $\ncalc \in \{\nnd,\ncd,\nndl,\ncdl\}$, then the $\doms$, $\existsl$, $\existsri$, $\existsrii$, $\allli$, $\alllii$, and $\allr$ rules are hp-invertible in $\ncalc$.
\end{lemma}

\begin{proof} The hp-invertibility of $\doms$, $\existsri$, $\existsrii$, $\allli$, and $\alllii$ follows from the hp-admissibility of $\wk$ (\lem~\ref{lem:wk-rules-admiss}) and $\wkv$ (\lem~\ref{lem:wkv-ctrv-admiss}). The $\existsl$ and $\allr$ cases are argued similarly to \lem~\ref{lem:hp-invert-prop} above.
\end{proof}

\begin{lemma}\label{lem:ctrl-hp-admiss}
If $\ncalc \in \{\nipc,\nnd,\ncd,\ngd,\nndl,\ncdl\}$, then the $\ctrl$ rule is hp-admissible in $\ncalc$.
\end{lemma}

\begin{proof} By induction on the height of the given derivation for $\nndl$.

\textit{Base case.} Any application of $\ctrl$ to $\idfo$ or $\botl$ yields another instance of the rule, showing the hp-admissibility of $\ctrl$ in these cases.

\textit{Inductive step.} Let us suppose that our derivation ends with an application of a rule $\ru$ followed by an application of $\ctrl$. If neither of the contracted formulae in $\ctrl$ are the principal formula of $\ru$, then we can freely permute $\ctrl$ above $\ru$ to obtain the same conclusion, but with the height of $\ctrl$ decreased by one. Therefore, let us assume that a principal formula of $\ru$ serves as one of the contracted formulae in $\ctrl$. We show the case where $\ru$ is $\existsl$ and note that the remaining cases are similar. The $\existsl$ case is resolved as shown below:
\begin{center}
\begin{tabular}{c c c}
\AxiomC{$\nx \hol \va,y; \phi(\parama/x)^{\inp}, \exists x \phi^{\inp}  \hor_{w} $}
\RightLabel{$\existsl$}
\UnaryInfC{$\nx \hol \va; \exists x \phi^{\inp}, \exists x \phi^{\inp}\hor_{w} $}
\RightLabel{$\ctrl$}
\UnaryInfC{$\nx \hol \va; \exists x \phi^{\inp}\hor_{w} $}
\DisplayProof

&

$\leadsto$

&

\AxiomC{$\nx \hol \va,y; \phi(y/x)^{\inp}, \exists x \phi^{\inp}  \hor_{w} $}
\RightLabel{\lem~\ref{lem:hp-invert-fo}}
\UnaryInfC{$\nx \hol \va,y,y; \phi(\parama/x)^{\inp}, \phi(\parama/x)^{\inp}  \hor_{w} $}
\RightLabel{$\ctrv$}
\UnaryInfC{$\nx \hol \va,y; \phi(\parama/x)^{\inp}, \phi(\parama/x)^{\inp}  \hor_{w} $}
\RightLabel{IH}
\UnaryInfC{$\nx \hol \va; \phi(\parama/x)^{\inp}  \hor_{w} $}
\RightLabel{$\existsl$}
\UnaryInfC{$\nx \hol \va; \exists x \phi^{\inp}\hor_{w} $}
\DisplayProof
\end{tabular}
\end{center}
\end{proof}

\begin{lemma}\label{lem:lft-admiss}
If $\ncalc \in \{\nipc,\nnd,\ncd,\ngd,\nndl,\ncdl\}$, then the $\lft$ rule is hp-admissible in $\ncalc$.
\end{lemma}

\begin{proof} We prove the result by induction on the height of the given derivation for $\nndl$.

\textit{Base case.} One can easily verify that any application of $\lft$ to $\botl$ yields another instance of the rule, resolving the case. If $\lft$ is applied to an instance of $\idfo$ and no principal formula of $\idfo$ is active in $\lft$, then the conclusion will also be an instance of $\idfo$. In the alternative case, our inference will be of the shape shown below. Observe that since $\rable{w}{u}$ and $\rable{u}{v}$ hold, we know that $\rable{w}{v}$ holds, showing that the conclusion is an instance of $\idfo$.
\begin{center}
\AxiomC{}
\RightLabel{$\idfo$}
\UnaryInfC{$\nsa\{\nsb\}_{w}\{\nsc_{1}, \nsc_{2}^{\inp}, p(\vec{x})^{\inp}\}_{u}\{\nsd,p(\vec{x})^{\outp}\}_{v}$}
\RightLabel{$\lft$}
\UnaryInfC{$\nsa\{\nsb,\nsc_{2}^{\inp}, p(\vec{x})^{\inp}\}_{w}\{\nsc_{1}\}_{u}\{\nsd,p(\vec{x})^{\outp}\}_{v}$}
\DisplayProof
\end{center}

\textit{Inductive step.} We show how $\lft$ permutes above the $\existsl$ and $\lin$ rules as the remaining cases are simple or similar.

$\existsl$. In the non-trivial $\existsl$ case the $\lft$ rule shifts the principal formula of the $\existsl$ inference as shown below left. The case is resolved by applying the hp-admissibility of $\ndr$ (\lem~\ref{lem:nd-cd-admiss} above), followed by $\lft$, and then $\existsl$ as shown below right:
\begin{flushleft}
\begin{tabular}{c c}
\AxiomC{$\nsa \hol \va; \nsb \hor_{u} \hol  \vb, y; \nsc, \nsd^{\inp},  \phi(y/x)^{\inp}  \hor_{w} $}
\RightLabel{$\existsl$}
\UnaryInfC{$\nsa \hol \va; \nsb \hor_{u} \hol  \vb; \nsc, \nsd^{\inp},  \exists x \phi^{\inp}  \hor_{w} $}
\RightLabel{$\lft$}
\UnaryInfC{$\nsa \hol \va; \nsb, \nsd^{\inp},  \exists x \phi^{\inp} \hor_{u} \hol  \vb; \nsc,  \hor_{w} $}
\DisplayProof

&

$\leadsto$
\end{tabular}
\end{flushleft}
\begin{flushright}
\AxiomC{$\nsa \hol \va; \nsb \hor_{u} \hol  \vb,y; \nsc, \nsd^{\inp},  \phi(y/x)^{\inp}  \hor_{w} $}
\RightLabel{$\ndr$}
\UnaryInfC{$\nsa \hol \va,y; \nsb \hor_{u} \hol  \vb; \nsc, \nsd^{\inp},  \phi(y/x)^{\inp}  \hor_{w} $}
\RightLabel{IH}
\UnaryInfC{$\nsa \hol \va,y; \nsb,\nsd^{\inp},  \phi(y/x)^{\inp}  \hor_{u} \hol  \vb; \nsc \hor_{w} $}
\RightLabel{$\existsl$}
\UnaryInfC{$\nsa \hol \va; \nsb,\nsd^{\inp},  \exists x \phi^{\inp}  \hor_{u} \hol  \vb; \nsc \hor_{w} $}
\DisplayProof
\end{flushright}

$\lin$. We show one of the non-trivial $\lin$ cases as the other cases are simple or similar. As can be seen below left, the $\lft$ rule shifts data from the $v$-component to the $w$-component. The case is resolved by applying $\lft$ to the premises of $\lin$ and then applying $\lin$ as shown below right.
\begin{flushleft}
\begin{tabular}{c c}
\AxiomC{$\nsa\{\nsd\}_{w}\{[\nsb,\nsb^{\inp},[\nsc]_{u}]_{v}\}_{i}$}
\AxiomC{$\nsa\{\nsd\}_{w}\{[\nsc,[\nsb,\nsb^{\inp}]_{v}]_{u}\}_{i}$}
\RightLabel{$\lin$}
\BinaryInfC{$\nsa\{\nsd\}_{w}\{[\nsb,\nsb^{\inp}]_{v},[\nsc]_{u}\}_{i}$}
\RightLabel{$\lft$}
\UnaryInfC{$\nsa\{\nsd, \nsb^{\inp}\}_{w}\{[\nsb]_{v},[\nsc]_{u}\}_{i}$}
\DisplayProof

&

$\leadsto$
\end{tabular}
\end{flushleft}
\begin{flushright}
\AxiomC{$\nsa\{\nsd\}_{w}\{[\nsb,\nsb^{\inp},[\nsc]_{u}]_{v}\}_{i}$}
\RightLabel{IH}
\UnaryInfC{$\nsa\{\nsd,\nsb^{\inp}\}_{w}\{[\nsb,[\nsc]_{u}]_{v}\}_{i}$}

\AxiomC{$\nsa\{\nsd\}_{w}\{[\nsc,[\nsb,\nsb^{\inp}]_{v}]_{u}\}_{i}$}
\RightLabel{IH}
\UnaryInfC{$\nsa\{\nsd,\nsb^{\inp}\}_{w}\{[\nsc,[\nsb]_{v}]_{u}\}_{i}$}
\RightLabel{$\lin$}
\BinaryInfC{$\nsa\{\nsd, \nsb^{\inp}\}_{w}\{[\nsb]_{v},[\nsc]_{u}\}_{i}$}
\DisplayProof
\end{flushright}
\end{proof}

\begin{lemma}\label{lem:mrg-admiss}
If $\ncalc \in \{\nipc,\nnd,\ncd,\ngd,\nndl,\ncdl\}$, then the $\mrg$ rule is hp-admissible in $\ncalc$.
\end{lemma}

\begin{proof} We prove the result by induction on the height of the given derivation for $\nndl$.

\textit{Base case.} Applying $\mrg$ to $\botl$ yields another instance of the rule, showing that the claim holds in this case. Therefore, let us consider an application of $\mrg$ to an instance of $\idfo$. Note that we only consider a non-trivial case below, which occurs when $\mrg$ is applied to a component containing a principal formula of $\idfo$; all remaining cases are similar.

Let us suppose that our instance of $\idfo$ is as shown below left. By the side condition on $\idfo$, we know that $\rable{w}{u}$ and $\rable{u}{v}$. By the former fact, it follows that the end sequent in the proof shown below is an instance of $\idfo$.
\begin{center}
\AxiomC{}
\RightLabel{$\idfo$}
\UnaryInfC{$\nsa \hol p(\vec{x})^{\inp} \hor_{w} \hol \va; \nsb,  \nest{\vb; \nsc, p(\vec{x})^{\outp}}{v} \hor_{u}$}
\RightLabel{$\mrg$}
\UnaryInfC{$\nsa \hol p(\vec{x})^{\inp} \hor_{w} \hol \va,\vb; \nsb, \nsc, p(\vec{x})^{\outp} \hor_{u}$}
\DisplayProof
\end{center}

\textit{Inductive step.} With the exception of the $\impl$, $\existsri$, $\existsrii$, $\allli$, $\alllii$, and $\lin$ cases, all other cases are easily resolved by invoking IH (i.e. applying $\mrg$) and then applying the corresponding rule. We show a non-trivial $\allli$ and $\lin$ case below, omitting the other cases as they are simple or similar. 

$\allli$. For the $\allli$ case, we suppose that $y$ is available for $v$. Therefore, since $\rable{w}{u}\rable{}{v}$ holds by the side condition on $\allli$, we have that $y$ is available for $u$. We may permute $\mrg$ and $\allli$ as shown below right.
\begin{flushleft}
\begin{tabular}{c c}
\AxiomC{$\nsa \hol \forall x \phi^{\inp} \hor_{w} \hol \va; \nsb, [\vb; \nsc, \phi(\parama/x)]_{v} \hor_{u}$}
\RightLabel{$\allli$}
\UnaryInfC{$\nsa \hol \forall x \phi^{\inp} \hor_{w} \hol \va; \nsb, [\vb; \nsc]_{v} \hor_{u}$}
\RightLabel{$\mrg$}
\UnaryInfC{$\nsa \hol \forall x \phi^{\inp} \hor_{w} \hol \va, \vb; \nsb, \nsc \hor_{u}$}
\DisplayProof

&

$\leadsto$
\end{tabular}
\end{flushleft}
\begin{flushright}
\AxiomC{$\nsa \hol \forall x \phi^{\inp} \hor_{w} \hol \va; \nsb, [\vb; \nsc, \phi(\parama/x)]_{v} \hor_{u}$}
\RightLabel{IH}
\UnaryInfC{$\nsa \hol \forall x \phi^{\inp} \hor_{w} \hol \va, \vb; \nsb, \nsc, \phi(\parama/x) \hor_{u}$}
\RightLabel{$\allli$}
\UnaryInfC{$\nsa \hol \forall x \phi^{\inp} \hor_{w} \hol \va, \vb; \nsb, \nsc \hor_{u}$}
\DisplayProof
\end{flushright}

$\lin$. For the $\lin$ case, we suppose that the $w$- and $v$-components are fused via an application of $\mrg$. We may derive the desired conclusion by applying IH (i.e. $\mrg$) to the left premise of $\lin$. 
\begin{flushleft}
\begin{tabular}{c c c}
\AxiomC{$\nsa\{\va;\nsb, [\vb; \nsc,[\nsd]_{u}]_{v}\}_{w}$}
\AxiomC{$\nsa\{\va;\nsb, [\nsd,[\vb; \nsc]_{v}]_{u}\}_{w}$}
\RightLabel{$\lin$}
\BinaryInfC{$\nsa\{\va;\nsb, [\vb; \nsc]_{v},[\nsd]_{u}\}_{w}$}
\RightLabel{$\mrg$}
\UnaryInfC{$\nsa\{\va,\vb;\nsb,\nsc,[\nsd]_{u}\}_{w}$}
\DisplayProof

&

$\leadsto$
\end{tabular}
\end{flushleft}
\begin{flushright}
\AxiomC{$\nsa\{\va;\nsb, [\vb; \nsc,[\nsd]_{u}]_{v}\}_{w}$}
\RightLabel{IH}
\UnaryInfC{$\nsa\{\va,\vb;\nsb,\nsc,[\nsd]_{u}\}_{w}$}
\DisplayProof
\end{flushright}
\end{proof}

We now argue that the $\ex$ rule is hp-admissible in $\nipc$, $\nnd$, and $\ncd$, while being strictly admissible in $\ngd$, $\nnd$, and $\ncd$. This discrepancy relies on the fact that to permute $\ex$ above $\lin$, we require two subsequent applications of $\lin$ to derive the same conclusion, thus potentially growing the size of the derivation.

\begin{lemma}\label{lem:ex-admiss} (1) The $\ex$ rule is hp-admissible in $\nipc$, $\nnd$, and $\ncd$; (2) The $\ex$ rule is admissible in $\ngd$, $\nndl$, and $\ncdl$.
\end{lemma}

\begin{proof} We prove the lemma by induction on the height of the given derivation for $\nndl$. In the inductive step, only when permuting $\ex$ above $\lin$ will the size of the derivation potentially grow, and thus, the following establishes the hp-admissibility of $\ex$ for $\nipc$, $\nnd$, and $\ncd$ which exclude the rule $\lin$.

\textit{Base case.} The $\botl$ case is easily resolved since any application of $\ex$ to $\botl$ gives another instance of the rule. The $\idfo$ is also straightforward; for example, suppose we have an instance of $\idfo$ followed by an application of the $\ex$ rule, as shown below left. Then, due to the side condition on $\idfo$, we know that $\rable{w}{u}$. After applying $\ex$, we can see that $\rable{w}{u}$ still holds, and thus, the conclusion of the proof is an instance of $\idfo$.
\begin{center}
\AxiomC{}
\RightLabel{$\idfo$}
\UnaryInfC{$\ns\{p(\vec{x})\}_{w}\{\nant, [\ncon\{p(\vec{x})\}_{u}]_{v} \}_{z}$}
\RightLabel{$\ex$}
\UnaryInfC{$\ns\{p(\vec{x})\}_{w}\{\nant, [\emptyset; [\ncon\{p(\vec{x})\}_{u}]_{v}]_{i} \}_{z}$}
\DisplayProof
\end{center}

\textit{Inductive step.} One can show via arguments similar to the base case that $\ex$ permutes above each reachability rule $(r)$ in our calculus. Furthermore, with the exception of $\lin$, it is simple to show that $\ex$ permutes above the remaining rules of our calculus. We show how to resolve a non-trivial $\lin$ case below and note that the remaining cases are argued similarly. 
\begin{center}
\AxiomC{$\nsa\{[\nsb,[\nsc]_{u}]_{v}\}_{w}$}
\AxiomC{$\nsa\{[\nsc,[\nsb]_{v}]_{u}\}_{w}$}
\RightLabel{$\lin$}
\BinaryInfC{$\nsa\{[\nsb]_{v},[\nsc]_{u}\}_{w}$}
\RightLabel{$\ex$}
\UnaryInfC{$\nsa\{[\emptyset; [\nsb]_{v}]_{i},[\nsc]_{u}\}_{w}$}
\DisplayProof
\end{center}
 In the proof shown above, the $\ex$ rule is applied in the $v$-component. We first apply the $\ex$ rule to each premise of $\lin$, followed by an application of $\lin$. By applying IH (i.e. $\ex$) to the right premise of $\lin$, we may then apply $\lin$ one last time to obtain the desired conclusion.
\begin{center}
\AxiomC{$\nsa\{[\nsb,[\nsc]_{u}]_{v}\}_{w}$}
\RightLabel{IH}
\UnaryInfC{$\nsa\{[\emptyset; [\nsb,[\nsc]_{u}]_{v}]_{i}\}_{w}$}

\AxiomC{$\nsa\{[\nsc,[\nsb]_{v}]_{u}\}_{w}$}
\RightLabel{IH}
\UnaryInfC{$\nsa\{[\emptyset; [\nsc,[\nsb]_{v}]_{u}]_{i}\}_{w}$}
\RightLabel{$\lin$}
\BinaryInfC{$\nsa\{[\emptyset; [\nsb]_{v}, [\nsc]_{u}]_{i}\}_{w}$}
\AxiomC{$\nsa\{[\nsc,[\nsb]_{v}]_{u}\}_{w}$}
\RightLabel{IH}
\UnaryInfC{$\nsa\{[\nsc,[\emptyset; [\nsb]_{v}]_{i}]_{u}\}_{w}$}
\RightLabel{$\lin$}
\RightLabel{$\lin$}
\BinaryInfC{$\nsa\{[\emptyset; [\nsb]_{v}]_{i},[\nsc]_{u}\}_{w}$}
\DisplayProof
\end{center}
\end{proof}

%Due to the fact that $\ex$ is merely admissible in $\ngd$, $\nndl$, and $\nncl$, 
%$\ec$, $\lwr$, and $\ctrr$
%The following three lemmas concern the hp-admissibility of $\ec$, $\lwr$, and $\ctrr$, where the hp-admissibility of each rule depends on the hp-admissibility of the former. 

\begin{lemma}\label{lem:ec-admiss} If $\ncalc \in \{\nipc,\nnd,\ncd\}$, then the $\ec$ rule is hp-admissible in $\ncalc$. %; (2) If $\ncalc \in \{\ngd,\nndl,\ncdl\}$, then the $\ec$ rule is admissible in $\ncalc$.
\end{lemma}

\begin{proof} We argue claim (1) by induction on the height of the given derivation for $\nnd$ and note that claim (2) follows from the fact that each calculus is complete (\thm~\ref{thm:cut-free-comp}) and $\ec$ is sound.

\textit{Base case.} If $\ec$ is applied to $\botl$, then it yields another instance of the rule, showing the hp-admissibility of the rule in this case. Let us now consider applying $\ec$ to an instance of $\idfo$. We consider the following non-trivial case and note that all other cases are similar or simple.
\begin{center}
\AxiomC{}
\RightLabel{$\idfo$}
\UnaryInfC{$\nsa \hol p(\vec{\parama})^{\inp} \hor_{w} \hol \nsb\{[\va; \nsc, p(\vec{\parama})^{\outp}]_{u},[\vb; \nsd]_{v}\}_{j} \hor_{i}$}
\RightLabel{$\ec$}
\UnaryInfC{$\nsa \hol p(\vec{\parama})^{\inp} \hor_{w} \hol \nsb\{[\va,\vb; \nsc, p(\vec{\parama})^{\outp}, \nsd]_{u}\}_{j} \hor_{i}$}
\DisplayProof
\end{center}
 Since $\rable{w}{u}$ holds in the instance of $\idfo$, we have that $\rable{w}{u}$ holds in the conclusion of $\ec$. Therefore, we may take the conclusion above to be an instance of $\idfo$, resolving the case.

\textit{Inductive step.} We consider the $\existsri$ case as the other cases are simple or similar.

$\existsri$. In the non-trivial $\existsri$ case (shown below left), we have that $\rable{w}{v}$, from which it follows that $\rable{w}{u}$ holds in the conclusion of the proof. Therefore, we may apply $\ec$ first and then $\existsri$ second as the side condition will still hold, thus showing that the two rules can be permuted as shown below right.
\begin{flushleft}
\begin{tabular}{c c}
\AxiomC{$\nsa \hol \va,y; \nsb \hor_{w} \hol \nsc\{[\vb; \nsd_{1}]_{u},[\vc; \nsd_{2},\phi(\parama/x)^{\outp}, \exists x \phi^{\outp}]_{v}\}_{j}  \hor _{i}$}
\RightLabel{$\existsri$}
\UnaryInfC{$\nsa \hol \va,y; \nsb \hor_{w} \hol \nsc\{[\vb;\nsd_{1}]_{u},[\vc; \nsd_{2}, \exists x \phi^{\outp}]_{v}\}_{j}  \hor _{i}$}
\RightLabel{$\ec$}
\UnaryInfC{$\nsa \hol \va,y; \nsb \hor_{w} \hol \nsc\{[\vb,\vc; \nsd_{1},\nsd_{2}, \exists x \phi^{\outp}]_{u}\}_{j}  \hor _{i}$}
\DisplayProof
&

$\leadsto$
\end{tabular}
\end{flushleft}
\begin{flushright}
\AxiomC{$\nsa \hol \va,y; \nsb \hor_{w} \hol \nsc\{[\vb; \nsd_{1}]_{u},[\vc; \nsd_{2},\phi(\parama/x)^{\outp}, \exists x \phi^{\outp}]_{v}\}_{j}  \hor _{i}$}
\RightLabel{IH}
\UnaryInfC{$\nsa \hol \va,y; \nsb \hor_{w} \hol \nsc\{[\vb,\vc; \nsd_{1}, \nsd_{2},\phi(\parama/x)^{\outp}, \exists x \phi^{\outp}]_{v}\}_{j}  \hor _{i}$}
\RightLabel{$\existsri$}
\UnaryInfC{$\nsa \hol \va,y; \nsb \hor_{w} \hol \nsc\{[\vb,\vc; \nsd_{1}, \nsd_{2}, \exists x \phi^{\outp}]_{v}\}_{j}  \hor _{i}$}
\DisplayProof
\end{flushright}
\end{proof}

%%%LWR Admiss
\begin{lemma}\label{lem:lwr-admiss} If $\ncalc \in \{\nipc,\nnd,\ncd\}$, then the $\lwr$ rule is hp-admissible in $\ncalc$. % (2) If $\ncalc \in \{\ngd,\nndl,\ncdl\}$, then the $\lwr$ rule is admissible in $\ncalc$.
\end{lemma}

\begin{proof} We prove claim (1) by induction on the height of the given derivation for $\nnd$ and note that claim (2) follows from the fact that each calculus is complete (\thm~\ref{thm:cut-free-comp}) and $\lwr$ is sound.

\textit{Base case.} The only non-trivial case to consider is when $\lwr$ is applied to an instance of $\idfo$ and a principal formula of $\idfo$ is active in $\lft$, as shown below. Observe that $\rable{w}{v}$ as $\rable{w}{u}$ and $\rable{u}{v}$, showing that the conclusion is an instance of $\idfo$ as well.
\begin{center}
\AxiomC{}
\RightLabel{$\idfo$}
\UnaryInfC{$\nsa\{\nsb,p(\vec{x})^{\inp}\}_{w}\{\nsc_{1}, \nsc_{2}^{\outp}, p(\vec{x})^{\outp}\}_{u}\{\nsd\}_{v}$}
\RightLabel{$\lwr$}
\UnaryInfC{$\nsa\{\nsb,p(\vec{x})^{\inp}\}_{w}\{\nsc_{1}\}_{u}\{\nsd,\nsc_{2}^{\outp}, p(\vec{x})^{\outp}\}_{v}$}
\DisplayProof
\end{center}

\textit{Inductive step.} We consider the non-trivial case of permuting $\lwr$ above $\allr$ as the remaining cases are simple or similar.

$\allr$. In the non-trivial $\allr$ case the $\lwr$ rule shifts the principal formula of $\allr$ as shown below left. The case is resolved as shown below right, and begins by invoking IH. Second, we repeatedly apply the hp-admissibility of $\ex$ (\lem~\ref{lem:ex-admiss}), creating a path whose terminal node is the $u$-component; we apply $\ex$ a sufficient number of times (say, $n$) for this path to be of a length one greater than the path between the $w$-component and the $v$-component. Third, we successively apply the hp-admissible $\ec$ rule (\lem~\ref{lem:ec-admiss}), fusing this path with the path from the $w$-component to the $v$-component until the $u$-component is nested within the $v$-component. A single application of $\allr$ gives the desired conclusion.
\begin{flushleft}
\begin{tabular}{c c}
\AxiomC{$\nsa \hol \nsb, \nsd^{\outp}, \nest{y ; \phi(y/x)^{\outp}}{u} \hor_{w} \hol \nsc \hor_{v} $}
\RightLabel{$\allr$}
\UnaryInfC{$\nsa \hol \nsb, \forall x \phi^{\outp}, \nsd^{\outp} \hor_{w} \hol \nsc \hor_{v}$}
\RightLabel{$\lwr$}
\UnaryInfC{$\nsa \hol \nsb \hor_{w} \hol \nsc, \forall x \phi^{\outp}, \nsd^{\outp} \hor_{v}$}
\DisplayProof

&

$\leadsto$
\end{tabular}
\end{flushleft}
\begin{flushright}
\AxiomC{$\nsa \hol \nsb, \nsd^{\outp}, \nest{y ; \phi(y/x)^{\outp}}{u} \hor_{w} \hol \nsc \hor_{v} $}
\RightLabel{IH}
\UnaryInfC{$\nsa \hol \nsb, \nest{y ; \phi(y/x)^{\outp}}{u} \hor_{w} \hol \nsc, \nsd^{\outp} \hor_{v} $}
\RightLabel{$\ex \times n$}
\UnaryInfC{$\nsa \hol \nsb, \nest{\emptyset; \ldots \nest{y ; \phi(y/x)^{\outp}}{u} \ldots }{i} \hor_{w} \hol \nsc, \nsd^{\outp} \hor_{v} $}
\RightLabel{$\ec \times n$}
\UnaryInfC{$\nsa \hol \nsb \hor_{w} \hol \nsc, \nsd^{\outp}, \nest{y ; \phi(y/x)^{\outp}}{u} \hor_{v} $}
\RightLabel{$\allr$}
\UnaryInfC{$\nsa \hol \nsb \hor_{w} \hol \nsc, \forall x \phi^{\outp}, \nsd^{\outp} \hor_{v}$}
\DisplayProof

\end{flushright}
\end{proof}

\begin{lemma}\label{lem:ctrr-hp-admiss} If $\ncalc \in \{\nipc,\nnd,\ncd\}$, then the $\ctrr$ rule is hp-admissible in $\ncalc$. % (2) If $\ncalc \in \{\ngd,\nndl,\ncdl\}$, then the $\ctrr$ rule is admissible in $\ncalc$.
\end{lemma}

\begin{proof} We argue claim (1) by induction on the height of the given derivation for $\nnd$. Claim (2) follows from the fact that each calculus is complete (\thm~\ref{thm:cut-free-comp}) and $\ctrr$ is sound.

\textit{Base case.} Any application of $\ctrr$ to $\idfo$ or $\botl$ yields another instance of the rule, showing the hp-admissibility of $\ctrr$ in these cases.

\textit{Inductive step.} Let us suppose that our derivation ends with an application of a rule $\ru$ followed by an application of $\ctrr$. If neither of the contracted formulae in $\ctrr$ are the principal formula of $\ru$, then we can freely permute $\ctrr$ above $\ru$ to obtain the same conclusion. Therefore, let us assume that the principal formulae of $\ru$ serves as one of the contracted formulae in $\ctrr$. We show the cases where $\ru$ is either $\impr$ or $\allr$, and note that the remaining cases are shown similarly.
\begin{center}
\begin{tabular}{c c c}
\AxiomC{$\nx \hol \Gamma, \phi \imp \psi^{\outp}, \onest{\emptyset;\phi^{\inp}, \psi^{\outp}}{u} \hor_{w} $}
\RightLabel{$\impr$}
\UnaryInfC{$\nx \hol \Gamma, \phi \imp \psi^{\outp}, \phi \imp \psi^{\outp} \hor_{w} $}
\RightLabel{$\ctrr$}
\UnaryInfC{$\nx \hol \Gamma, \phi \imp \psi^{\outp} \hor_{w} $}
\DisplayProof

&

$\leadsto$

&

\AxiomC{$\nx \hol \Gamma, \phi \imp \psi^{\outp}, \onest{\emptyset; \phi^{\inp}, \psi^{\outp}}{u} \hor_{w} $}
\RightLabel{$\lwr$}
\UnaryInfC{$\nx \hol \Gamma, \onest{\emptyset; \phi^{\inp}, \psi^{\outp},\phi \imp \psi^{\outp}}{u} \hor_{w} $}
\RightLabel{\lem~\ref{lem:hp-invert-prop}}
\UnaryInfC{$\nx \hol \Gamma, \onest{\emptyset; \phi^{\inp}, \psi^{\outp},[\emptyset; \phi^{\inp}, \psi^{\outp}]_{v}}{u} \hor_{w} $}
\RightLabel{$\mrg$}
\UnaryInfC{$\nx \hol \Gamma, \onest{\emptyset; \phi^{\inp}, \psi^{\outp},\phi^{\inp}, \psi^{\outp}}{u} \hor_{w} $}
\RightLabel{$\ctrl$}
\UnaryInfC{$\nx \hol \Gamma, \onest{\emptyset; \phi^{\inp}, \psi^{\outp}, \psi^{\outp}}{u} \hor_{w} $}
\RightLabel{IH}
\UnaryInfC{$\nx \hol \Gamma, \onest{\emptyset; \phi^{\inp}, \psi^{\outp}}{u} \hor_{w} $}
\RightLabel{$\impr$}
\UnaryInfC{$\nx \hol \Gamma, \phi \imp \psi^{\outp} \hor_{w} $}
\DisplayProof
\end{tabular}
\end{center}
 The $\allr$ case is resolved as shown below. We leverage the hp-admissible rules $\lwr$, $\mrg$, and $\ctrv$ as well as the hp-invertibility of $\allr$ in our proof.
\begin{flushleft}
\begin{tabular}{c c c}
\AxiomC{$\nx \hol \Gamma, \forall x \phi^{\outp}, \onest{y; \phi(y/x)^{\outp}}{u} \hor_{w} $}
\RightLabel{$\allr$}
\UnaryInfC{$\nx \hol \Gamma, \forall x \phi^{\outp}, \forall x \phi^{\outp} \hor_{w} $}
\RightLabel{$\ctrr$}
\UnaryInfC{$\nx \hol \Gamma, \forall x \phi^{\outp} \hor_{w} $}
\DisplayProof

&

$\leadsto$
\end{tabular}
\end{flushleft}
\begin{flushright}
\AxiomC{$\nx \hol \Gamma, \forall x \phi^{\outp}, \onest{y; \phi(y/x)^{\outp}}{u} \hor_{w} $}
\RightLabel{$\lwr$}
\UnaryInfC{$\nx \hol \Gamma, \onest{y; \forall x \phi^{\outp}, \phi(y/x)^{\outp}}{u} \hor_{w} $}
\RightLabel{\lem~\ref{lem:hp-invert-fo}}
\UnaryInfC{$\nx \hol \Gamma, \onest{y; \phi(y/x)^{\outp}, \onest{y; \phi(y/x)^{\outp}}{v}}{u} \hor_{w} $}
\RightLabel{$\mrg$}
\UnaryInfC{$\nx \hol \Gamma, \onest{y,y; \phi(y/x)^{\outp}, \phi(y/x)^{\outp}}{u} \hor_{w} $}
\RightLabel{$\ctrv$}
\UnaryInfC{$\nx \hol \Gamma, \onest{y; \phi(y/x)^{\outp}, \phi(y/x)^{\outp}}{u} \hor_{w} $}
\RightLabel{IH}
\UnaryInfC{$\nx \hol \Gamma, \onest{y; \phi(y/x)^{\outp}}{u} \hor_{w} $}
\RightLabel{$\allr$}
\UnaryInfC{$\nx \hol \Gamma, \forall x \phi^{\outp} \hor_{w} $}
\DisplayProof
\end{flushright}
\end{proof}

%Cut-elimination
\section{Syntactic Cut-elimination}\label{sec:cut-elim}

 We now show that the intuitionistic calculi $\nipc$, $\nnd$, and $\ncd$ satisfy syntactic cut-elimination, that is, the $\cut$ rule (\fig~\ref{fig:lab-struc-rules}) can be permuted upward in any given derivation and deleted at the initial rules. Syntactic cut-elimination results were first provided in the context of nested systems for propositional modal logics; in particular, \citet{Bru09} showed how to eliminate an additive (i.e. context-sharing) version of cut (similar to the $\cut$ rule we consider) and \citet{Pog09} showed how to eliminate a multiplicative (i.e. context-independent) version of cut. We will first prove our syntactic cut-elimination theorem and then comment on the issues associated with eliminating $\cut$ in the presence of $\lin$, i.e. for the $\ngd$, $\nndl$, and $\ncdl$ calculi.

\begin{theorem}[Cut-elimination]\label{thm:cut-elim-int}
If $\ncalc \in \{\nipc,\nnd,\ncd\}$, then the $\cut$ rule is eliminable in $\ncalc$.
\end{theorem}

\begin{proof} We prove the result for $\nnd$ as the proof for the other nested calculi are subsumed by this case or similar. The result is shown by induction on the lexicographic ordering of pairs $(\vert \phi \vert,h_{1}+h_{2})$, where $\vert \phi \vert$ is the complexity of the cut formula $\phi$, $h_{1}$ is the height of the derivation above the left premise of $\cut$, and $h_{2}$ is the height of the derivation above the right premise of $\cut$. We assume w.l.o.g. that $\cut$ is used once as the last inference in our given proof; the general result follows by successively applying the described procedure to topmost instances of $\cut$ in a given derivation.

1. Suppose that the complexity of the cut formula is $0$, i.e. the cut formula is either an atomic formula $p(\vec{x})$ or $\bot$. 

1.1. Suppose that both premises of $\cut$ are instances of $\idfo$, $\botl$, or $\doms$.

1.1.1. Suppose both premises of $\cut$ are instances of $\idfo$. If the cut formula is not principal in both premises, then the conclusion of $\cut$ is an instance of $\idfo$. Suppose then that the cut formula is principal in both premises, meaning that our $\cut$ is of the following form:
\begin{center}
\AxiomC{}
\RightLabel{$\idfo$}
\UnaryInfC{$\ns\{\Gamma_{1}, p(\vec{x})^{\inp}\}_{w}\{\Gamma_{2}, p(\vec{x})^{\outp}\}_{u}$}

\AxiomC{}
\RightLabel{$\idfo$}
\UnaryInfC{$\ns\{\Gamma_{2}, p(\vec{x})^{\inp}\}_{u}\{\Gamma_{3}, p(\vec{x})^{\outp}\}_{v}$}

\RightLabel{$\cut$}
\BinaryInfC{$\ns\{\Gamma_{1}, p(\vec{x})^{\inp}\}_{w}\{\Gamma_{2}\}_{u}\{\Gamma_{3}, p(\vec{x})^{\outp}\}_{v}$}
\DisplayProof
\end{center}
Since the left and right premises are instances of $\idfo$, $\rable{w}{u}$ and $\rable{u}{v}$ hold, implying that $\rable{w}{v}$ holds. Hence, the conclusion of $\cut$ is an instance of $\idfo$.

1.1.2. If the left premise of $\cut$ is an instance of $\doms$, then applying the hp-invertibility of $\doms$ (\lem~\ref{lem:hp-invert-fo}) to the right premise of $\cut$ lets us apply $\cut$ with the premise of $\doms$, and applying $\doms$ afterward gives the desired conclusion. Let us suppose then that the right premise of $\cut$ is an instance of $\doms$. If the left premise of $\cut$ is an instance of $\doms$, then the previous case applies, and if the left premise of $\cut$ is an instance of $\botl$, then the conclusion will be an instance of $\botl$ as well. Furthermore, if the left premise of $\cut$ is an instance of $\idfo$ and the cut formula is not principal in $\idfo$, then the conclusion is an instance of $\idfo$ as well. Therefore, the only case left that we need to consider is when $\idfo$ is the left premise of $\cut$ and a principal formula of $\idfo$ is the cut formula. If the principal formula in $\doms$ is not the cut formula, then we can apply the hp-invertibility of $\doms$ (\lem~\ref{lem:hp-invert-fo}) to the instance of $\idfo$, and shift the $\cut$ upward to the premise of $\doms$. Let us assume then that the principal formula in $\doms$ is the cut formula, meaning the $\cut$ is of the following form:
\begin{center}
\AxiomC{}
\RightLabel{$\idfo$}
\UnaryInfC{$\ns\{\va; \Gamma_{1}, p(\vec{x})^{\inp}\}_{w}\{\vb; \Gamma_{2}, p(\vec{x})^{\outp}\}_{u}$}

\AxiomC{$\ns\{\va; \Gamma_{1}, p(\vec{x})^{\inp}\}_{w}\{\vb, \vec{x}; \Gamma_{2}, p(\vec{x})^{\inp}\}_{u}$}
\RightLabel{$\doms$}
\UnaryInfC{$\ns\{\va; \Gamma_{1}, p(\vec{x})^{\inp}\}_{w}\{\vb; \Gamma_{2}, p(\vec{x})^{\inp}\}_{u}$}

\RightLabel{$\cut$}
\BinaryInfC{$\ns\{\va; \Gamma_{1}, p(\vec{x})^{\inp}\}_{w}\{\vb; \Gamma_{2}\}_{u}$}
\DisplayProof
\end{center}
 This case can be resolved by invoking the hp-admissibility of $\ndr$ (\lem~\ref{lem:nd-cd-admiss}), as shown below:
\begin{center}
\AxiomC{}
\RightLabel{$\idfo$}
\UnaryInfC{$\ns\{\va; \Gamma_{1}, p(\vec{x})^{\inp}\}_{w}\{\vb; \Gamma_{2}, p(\vec{x})^{\outp}\}_{u}$}

\AxiomC{$\ns\{\va; \Gamma_{1}, p(\vec{x})^{\inp}\}_{w}\{\vb, \vec{x}; \Gamma_{2}, p(\vec{x})^{\inp}\}_{u}$}
\RightLabel{$\ndr$}
\UnaryInfC{$\ns\{\va, \vec{x}; \Gamma_{1}, p(\vec{x})^{\inp}\}_{w}\{\vb; \Gamma_{2}, p(\vec{x})^{\inp}\}_{u}$}
\RightLabel{$\cut$}
\BinaryInfC{$\ns\{\va, \vec{x}; \Gamma_{1}, p(\vec{x})^{\inp}\}_{w}\{\vb; \Gamma_{2}\}_{u}$}
\RightLabel{$\doms$}
\UnaryInfC{$\ns\{\va; \Gamma_{1}, p(\vec{x})^{\inp}\}_{w}\{\vb; \Gamma_{2}\}_{u}$}
\DisplayProof
\end{center}

1.1.3. If the left premise of $\cut$ is an instance of $\botl$, then the conclusion of $\cut$ is an instance of $\botl$. Therefore, let us assume that the right premise of $\cut$ is an instance of $\botl$. We may also assume that the left premise is not an instance of $\doms$ as this case was already considered above; hence, we suppose that the left premise of $\cut$ is an instance of $\idfo$. If the principal $\bot^{\inp}$ in $\botl$ is the cut formula, then the conclusion is an instance of $\idfo$, and if $\bot^{\inp}$ is not the cut formula, then the conclusion is an instance of $\botl$. 
 
1.2 Suppose that exactly one premise of $\cut$ is an instance of $\idfo$, $\botl$, or $\doms$, and the other premise of $\cut$ is a non-initial rule $(r)$. It follows by assumption 1 above that the principal formula of $(r)$ is not the cut formula (as its complexity will be greater than $0$), meaning that we may apply $\cut$ between an instance of $\idfo$, $\botl$, or $\doms$ to the premise(s) of $(r)$, making sure that the contexts match (i.e. applying the hp-invertibility of $(r)$ on $\idfo$, $\botl$, or $\doms$; see \lem~\ref{lem:hp-invert-prop} and \ref{lem:hp-invert-fo}) when we do so. After the $\cut$ instance, we apply $(r)$ to obtain the desired conclusion. We note that in such a case $h_{1}+h_{2}$ has decreased.

1.3. Suppose that neither premise of $\cut$ is an instance of $\idfo$, $\botl$, or $\doms$. Let $\rone$ and $\rtwo$ be the rules used to derive the left and right premises of $\cut$, respectively. Also, we assume that $\rone$ is a two premise rule and $\rtwo$ is a one premise rule; the other cases are argued in a similar fashion. By assumption 1 above, the complexity of the cut formula is $0$, implying that principal formulae of $\rone$ and $\rtwo$ (which are assumed to be non-initial, and therefore, have a complexity greater than $0$) are not cut formulae. From what has been said, our inferences and $\cut$ must be of the form shown below left, and can be resolved as shown below right (with $h_{1}+h_{2}$ decreased). We apply the hp-invertibility of $\rtwo$ (\lem~\ref{lem:hp-invert-prop} and \ref{lem:hp-invert-fo}) to ensure that the contexts match so that $\cut$ may be applied. Moreover, we note that if $\rtwo$ is subject to a side condition, then $\ns_{3}\{\Gamma_{3}\}_{w}$ will satisfy the side condition as well, showing that $\rtwo$ can indeed be applied after the $\cut$ below.
\begin{flushleft}
\begin{tabular}{c c}
\AxiomC{$\Sigma_{1}\{\Gamma_{1},\phi^{\outp}\}_{w}$}
\AxiomC{$\Sigma_{2}\{\Gamma_{2},\phi^{\outp}\}_{w}$}
\RightLabel{$\rone$}
\BinaryInfC{$\ns\{\Gamma,\phi^{\outp}\}_{w}$}
\AxiomC{$\ns_{3}\{\Gamma_{3},\phi^{\inp}\}_{w}$}
\RightLabel{$\rtwo$}
\UnaryInfC{$\ns\{\Gamma,\phi^{\inp}\}_{w}$}
\RightLabel{$\cut$}
\BinaryInfC{$\ns\{\Gamma\}_{w}$}
\DisplayProof

&

$\leadsto$

\end{tabular}
\end{flushleft}
\begin{flushright}
\AxiomC{$\ns_{1}\{\Gamma_{1},\phi^{\outp}\}_{w}$}
\AxiomC{$\ns_{2}\{\Gamma_{2},\phi^{\outp}\}_{w}$}
\RightLabel{$\rone$}
\BinaryInfC{$\ns\{\Gamma,\phi^{\outp}\}_{w}$}
\RightLabel{\lem~\ref{lem:hp-invert-prop} and \ref{lem:hp-invert-fo}}
\UnaryInfC{$\ns_{3}\{\Gamma_{3},\phi^{\outp}\}_{w}$}

\AxiomC{$\ns_{3}\{\Gamma_{3},\phi^{\inp}\}_{w}$}
\RightLabel{$\cut$}
\BinaryInfC{$\ns_{3}\{\Gamma_{3}\}_{w}$}
\RightLabel{$\rtwo$}
\UnaryInfC{$\ns\{\Gamma\}_{w}$}
\DisplayProof
\end{flushright}

2. Suppose that the complexity of the cut formula is greater than $0$. 

2.1 Suppose that the cut formula is not principal in at least one premise of $\cut$ and let $\rtwo$ be the rule deriving that premise with $\rone$ the rule deriving the other premise (whose principal formula may or may not be the cut formula). Then, similar to case 1.3 above, we apply the hp-invertibility of $\rtwo$ to the conclusion of $\rone$ (\lem~\ref{lem:hp-invert-prop} and \ref{lem:hp-invert-fo}) to ensure the contexts match, then cut with the premise(s) of $\rtwo$, thus decreasing $h_{1} + h_{2}$, and last, apply the $\rtwo$ rule.

2.2 Suppose that the cut formula is principal in both premises of $\cut$. To complete our proof, we make a final case distinction on the main connective of the cut formula, and show how to reduce the $\cut$ instance in each case.

2.2.1. If our cut formula is of the form $\phi \lor \psi$, then the case is as shown below left, and may be resolved as shown below right (where the complexity of the cut formulae has decreased). We note that the case where the cut formula is of the form $\phi \land \psi$ is similar, so we omit it.

\begin{flushleft}
\begin{tabular}{c c}
\AxiomC{$\ns\{\Gamma,\phi^{\outp}, \psi^{\outp}\}_{w}$}
\RightLabel{$\disr$}
\UnaryInfC{$\ns\{\Gamma,\phi \lor \phi^{\outp}\}_{w}$}

\AxiomC{$\ns\{\Gamma,\phi^{\inp}\}_{w}$}
\AxiomC{$\ns\{\Gamma,\phi^{\inp}\}_{w}$}
\RightLabel{$\disl$}
\BinaryInfC{$\ns\{\Gamma,\phi \lor \psi^{\inp}\}_{w}$}

\RightLabel{$\cut$}
\BinaryInfC{$\ns\{\Gamma\}_{w}$}
\DisplayProof

&

$\leadsto$

\end{tabular}
\end{flushleft}
\begin{flushright}
\AxiomC{$\ns\{\Gamma, \phi^{\outp}, \psi^{\outp}\}_{w}$}

\AxiomC{$\ns\{\Gamma, \phi^{\inp}\}_{w}$}
\RightLabel{$\wk$}
\UnaryInfC{$\ns\{\Gamma, \phi^{\inp}, \psi^{\outp}\}_{w}$}
\RightLabel{$\cut$}
\BinaryInfC{$\ns\{\Gamma, \psi^{\outp}\}_{w}$}
\AxiomC{$\ns\{\Gamma, \psi^{\inp}\}_{w}$}
\RightLabel{$\cut$}
\BinaryInfC{$\ns\{\Gamma\}_{w}$}
\DisplayProof
\end{flushright}

2.2.2. We now consider the case where the cut formula is of the form $\phi \imp \psi$, as shown below.
\begin{center}
\begin{tabular}{c c}
$\prf =$

&

\AxiomC{$\ns\{\Gamma, \phi \imp \psi^{\inp}\}_{w}\{\Delta, \phi^{\outp}\}_{u}$}
\AxiomC{$\ns\{\Gamma,  \phi \imp \psi^{\inp}\}_{w}\{\Delta, \psi^{\inp}\}_{u}$}
\RightLabel{$\impl$}
\BinaryInfC{$\ns\{\Gamma, \phi \imp \psi^{\inp}\}_{w}\{\Delta\}_{u}$}
\DisplayProof
\end{tabular}
\end{center}
\begin{center}
\AxiomC{$\ns\{\Gamma, [\emptyset; \phi^{\inp}, \psi^{\outp}]_{v}\}_{w}\{\Delta\}_{u}$}
\RightLabel{$\impr$}
\UnaryInfC{$\ns\{\Gamma, \phi \imp \psi^{\outp}\}_{w}\{\Delta\}_{u}$}

\AxiomC{$\prf$}

\RightLabel{$\cut$}
\BinaryInfC{$\ns\{\Gamma\}_{w}\{\Delta\}_{u}$}
\DisplayProof
\end{center}

 We may resolve the case as shown below. Observe that each $\cut$ in $\prf_{1}$ and $\prf_{2}$ is of height $h_{1} + h_{2} - 1$ and the other two cuts are on formulae of smaller complexity, thus allowing for their elimination by IH. Moreover, we apply the hp-admissibility of $\wk$ (\lem~\ref{lem:wk-rules-admiss}), $\lwr$ (\lem~\ref{lem:lwr-admiss}), and $\mrg$ (\lem~\ref{lem:mrg-admiss}) to aid us on concluding the case. We note that $\lwr$ must be applied a sufficient number (say, $n$) times to shift $\phi \imp \psi^{\outp}$ in the correct component. The final applications of $\cut$ are on formulae of less complexity.
\begin{flushleft}
\begin{tabular}{c c}
$\prf_{1} =$

&

\AxiomC{$\ns\{\Gamma, [\emptyset; \phi^{\inp}, \psi^{\outp}]_{v}\}_{w}\{\Delta\}_{u}$}
\RightLabel{$\impr$}
\UnaryInfC{$\ns\{\Gamma, \phi \imp \psi^{\outp}\}_{w}\{\Delta\}_{u}$}
\RightLabel{$\wk$}
\UnaryInfC{$\ns\{\Gamma, \phi \imp \psi^{\outp}\}_{w}\{\Delta,\phi^{\outp}\}_{u}$}

\AxiomC{$\ns\{\Gamma, \phi \imp \psi^{\inp}\}_{w}\{\Delta, \phi^{\outp}\}_{u}$}

\RightLabel{$\cut$}
\BinaryInfC{$\ns\{\Gamma\}_{w}\{\Delta,\phi^{\outp}\}_{u}$}
\DisplayProof
\end{tabular}
\end{flushleft}

\begin{flushleft}
\begin{tabular}{c c}
$\prf_{2} =$

&

\AxiomC{$\ns\{\Gamma, [\emptyset; \phi^{\inp}, \psi^{\outp}]_{v}\}_{w}\{\Delta\}_{u}$}
\RightLabel{$\impr$}
\UnaryInfC{$\ns\{\Gamma, \phi \imp \psi^{\outp}\}_{w}\{\Delta\}_{u}$}
\RightLabel{$\wk$}
\UnaryInfC{$\ns\{\Gamma, \phi \imp \psi^{\outp}\}_{w}\{\Delta,\psi^{\inp}\}_{u}$}

\AxiomC{$\ns\{\Gamma, \phi \imp \psi^{\inp}\}_{w}\{\Delta, \psi^{\inp}\}_{u}$}

\RightLabel{$\cut$}
\BinaryInfC{$\ns\{\Gamma\}_{w}\{\Delta,\psi^{\inp}\}_{u}$}
\DisplayProof
\end{tabular}
\end{flushleft}

\begin{flushright}
\AxiomC{$\prf_{1}$}

\AxiomC{$\ns\{\Gamma, [\emptyset; \phi^{\inp}, \psi^{\outp}]_{v}\}_{w}\{\Delta\}_{u}$}
\RightLabel{$\impr$}
\UnaryInfC{$\ns\{\Gamma, \phi \imp \psi^{\outp}\}_{w}\{\Delta\}_{u}$}
\RightLabel{$\lwr \times n$}
\UnaryInfC{$\ns\{\Gamma\}_{w}\{\Delta,\phi \imp \psi^{\outp}\}_{u}$}
\RightLabel{\lem~\ref{lem:hp-invert-prop}}
\UnaryInfC{$\ns\{\Gamma\}_{w}\{\Delta,[\emptyset; \phi^{\inp}, \psi^{\outp}]_{v}\}_{u}$}
\RightLabel{$\mrg$}
\UnaryInfC{$\ns\{\Gamma\}_{w}\{\Delta,\phi^{\inp}, \psi^{\outp}\}_{u}$}

%\AxiomC{$\ns\{\Gamma\}_{w}\{\Delta, \phi^{\inp},\psi^{\inp}\}_{u}$}
\AxiomC{$\prf_{2}$}
\RightLabel{$\wk$}
\UnaryInfC{$\ns\{\Gamma\}_{w}\{\Delta, \phi^{\inp},\psi^{\inp}\}_{u}$}

\RightLabel{$\cut$}
\BinaryInfC{$\ns\{\Gamma\}_{w}\{\Delta, \phi^{\inp}\}_{u}$}

\RightLabel{$\cut$}
\BinaryInfC{$\ns\{\Gamma\}_{w}\{\Delta\}_{u}$}
\DisplayProof
\end{flushright}

2.2.3. We now consider the case where the cut formula is of the form $\forall x \phi$ and argue the case where the formula is introduced by $\allli$ and $\allr$, omitting the $\alllii$ and $\allr$ case as it is similar.
\begin{flushleft}
\begin{tabular}{c c}
$\prf = $

&

\AxiomC{$\ns\{\va,z; \Pi\}_{i}\{\Gamma, [y; \phi(y/x)^{\outp}]_{v}\}_{w}\{\vb; \nsc\}_{u}$}
\RightLabel{$\allr$}
\UnaryInfC{$\ns\{\va,z; \Pi\}_{i}\{\Gamma, \forall x \phi^{\outp}\}_{w}\{\vb; \nsc\}_{u}$}
\DisplayProof
\end{tabular}
\end{flushleft}
\begin{center}
\AxiomC{$\prf$}

\AxiomC{$\ns\{\va,z; \Pi\}_{i}\{\Gamma, \forall x \phi^{\inp}\}_{w}\{\vb; \nsc, \phi(z/x)^{\inp}\}_{u}$}
\RightLabel{$\allli$}
\UnaryInfC{$\ns\{\va,z; \Pi\}_{i}\{\Gamma, \forall x \phi^{\inp}\}_{w}\{\vb; \nsc\}_{u}$}

\RightLabel{$\cut$}
\BinaryInfC{$\ns\{\va,z; \Pi\}_{i}\{\Gamma\}_{w}\{\vb; \nsc\}_{u}$}
\DisplayProof
\end{center}

 The case is resolved as shown below. Since the $\cut$ in $\prf_{2}$ is of height $h_{1} + h_{2} - 1$ and the other cut is on a formula of smaller complexity, each of these cuts may be eliminated by IH. As in the previous case, we apply hp-admissible rules and $\lwr$ a sufficient number of times (say, $n$) to shift $\phi(y/x)^{\outp}$ into the correct component. Moreover, we apply $\ndr$ a sufficient number of times (say, $k$) to shift the variable $z$ to the $i$-component, which may then be removed by an application of $\ctrv$.
\begin{flushleft}
\begin{tabular}{c c}
$\prf_{1} = $

&

\AxiomC{$\ns\{\va,z; \Pi\}_{i}\{\Gamma, [y; \phi(y/x)^{\outp}]_{v}\}_{w}\{\vb; \nsc\}_{u}$}
\RightLabel{$\allr$}
\UnaryInfC{$\ns\{\va,z; \Pi\}_{i}\{\Gamma, \forall x \phi^{\outp}\}_{w}\{\vb; \nsc\}_{u}$}
\RightLabel{$\wk$}
\UnaryInfC{$\ns\{\va,z; \Pi\}_{i}\{\Gamma, \forall x \phi^{\outp}\}_{w}\{\vb; \nsc,\phi(z/x)^{\inp}\}_{u}$}
\DisplayProof
\end{tabular}
\end{flushleft}

\begin{flushleft}
\begin{tabular}{c c}
$\prf_{2} = $

&

\AxiomC{$\prf_{1}$}
\AxiomC{$\ns\{\va,z; \Pi\}_{i}\{\Gamma, \forall x \phi^{\inp}\}_{w}\{\vb; \nsc, \phi(z/x)^{\inp}\}_{u}$}

\RightLabel{$\cut$}
\BinaryInfC{$\ns\{\va,z; \Pi\}_{i}\{\Gamma\}_{w}\{\vb; \nsc,\phi(z/x)^{\inp}\}_{u}$}
\DisplayProof
\end{tabular}
\end{flushleft}
\begin{flushright}
\AxiomC{$\ns\{\va,z; \Pi\}_{i}\{\Gamma, [y; \phi(y/x)^{\outp}]_{v}\}_{w}\{\vb; \nsc\}_{u}$}
\RightLabel{$\allr$}
\UnaryInfC{$\ns\{\va,z; \Pi\}_{i}\{\Gamma, \forall x \phi^{\outp}\}_{w}\{\vb; \nsc\}_{u}$}
\RightLabel{$\lwr \times n$}
\UnaryInfC{$\ns\{\va,z; \Pi\}_{i}\{\Gamma\}_{w}\{\vb; \nsc,\forall x \phi^{\outp}\}_{u}$}
\RightLabel{\lem~\ref{lem:hp-invert-fo}}
\UnaryInfC{$\ns\{\va,z; \Pi\}_{i}\{\Gamma\}_{w}\{\vb; \nsc, [y; \phi(y/x)^{\outp}]_{v}\}_{u}$}
\RightLabel{$\mrg$}
\UnaryInfC{$\ns\{\va,z; \Pi\}_{i}\{\Gamma\}_{w}\{\vb,y; \nsc, \phi(y/x)^{\outp}\}_{u}$}
\RightLabel{$\psub$}
\UnaryInfC{$\ns\{\va,z; \Pi\}_{i}\{\Gamma\}_{w}\{\vb,z; \nsc, \phi(z/x)^{\outp}\}_{u}$}
\RightLabel{$\ndr \times k$}
\UnaryInfC{$\ns\{\va,z,z; \Pi\}_{i}\{\Gamma\}_{w}\{\vb; \nsc, \phi(z/x)^{\outp}\}_{u}$}
\RightLabel{$\ctrv$}
\UnaryInfC{$\ns\{\va,z; \Pi\}_{i}\{\Gamma\}_{w}\{\vb; \nsc, \phi(z/x)^{\outp}\}_{u}$}

\AxiomC{$\prf_{2}$}

\RightLabel{$\cut$}
\BinaryInfC{$\ns\{\va,z; \Pi\}_{i}\{\Gamma\}_{w}\{\Delta\}_{u}$}
\DisplayProof
\end{flushright}

2.2.4. Let us consider the case where the cut formula is of the form $\exists x \phi$. We show the case where $\existsri$ derives the principal formulae in the left premise of $\cut$ and omit the $\existsrii$ case as it is similar.
\begin{center}
\AxiomC{$\ns\{\va; \Gamma, \exists x \phi^{\outp}, \phi(z/x)^{\outp}\}_{w}$}
\RightLabel{$\existsri$}
\UnaryInfC{$\ns\{\vb,z; \nsc\}_{u}\{\va; \Gamma, \exists x \phi^{\outp}\}_{w}$}

\AxiomC{$\ns\{\vb,z; \nsc\}_{u}\{\va,y; \Gamma, \phi(y/x)^{\inp}\}_{w}$}
\RightLabel{$\existsl$}
\UnaryInfC{$\ns\{\vb,z; \nsc\}_{u}\{\va; \Gamma, \exists x \phi^{\inp}\}_{w}$}

\RightLabel{$\cut$}
\BinaryInfC{$\ns\{\vb,z; \nsc\}_{u}\{\va; \Gamma\}_{w}$}
\DisplayProof
\end{center}
 The case is resolved as shown below bottom. Since the $\cut$ in $\prf_{2}$ is of height $h_{1} + h_{2} - 1$ and the cut formula of the other cut is of smaller complexity, each of these cuts may be eliminated by IH. As in the previous case, we apply $\nd$ a sufficient number (say, $n$) of times to shift $z$ to the $i$-component, where it may be removed by an application of $\ctrv$.
 \begin{flushleft}
\begin{tabular}{c c}
$\prf_{1} =$

&

\AxiomC{$\ns\{\vb,z; \nsc\}_{u}\{\va,y; \Gamma, \phi(y/x)^{\inp}\}_{w}$}
\RightLabel{$\existsl$}
\UnaryInfC{$\ns\{\vb,z; \nsc\}_{u}\{\va; \Gamma, \exists x \phi^{\inp}\}_{w}$}
\RightLabel{$\wk$}
\UnaryInfC{$\ns\{\vb,z; \nsc\}_{u}\{\va; \Gamma, \exists x \phi^{\inp}, \phi(z/x)^{\outp}\}_{w}$}
\DisplayProof
\end{tabular}
\end{flushleft}
\begin{flushleft}
\begin{tabular}{c c}
$\prf_{2} =$

&

\AxiomC{$\ns\{\vb,z; \nsc\}_{u}\{\va; \Gamma, \exists x \phi^{\outp}, \phi(z/x)^{\outp}\}_{w}$}

\AxiomC{$\prf_{1}$}

\RightLabel{$\cut$}
\BinaryInfC{$\ns\{\vb,z; \nsc\}_{u}\{\va; \Gamma, \phi(z/x)^{\outp}\}_{w}$}
%\RightLabel{$\wkv$}
%\UnaryInfC{$\ns\{\va,z; \Gamma, \phi(z/x)^{\outp}\}_{w}$}
\DisplayProof
\end{tabular}
\end{flushleft}
\begin{center}
\AxiomC{$\prf_{2}$}

\AxiomC{$\ns\{\vb,z; \nsc\}_{u}\{\va,y; \Gamma, \phi(y/x)^{\inp}\}_{w}$}
\RightLabel{$\psub$}
\UnaryInfC{$\ns\{\vb,z; \nsc\}_{u}\{\va,z; \Gamma, \phi(z/x)^{\inp}\}_{w}$}
\RightLabel{$\ndr \times n$}
\UnaryInfC{$\ns\{\vb,z,z; \nsc\}_{u}\{\va; \Gamma, \phi(z/x)^{\inp}\}_{w}$}
\RightLabel{$\ctrv$}
\UnaryInfC{$\ns\{\vb,z; \nsc\}_{u}\{\va; \Gamma, \phi(z/x)^{\inp}\}_{w}$}
\RightLabel{$\cut$}
\BinaryInfC{$\ns\{\vb,z; \nsc\}_{u}\{\Gamma\}_{w}$}
\DisplayProof
\end{center}
 This concludes the proof of the cut-elimination theorem.
\end{proof}

 Although we have syntactic cut-elimination for the intuitionistic systems, it is not clear how to prove such a theorem for the G\"odel-Dummett calculi. One issue is that the $\lin$ rule appears to resist permutations with $\cut$. This is in spite of the fact that $\cut$ is not required for completeness in each of these systems by the cut-free completeness theorem (\thm~\ref{thm:cut-free-comp}). As a case in point, let us consider the case where $\cut$ is applied between $\lin$ and a unary rule $(r)$ from one of our nested calculi, as shown below:
\begin{flushleft}
\begin{tabular}{c c c}
$\prf$

&

$=$

&

\AxiomC{$\Sigma\{[\nsb,\phi^{\outp},[\nsc]_{u}]_{v}\}_{w}$}
\AxiomC{$\Sigma\{[\nsc,[\nsb,\phi^{\outp}]_{v}]_{u}\}_{w}$}
\RightLabel{$\lin$}
\BinaryInfC{$\Sigma\{[\nsb,\phi^{\outp}]_{v},[\nsc]_{u}\}_{w}$}
\DisplayProof
\end{tabular}
\end{flushleft}
\begin{flushright}
\AxiomC{$\ns'$}
\RightLabel{$(r)$}
\UnaryInfC{$\Sigma\{[\nsb, \phi^{\inp}]_{v},[\nsc]_{u}\}_{w}$}
\AxiomC{$\prf$}
\RightLabel{$\cut$}
\BinaryInfC{$\Sigma\{[\nsb]_{v},[\nsc]_{u}\}_{w}$}
\DisplayProof
\end{flushright}

%\medskip

 If we could show the hp-invertibility of $\lin$, indicated by $\star$ below, then the following would demonstrate how to eliminate $\cut$ in the case above.
\begin{center}
\begin{tabular}{c c @{\hskip 3em} c c}
$\prf_{1} =$

&

\AxiomC{$\ns'$}
\RightLabel{$(r)$}
\UnaryInfC{$\Sigma\{[\nsb, \phi^{\inp}]_{v},[\nsc]_{u}\}_{w}$}
\RightLabel{$\star$}
\UnaryInfC{$\Sigma\{[\nsb, \phi^{\inp},[\nsc]_{u}]_{v}\}_{w}$}
\DisplayProof

&

$\prf_{2} =$

&

\AxiomC{$\ns'$}
\RightLabel{$(r)$}
\UnaryInfC{$\Sigma\{[\nsb, \phi^{\inp}]_{v},[\nsc]_{u}\}_{w}$}
\RightLabel{$\star$}
\UnaryInfC{$\Sigma\{[\nsc,[\nsb, \phi^{\inp}]_{v}]_{u}\}_{w}$}
\DisplayProof
\end{tabular}
\end{center}
\begin{center}
\AxiomC{$\prf_{1}$}
\AxiomC{$\Sigma\{[\nsb,\phi^{\outp},[\nsc]_{u}]_{v}\}_{w}$}
\RightLabel{$\cut$}
\BinaryInfC{$\Sigma\{[\nsb,[\nsc]_{u}]_{v}\}_{w}$}
\AxiomC{$\prf_{2}$}
\AxiomC{$\Sigma\{[\nsc,[\nsb,\phi^{\outp}]_{v}]_{u}\}_{w}$}
\RightLabel{$\cut$}
\BinaryInfC{$\Sigma\{[\nsc,[\nsb]_{v}]_{u}\}_{w}$}
\RightLabel{$\lin$}
\BinaryInfC{$\Sigma\{[\nsb]_{v},[\nsc]_{u}\}_{w}$}
\DisplayProof
\end{center}

 However, it is not clear if the $\lin$ rule is hp-invertible, thus obstructing the above cut-elimination strategy. %It is conceivable that our formulation of $\cut$ is \emph{too strict}, and perhaps a more liberal form of $\cut$ would suffice, or perhaps a stronger measure used in the induction proof (beyond the complexity of the cut formula $\phi$ and sum $h_{1} + h_{2}$ of the heights) is required. 
 We leave the question of syntactic cut-elimination for $\ngd$, $\nndl$, and $\ncdl$ open and defer the problem to future work.

% However, we cannot apply the above cut-elimination strategy as $\lin$ has not been proven hp-invertible. We leave the question of syntactic cut-elimination for $\ngd$, $\nndl$, and $\ncdl$ open and defer the problem to future work.

%Conclusion
\section{Concluding Remarks and Possible Extensions}
\label{sec:conclusion}

 In this paper, we gave a unified nested sequent presentation of propositional and first-order intuitionistic and G\"odel-Dummett logics. We showed how to capture both non-constant and constant domain reasoning by means of reachability rules, which relied on an extension of the nested sequent formalism that included signatures in nested sequents. In addition, we defined a novel structural rule $\lin$, which captures the linearity property of Kripke frames for G\"odel-Dummett logics. Our analytic systems were shown to possess a variety of (hp-)admissibility and (hp-)invertibility properties, are sound and cut-free complete, and syntactic cut-elimination was shown for $\nipc$, $\nnd$, and $\ncd$. As such, our intuitionistic systems serve as viable base systems for the development of a general nested proof theory for intermediate logics.
 
 In future work, we aim to consider further extensions of $\nipc$, $\nnd$, and $\ncd$ to capture other intermediate logics within the formalism of nested sequents. In particular, we aim to investigate the nested proof theory of intermediate logics whose frames satisfy \emph{disjunctive linear conditions (DLC)} or \emph{disjunctive branching conditions (DBC)}. We define a DLC to be a formula of the form
$$
\bigwedge_{1 \leq i \leq n} w_{i} \leq w_{i+1} \rightarrow C_{1} \quad \text{ such that } \quad C_{1} = \bigvee_{1 \leq j \leq k} A_{j}
$$
 where $A_{j} \in \{w_{i+1} \leq w_{i} \ \vert \ 1 \leq i \leq n\}$. The antecedent of a DLC consists of a linear sequence of related worlds, and the consequent contains a disjunction of relations, each of which stipulates that a successor world relates to its predecessor in the linear sequence of the antecedent. We define a DBC to be a formula of the form
$$
\bigwedge_{1 \leq i \leq n} w \leq w_{i} \rightarrow C_{2} \quad \text{ such that } \quad C_{2} = \bigvee_{1 \leq j \leq k} A_{j}
$$
 where $A_{j} \in \{u = v, u \leq v' \ \vert \ u,v \in \{w_{1}, \ldots, w_{n}\}, v' \in \{w,w_{1}, \ldots, w_{n}\}\}$. The antecedent of a DBC consists of a tree of depth one with a root world $w$ that relates to $n$ children worlds, and the consequent consists of a disjunction of equations identifying children worlds and relations that relate worlds occurring in the antecedent. %We refer to $C_{1}$ as a \emph{DLC consequent} and $C_{2}$ as a \emph{DBC consequent}.
 
 Such frame conditions appear to be readily convertible into nested structural rules. For example, each DLC appears to correspond to a structural rule, which we dub $\dlc$, of the following form:

\medskip
 
\begin{center}
\AxiomC{$\Big\{\nsa\{\nsb_{1}, [\ldots [\nsb_{i},\nsb_{i+1},[\nsb_{i+2},\ldots [\nsb_{n}]_{w_{n}} \ldots]_{w_{i+2}} ]_{w_{i}} \ldots ]_{w_{2}} \}_{w_{1}} \ \Big\vert \  w_{i+1} \leq w_{i} \in C_{1}  \Big\}$}
\UnaryInfC{$\nsa\{\nsb_{1}, [ \ldots [\nsb_{n}]_{w_{n}} \ldots]_{w_{2}} \}_{w_{1}}$}
\DisplayProof
\end{center}
 
\medskip

 The conclusion of a $\dlc$ contains a nested linear sequence of components of depth $n$ corresponding to the (linear) antecedent of a DLC. Each premise `merges' a parent $w_{i}$-component with its child $w_{i+1}$-component \iffi $w_{i+1} \leq w_{i}$ occurs in the consequent of the DLC. 
 
 Known intermediate logics appear to admit a nested sequent characterization by means of the above rules. For example, the intermediate logic of \emph{bounded-depth 2} ($\bdtwo$) (see~\citet{GabSheSkv09}) is obtained from \ipc \ by imposing the following frame condition on \ipc-frames: for each world $w$, $u$, and $v$, if $w \leq u \leq v$, then $u \leq w$ or $v \leq u$. We observe that this condition is in fact a DLC, and thus, we may transform the condition into the following $\dlc$ structural rule:
 
\medskip

\begin{center}
\AxiomC{$\nsa\{\nsb_{1}, \nsb_{2}, [\nsb_{3}]_{w_{3}}\}_{w_{1}}$}
\AxiomC{$\nsa\{\nsb_{1}, [\nsb_{2}, \nsb_{3}]_{w_{2}}\}_{w_{1}}$}
\RightLabel{$(bd_{2})$}
\BinaryInfC{$\nsa\{\nsb_{1}, [\nsb_{2}, [\nsb_{3}]_{w_{3}}]_{w_{2}}\}_{w_{1}}$}
\DisplayProof
\end{center}

\medskip

 We note that the logic $\bdtwo$ may be obtained from \ipc \ by extending \ipc's axiomatization with the $\bdtwo$ axiom $\phi \lor (\phi \imp (\psi \lor (\psi \imp \bot)))$. Indeed, one can show that the above structural rule derives this axiom if we add it to our nested calculus $\nipc$. Therefore, we should be able to provide a nested calculus for \emph{Smetanich logic} (see~\citet{ChaZak97}) as well (which is axiomatized by adding the $\bdtwo$ axiom to the axioms of \gd), by extending $\ngd$ with the $(bd_{2})$ rule above.
 
 Each DBC condition also appears to correspond to a nested structural rule, which we dub $\dbc$, of the following form:
 
\medskip

\begin{center}
\AxiomC{$\nsd_{1}, \ldots, \nsd_{k}$}
\UnaryInfC{$\nsa\{\nsb, [\nsc_{1}]_{w_{1}}, \ldots, [\nsc_{j}]_{w_{j}}, \ldots, [\nsc_{k}]_{w_{k}}, \ldots, [\nsc_{n}]_{w_{n}}\}_{w}$}
\DisplayProof
\end{center}

\medskip

 The premises fall into three distinct classes depending on the relations that occur in the consequent $C_{2}$ of the given DBC. We let $1 \leq m \leq k$ and define each premise accordingly:

\begin{enumerate}

\item If the $m^{th}$ disjunct of $C_{2}$ is the relation $w_{j} \leq w$, then the premise is:
$$
\nsd_{m} = \nsa\{\nsb, \nsc_{j}, [\nsc_{1}]_{w_{1}}, \ldots, [\nsc_{k}]_{w_{k}}, \ldots, [\nsc_{n}]_{w_{n}}\}_{w}
$$
 where the $w_{j}$-component is merged into the $w$-component.
 
\item If the $m^{th}$ disjunct of $C_{2}$ is the relation $w_{j} \leq w_{k}$, then the premise is:
$$
\nsd_{m} = \nsa\{\nsb, [\nsc_{1}]_{w_{1}}, \ldots, [\nsc_{j}, [\nsc_{k}]_{w_{k}}]_{w_{j}}, , \ldots, [\nsc_{n}]_{w_{n}}\}_{w}
$$
 where the $w_{k}$-component is placed within the $w_{j}$-component. 
 
\item If the $m^{th}$ disjunct of $C_{2}$ is the equation $w_{j} = w_{k}$, then the premise is:
$$
\nsd_{m} = \nsa\{\nsb, [\nsc_{1}]_{w_{1}}, \ldots, [\nsc_{j},\nsc_{k}]_{w_{j}}, \ldots, [\nsc_{n}]_{w_{n}}\}_{w}
$$
 where the $w_{k}$-component and $w_{j}$-component are contracted.

\end{enumerate}

 We observe that the connectivity condition imposed on \gd-frames (see \dfn~\ref{def:frame-model}) falls within the class of DBCs. In fact, the linearity rule $\lin$ serves as an example of a $\dbc$ structural rule, having two premises determined by case 2 in the three cases described above. Moreover, as classical logic is characterizable over intuitionistic frames satisfying \emph{symmetry} (i.e. for any two worlds $w$ and $u$, if $w \leq u$, then $u \leq w$), we could transform our intuitionistic systems into classical systems via the addition of the following $\dbc$ structural rule, whose premise is obtained from case 1 above.

\begin{center}
\AxiomC{$\nsa\{\nsb_{1}, \nsb_{2}\}_{w}$}
\RightLabel{$(sym)$}
\UnaryInfC{$\nsa\{\nsb_{1}, [\nsb_{2}]_{u}\}_{w}$}
\DisplayProof
\end{center}

 Indeed, one can derive the law of the excluded middle $\phi \lor (\phi \imp \bot)$ by adding the above rule to $\nipc$, $\nnd$, or $\ncd$, thus yielding a nested system for propositional or first-order classical logic.
 
 The DLC and DBC conditions are special in that the structural rules they generate naturally correspond to reasoning within tree structures. Moreover, it is conceivable that our cut-free completeness theorem (\thm~\ref{thm:cut-free-comp}) could be adapted to cover intermediate logics satisfying DLC and DBC conditions, or that terminating proof-search algorithms could be defined with such rules (in the propositional setting). %We aim to investigate these open problems in future work.
 
 %Moreover, since all such rules reduce the structure of the conclusion in some way (e.g. the $\dlc$ rules reduce the length of the nested linear path) it is conceivable that our cut-free completeness theorem (\thm~\ref{thm:cut-free-comp}) could be adapted to cover intermediate logics satisfying DLC and DBC conditions. Moreover, such rules also appear suitable for defining proof-search algorithms. Based on the above discussion, we put forth the following conjecture:

%\begin{conjecture}
%Every intermediate logic characterized by imposing DLC or DBC conditions on \ipc-, \nd-, or \cd-frames admits a cut-free nested sequent presentation.
%\end{conjecture}

%Funding
\section*{Funding} Work supported by the European Research Council (ERC) Consolidator Grant 771779. %\textit{A Grand Unified Theory of Decidability in Logic-Based Knowledge Representation} (DeciGUT).

%%%Bibliography
\bibliographystyle{apacite}
\bibliography{bibliography}

%%%Appendices
\appendix

%%Appendix A
\section{Cut-free Completeness Theorem}\label{app:completeness}

 We let $\ncalc \in \{\nipc, \nnd, \ncd, \ngd, \nndl, \ncdl\}$ and prove the cut-free completeness of $\ncalc$ by extracting a counter-model from failed proof-search. First, we introduce useful terminology. We define a \emph{pseudo-derivation} to be an object constructed by applying rules from $\ncalc$ bottom-up (potentially an infinite number of times) to an arbitrary nested sequent (which serves as the \emph{conclusion} of the pseudo-derivation). Note that a derivation is a pseudo-derivation with all top sequents axiomatic. We define a \emph{branch} $\branch$ in a pseudo-derivation to be a path of nested sequents satisfying: (1) the conclusion of the pseudo-derivation is the $1^{st}$ element of the path, (2) if a nested sequent in the pseudo-derivation is the $n^{th}$ element in the path and is not an instance of $\id$ or $\botl$, then one of its premises is the $(n+1)^{th}$ element of the path.
 %We define a \emph{branch} $\branch$ to be a sequence of nested sequents from the conclusion of a pseudo-derivation to a top sequent. 
 For a nested sequent $\ns$, we use the notation $w : A^{\io} \in \ns$ to indicate that $\ns\{A^{\io}\}_{w}$ with $\io \in \{\inp,\outp\}$, the notation $x : \va(w) \in \ns$ to indicate that the variable $x$ occurs in the signature $\va$ of the $w$-component of $\ns$, and the notation $\lab(\ns)$ to be the set of all labels occurring in $\ns$. We define a nested sequent $\ns$ to be \emph{linear} \iffi for every $w,u \in \lab(\ns)$, either $\rable{w}{u}$ or $\rable{u}{w}$. The following lemmas are useful in our proof:

\begin{lemma}\label{lem:rable-preserved-up} Let $\ncalc \in \{\nipc, \nnd, \ncd, \ngd, \nndl, \ncdl\}$.\\
\begin{enumerate}

\item If $\rable{w}{u}$ holds for the conclusion of a rule $(r)$ in $\ncalc$, then $\rable{w}{u}$ holds for the premises of $(r)$;

\item If $w : p(x_{1},\ldots,x_{n}) \in \ns$ and $\ns$ is the conclusion of a rule $(r)$ in $\ncalc$ with $\ns_{1}, \ldots, \ns_{n}$ the premises of $(r)$, then for $1 \leq i \leq n$, $w : p(x_{1},\ldots,x_{n}) \in \ns_{i}$ holds;

\item If $x : \va(w) \in \ns$ and $\ns$ is the conclusion of a rule $(r)$ in $\ncalc$ with $\ns_{1}, \ldots, \ns_{n}$ the premises of $(r)$, then for $1 \leq i \leq n$, $x$ occurs in the signature of the $w$-component of $\ns_{i}$.

\end{enumerate}
\end{lemma}

\begin{proof} By inspection of the rules of $\ncalc$.
\end{proof}

 In essence, the above lemma states that the $\rable{}{}$ relation is preserved bottom-up in rule applications and the position of atomic formulae and variables is bottom-up fixed.
 
\begin{lemma}\label{lem:fresh-delete}
Let $\ncalc \in \{\nnd, \ncd, \nndl, \ncdl\}$. If $\va,y; \nsa$ is derivable in $\ncalc$ with $y$ not occurring in $\va; \nsa$, then $\va; \nsa$ is derivable in $\ncalc$.
\end{lemma}

\begin{proof} The lemma is shown by induction on the height of the given derivation. We prove the lemma for $\nndl$ as the remaining cases are similar.

\textit{Base case.} The $\botl$ case is simple, so we show the $\idfo$ case. Suppose we have an instance of $\idfo$ as shown below left, where the variable $y$ does not occur anywhere else in the nested sequent. Then, as shown below right, $y$ may be deleted as this is still an instance of $\idfo$.
\begin{center}
\begin{tabular}{c c c}
\AxiomC{}
\RightLabel{$\idfo$}
\UnaryInfC{$\va,y; \nsa \hol p(\vec{x})^{\inp} \hor_{w} \hol p(\vec{x})^{\outp} \hor_{u}$}
\DisplayProof

&

$\leadsto$

&

\AxiomC{}
\RightLabel{$\idfo$}
\UnaryInfC{$\va; \nsa \hol p(\vec{x})^{\inp} \hor_{w} \hol p(\vec{x})^{\outp} \hor_{u}$}
\DisplayProof
\end{tabular}
\end{center}

\textit{Inductive step.} Most cases of the induction step are trivial with the exception of the $\existsri$ and $\allli$ cases. We show how the $\allli$ case is resolved and note that the $\existsri$ case is similar.

 Suppose that the variable $y$ is active in an $\allli$ inference as shown below left. By substituting the $\allli$ inference for an $\alllii$ inference, as shown below right, we obtain the desired conclusion.
\begin{center}
\begin{tabular}{c c c}
\AxiomC{$\va,y; \nsa \hol \forall x \phi^{\inp} \hor_{w} \hol  \nsb, \phi(y/x)^{\inp} \hor_{u}$}
\RightLabel{$\allli$}
\UnaryInfC{$\va,y; \nsa \hol \forall x \phi^{\inp} \hor_{w} \hol  \nsb \hor_{u}$}
\DisplayProof

$\leadsto$

&

\AxiomC{$\va,y; \nsa \hol \forall x \phi^{\inp} \hor_{w} \hol  \nsb, \phi(y/x)^{\inp} \hor_{u}$}
\RightLabel{$\alllii$}
\UnaryInfC{$\va; \nsa \hol \forall x \phi^{\inp} \hor_{w} \hol  \nsb \hor_{u}$}
\DisplayProof
\end{tabular}
\end{center}
\end{proof}

 We now prove our cut-free completeness result. We take a nested sequent $\vec{x}; \phi(\vec{x})$ with $\vec{x}$ all free variables in $\phi(\vec{x})$ and a fresh label $z$ that does not occur in $\vec{x}; \phi(\vec{x})$, and apply rules from $\ncalc$ bottom-up on $z,\vec{x}; \phi(\vec{x})$ with the goal of finding a proof thereof. If a proof is found, then by \lem~\ref{lem:fresh-delete}, we know that $\vec{x}; \phi(\vec{x})$ is derivable in $\ncalc$, and if a proof is not found, then we provide a counter-model for $z,\vec{x}; \phi(\vec{x})$, which is also a counter-model for $\vec{x}; \phi(\vec{x})$ by \dfn~\ref{def:sequent-semantics}. We note that the inclusion of the fresh variable $z$ in our input is required to ensure that the domains of the counter-model are non-empty, which explains its presence. We refer to this special variable $z$ as the \emph{starting variable} and fix it throughout the course of the proof.

\begin{customthm}{\ref{thm:cut-free-comp}}
Let $\mathrm{L} \in \{\ipc, \nd, \cd, \gd, \ndl, \cdl\}$ and $\vec{x}; \phi(\vec{x})$ be a nested sequent with $\vec{x}$ all free variables in $\phi(\vec{x})$. If a nested sequent $\vec{x}; \phi(\vec{x})$ is $\mathrm{L}$-valid, then it is derivable in $\ncalc_{\mathrm{L}}$.
\end{customthm}

\begin{proof} We prove the theorem for $\nndl$ as the other cases are similar. Let $\vec{x}; \phi(\vec{x})$ be a nested sequent with $\vec{x}$ all free variables in $\phi(\vec{x})$ and let $z$ be our starting variable. We define a proof-search algorithm $\prove$ that applies rules from $\nndl$ bottom-up, generating a pseudo-derivation of $z,\vec{x}; \phi(\vec{x})$. If a proof is found, then by \lem~\ref{lem:fresh-delete}, we know that $\vec{x}; \phi(\vec{x})$ is derivable in $\nndl$, and if a proof is not found, then we construct a \ndl-model $M$ witnessing the \ndl-invalidity of $\vec{x}; \phi(\vec{x})$. Let us now describe the proof-search algorithm $\prove$.\\

\noindent
$\prove$. We take $z,\vec{x}; \phi(\vec{x})$ as input and continue to the next step.\\

$\idfo$ and $\botl$. Let $\branch_{1}, \ldots, \branch_{n}$ be all branches in the pseudo-derivation so far constructed with $\ns_{1}, \ldots, \ns_{n}$ the top sequents of each branch, respectively. For each branch $\branch_{i}$ such that $\ns_{i}$ is of the form $\idfo$ or $\botl$, halt $\prove$. If $\prove$ has halted on $\branch_{i}$ for each $1 \leq i \leq n$, then return $\success$ as we have found a proof of the input. If $\prove$ has not halted on $\branch_{i}$ for each $1 \leq i \leq n$, then let $\branch_{j_{1}}, \ldots, \branch_{j_{k}}$ be the remaining branches for which $\prove$ did not halt. For each branch, copy the top sequent above itself and continue to the next step.

\medskip

$\doms$. Let $\branch_{1}, \ldots, \branch_{n}$ be all branches in the pseudo-derivation so far constructed with $\ns_{1}, \ldots, \ns_{n}$ the top sequents of each branch, respectively. We successively consider each $\branch_{i}$ for $1 \leq i \leq n$, performing a set of operations which extend the branch with $\disl$ rules bottom-up. Suppose that $\branch_{1},\ldots,\branch_{k}$ have already been processed, so that $\branch_{k+1}$ is the current branch under consideration. Let $\ns_{k+1}$ be of the form
$$
\ns\{\va_{1}; p_{1,1}(\vec{x}_{1,1})^{\inp}, \ldots, p_{1,m_{1}}(\vec{x}_{1,m_{1}})^{\inp}\}_{w_{1}} \cdots \{\va_{\ell}; p_{\ell,1}(\vec{x}_{\ell,1})^{\inp}, \ldots, p_{\ell,m_{\ell}}(\vec{x}_{\ell,m_{\ell}})^{\inp}\}_{w_{\ell}}
$$
 with $p_{i,j}(\vec{x}_{i,j})^{\inp}$ all atomic input formulae in $\ns_{k+1}$ for $1 \leq i \leq \ell$ and $1 \leq j \leq m_{i}$. We successively consider each atomic input formula, and apply the $\doms$ rule bottom-up in each case. This yields a new branch extending $\branch_{k+1}$, whose top sequent incorporates the variables from all atomic input formulae into the signatures of their respective components. Once each branch $\branch_{i}$ has been processed for each $1 \leq i \leq n$, we continue to the next step.

\medskip

$\disl$. Let $\branch_{1}, \ldots, \branch_{n}$ be all branches in the pseudo-derivation so far constructed with $\ns_{1}, \ldots, \ns_{n}$ the top sequents of each branch, respectively. We successively consider each $\branch_{i}$ for $1 \leq i \leq n$, performing a set of operations which extend the branch with $\disl$ rules bottom-up. Suppose that $\branch_{1},\ldots,\branch_{k}$ have already been processed, so that $\branch_{k+1}$ is the current branch under consideration. Let $\ns_{k+1}$ be of the form
$$
\ns\{\phi_{1} \lor \psi_{1}^{\inp}\}_{w_{1}} \cdots \{\phi_{m} \lor \psi_{m}^{\inp}\}_{w_{m}}
$$
 with $\phi_{i} \lor \psi_{i}^{\inp}$ all disjunctive input formulae in $\ns_{k+1}$. We successively consider each disjunctive input formula, and apply the $\disl$ rule bottom-up in each case. This yields $2^{m}$ new branches extending $\branch_{k+1}$, each having a top sequent of the form
$$
\ns\{\chi_{1}^{\inp}\}_{w_{1}} \cdots \{\chi_{n}^{\inp}\}_{w_{n}}
$$
 where $\chi_{i} \in \{\phi_{i},\psi_{i}\}$ for $1 \leq i \leq n$. Once each branch $\branch_{i}$ has been processed for each $1 \leq i \leq n$, we continue to the next step.
 
\medskip 
 
$\disr$. Let $\branch_{1}, \ldots, \branch_{n}$ be all branches in the pseudo-derivation so far constructed with $\ns_{1}, \ldots, \ns_{n}$ the top sequents of each branch, respectively. We successively consider each $\branch_{i}$ for $1 \leq i \leq n$, performing a set of operations which extend the branch with $\disr$ rules bottom-up. Suppose that $\branch_{1},\ldots,\branch_{k}$ have already been processed, so that $\branch_{k+1}$ is the current branch under consideration. Let $\ns_{k+1}$ be of the form
$$
\ns\{\phi_{1} \lor \psi_{1}^{\outp}\}_{w_{1}} \cdots \{\phi_{n} \lor \psi_{n}^{\outp}\}_{w_{n}}
$$
 with $\phi_{i} \lor \psi_{i}^{\outp}$ all disjunctive output formulae in $\ns_{k+1}$. We successively consider each disjunctive output formula, and apply the $\disr$ rule bottom-up in each case. This extends $\branch_{k+1}$, so that it now has a top sequent of the form
$$
\ns\{\phi_{1}^{\outp}, \psi_{1}^{\outp}\}_{w_{1}} \cdots \{\phi_{n}^{\outp}, \psi_{n}^{\outp}\}_{w_{n}}
$$
 Once each branch $\branch_{i}$ has been processed for each $1 \leq i \leq n$, we continue to the next step.
 
\medskip

$\conl$. Similar to the $\disr$ case above.

\medskip

$\conr$. Similar to the $\disl$ case above.

\medskip

$\impl$. Let $\branch_{1}, \ldots, \branch_{n}$ be all branches in the pseudo-derivation so far constructed with $\ns_{1}, \ldots, \ns_{n}$ the top sequents of each branch, respectively. We successively consider each $\branch_{i}$ for $1 \leq i \leq n$, performing a set of operations which extend the branch with $\impl$ rules bottom-up. Suppose that $\branch_{1},\ldots,\branch_{k}$ have already been processed, so that $\branch_{k+1}$ is the current branch under consideration. Let $\ns_{k+1}$ be of the form
$$
\ns\{\phi_{1} \imp \psi_{1}^{\inp}\}_{w_{1}} \cdots \{\phi_{m} \imp \psi_{m}^{\inp}\}_{w_{m}}
$$
 with $\phi_{i} \imp \psi_{i}^{\inp}$ all implicational input formulae in $\ns_{k+1}$. We successively consider each implicational input formula, and apply the $\impl$ rule bottom-up in each case. Suppose we have already processed $\phi_{1} \imp \psi_{1}^{\inp}, \ldots, \phi_{\ell} \imp \psi_{\ell}^{\inp}$, so that $\phi_{\ell+1} \imp \psi_{\ell+1}^{\inp}$ is the current implicational input formula under consideration. For each label $u$ occurring in the top nested sequent of the branches extending $\branch_{k+1}$ such that $\rable{w_{\ell+1}}{u}$, successively apply the $\impl$ rule bottom-up. Once each branch $\branch_{i}$ has been processed for each $1 \leq i \leq n$, we continue to the next step.

\medskip

$\impr$. Let $\branch_{1}, \ldots, \branch_{n}$ be all branches in the pseudo-derivation so far constructed with $\ns_{1}, \ldots, \ns_{n}$ the top sequents of each branch, respectively. We successively consider each $\branch_{i}$ for $1 \leq i \leq n$, performing a set of operations which extend the branch with $\impr$ rules bottom-up. Suppose that $\branch_{1},\ldots,\branch_{k}$ have already been processed, so that $\branch_{k+1}$ is the current branch under consideration. Let $\ns_{k+1}$ be of the form
$$
\ns\{\phi_{1} \imp \psi_{1}^{\outp}\}_{w_{1}} \cdots \{\phi_{m} \imp \psi_{m}^{\outp}\}_{w_{m}}
$$
 with $\phi_{i} \imp \psi_{i}^{\outp}$ all implicational output formulae in $\ns_{k+1}$. We successively consider each implicational output formula, and apply the $\impr$ rule bottom-up in each case. This extends $\branch_{k+1}$, so that it now has a top sequent of the form
$$
\ns\{[\emptyset; \phi_{1}^{\inp}, \psi_{1}^{\outp}]_{v_{1}}\}_{w_{1}} \cdots \{[\emptyset; \phi_{m}^{\inp}, \psi_{m}^{\outp}]_{v_{m}}\}_{w_{n}}
$$
 Once each branch $\branch_{i}$ has been processed for each $1 \leq i \leq n$, we continue to the next step.

\medskip

$\existsl$. Let $\branch_{1}, \ldots, \branch_{n}$ be all branches in the pseudo-derivation so far constructed with $\ns_{1}, \ldots, \ns_{n}$ the top sequents of each branch, respectively. We successively consider each $\branch_{i}$ for $1 \leq i \leq n$, performing a set of operations which extend the branch with $\existsl$ rules bottom-up. Suppose that $\branch_{1},\ldots,\branch_{k}$ have already been processed, so that $\branch_{k+1}$ is the current branch under consideration. Let $\ns_{k+1}$ be of the form
$$
\ns\{\va_{1}; \exists x_{1} \phi_{1}^{\inp}\}_{w_{1}} \cdots \{\va_{m}; \exists x_{m} \phi_{m}^{\inp}\}_{w_{m}}
$$
 with $\exists x_{m} \phi_{i}^{\inp}$ all existential input formulae in $\ns_{k+1}$. We successively consider each existential input formula, and apply the $\existsl$ rule bottom-up in each case. This extends $\branch_{k+1}$, so that it now has a top sequent of the form
$$
\ns\{\va_{1},y_{1}; \phi_{1}(y_{1}/x_{1})^{\inp}\}_{w_{1}} \cdots \{\va_{m},y_{m}; \phi_{m}(y_{m}/x_{m})^{\inp}\}_{w_{m}}
$$
 with $y_{1},\ldots,y_{m}$ fresh variables. Once each branch $\branch_{i}$ has been processed for each $1 \leq i \leq n$, we continue to the next step.

\medskip

$\existsri$. Let $\branch_{1}, \ldots, \branch_{n}$ be all branches in the pseudo-derivation so far constructed with $\ns_{1}, \ldots, \ns_{n}$ the top sequents of each branch, respectively. We successively consider each $\branch_{i}$ for $1 \leq i \leq n$, performing a set of operations which extend the branch with $\existsri$ rules bottom-up. Suppose that $\branch_{1},\ldots,\branch_{k}$ have already been processed, so that $\branch_{k+1}$ is the current branch under consideration. Let $\ns_{k+1}$ be of the form
$$
\ns\{\exists x_{1} \phi_{1}^{\outp}\}_{w_{1}} \cdots \{\exists x_{m} \phi_{m}^{\outp}\}_{w_{m}}
$$
 with $\exists x_{i} \phi_{i}$ all existential output formulae in $\ns_{k+1}$. We successively consider each existential output formula, and apply the $\existsri$ rule bottom-up in each case. Suppose we have already processed $\exists x_{1} \phi_{1}^{\outp}, \ldots, \exists x_{\ell} \phi_{\ell}^{\outp}$, so that $\exists x_{\ell+1} \phi_{\ell+1}$ is the current existential output formula under consideration. For each label $u$ occurring in the top nested sequent $\nsb$ of the branch extending $\branch_{k+1}$ such that $\rable{u}{w_{\ell+1}}$ and $y : \vb(u) \in \nsb$, successively apply the $\existsri$ rule bottom-up instantiating $\exists x_{\ell+1} \phi_{\ell+1}$ with each label $y$. Once each branch $\branch_{i}$ has been processed for each $1 \leq i \leq n$, we continue to the next step.
 
 \medskip

$\allli$. Similar to the $\existsri$ case above.

\medskip

$\allr$. Similar to the $\impr$ and $\existsl$ cases above.

\medskip

$\lin$. Let $\branch_{1}, \ldots, \branch_{n}$ be all branches in the pseudo-derivation so far constructed with $\ns_{1}, \ldots, \ns_{n}$ the top sequents of each branch, respectively. We successively consider each $\branch_{i}$ for $1 \leq i \leq n$, performing a set of operations which extend the branch with $\lin$ rules bottom-up. Suppose that $\branch_{1},\ldots,\branch_{k}$ have already been processed, so that $\branch_{k+1}$ is the current branch under consideration with $\ns_{k+1}$ its top sequent. Choose any pair of nestings $[\nsb]_{u},[\nsc]_{v}$ occurring side-by-side at any depth in $\ns_{k+1}$ and bottom-up apply the $\lin$ rule to this pair; repeat this process until all branches extending $\branch_{k+1}$ have linear nested sequents as top sequents. We are guaranteed that this process will terminate as $\ns_{k+1}$ is finite and each bottom-up application of $\lin$ shifts branching toward the leaves of a nested sequent until it is completely linearized. Once each branch $\branch_{i}$ has been processed for each $1 \leq i \leq n$, we continue to step $\idfo$ and $\botl$ above.

\medskip

\noindent
This concludes the description of $\prove$.\\

If $\prove$ returns $\success$, then by \lem~\ref{lem:fresh-delete}, a proof of the input has been found with the caveat that all redundant inferences generated by the $\idfo$ and $\botl$ step must be contracted. If $\prove$ does not halt, then it generates an infinite tree with finite branching. Hence, by K\"onig's lemma, we know that an infinite branch $\branch$ exists. We define a model $M = (W,\leq,D,V)$ by means of this branch accordingly:
\begin{itemize}

\item $w \in W$ \iffi $w = w_{0}$ or $w$ is a label occurring in $\branch$;

\item $w \leq u$ \iffi there exists a nested sequent $\nsb$ in $\branch$ such that $w,u \in \lab(\nsb)$ and $\rable{w}{u}$;

\item $x \in D(w)$ \iffi there exists a nested sequent $\nsb$ in $\branch$ such that (1) $u,w \in \lab(\nsb)$, (2) $\rable{u}{w}$, and (3) $x : \vb(u) \in \nsb$;

\item $(x_{1}, \ldots, x_{n}) \in V(p,w)$ \iffi there exists a nested sequent $\nsb$ in $\branch$ such that (1) $u,w \in \lab(\nsb)$, (2) $\rable{u}{w}$, and (3) $ u : p(x_{1}, \ldots, x_{n})^{\inp} \in \nsb$.

\end{itemize}
 Let us verify that $M$ is indeed an \nd-model. First, we know that $W \neq \emptyset$ because $w_{0} \in W$. Second, we show that $\leq$ is reflexive, transitive, and connected. The $\leq$ relation is reflexive by definition, so let us first show that it is transitive, whereby we assume that for $w,u,v \in W$, $w \leq u$ and $u \leq v$. Then, there exists a nested sequent $\ns_{1}$ such that $\rable{w}{u}$ holds and a nested sequent $\ns_{2}$ such that $\rable{u}{v}$ holds. We know that either $\ns_{1}$ occurs above $\ns_{2}$, vice-versa, or the two are identical. We suppose the first case without loss of generality. By \lem~\ref{lem:rable-preserved-up}, we know that $\rable{u}{v}$ holds for $\ns_{1}$ as well, and thus, $\rable{w}{v}$ holds for $\ns_{1}$, showing that $w \leq v$. We now show that $M$ is connected and suppose that for $w,u,v \in W$, $w \leq u$ and $w \leq v$. Then, there exists a nested sequent $\ns_{1}$ such that $\rable{w}{u}$ holds and a nested sequent $\ns_{2}$ such that $\rable{w}{v}$ holds. By the $\lin$ step in $\prove$, we know that there will exist a nested sequent $\ns_{3}$ above $\ns_{1}$ and $\ns_{2}$ that is \emph{linear}. Thus, either $\rable{u}{v}$ or $\rable{v}{u}$ will hold in $\ns_{3}$, showing that either $u \leq v$ or $v \leq u$.
 
 We now show that $M$ satisfies the (ND) property. Suppose for $w,u \in W$ that $x \in D(w)$ and $w \leq u$. By the first fact, there exists a nested sequent $\ns_{1}$ in $\branch$ containing the labels $v$ and $w$ such that $\rable{v}{w}$ and $x \in \va$ with $\va$ the signature of the $v$-component. By the second fact, we know there exists a nested sequent $\ns_{2}$ in $\branch$ containing the labels $w$ and $u$ such that $\rable{w}{u}$. We know that either $\ns_{1}$ occurs above $\ns_{2}$ in $\branch$, vice-versa, or both are identical. We suppose the first case without loss of generality. Therefore, by \lem~\ref{lem:rable-preserved-up}, we know that $\rable{v}{u}$ holds in $\ns_{1}$, implying that $\rable{v}{u}$ holds as well. Thus, $x \in D(u)$ by the definition of $M$ above.
 
 We must additionally show that (i) for each $w \in W$, $D(w) \neq \emptyset$, and (ii) for each $w \in W$ and $p \in \pred$, if $\ari{p} = n$, then $V(p,w) \subseteq D(w)^{n}$. (i) Since our input is of the form $z,\vec{x};\phi(\vec{x})$ with $z$ a fresh variable, and because every rule of $\nnd$ bottom-up preserves the place of variables (\lem~\ref{lem:rable-preserved-up}), we have that $z$ will occur in the signature of the $w_{0}$-component (i.e. root) of each nested sequent in $\branch$. As $w_{0} \leq w$ for every label $w \in W$, it follows by the definition of $M$ that $z \in D(w)$ for every $w \in W$. (ii) Let $w \in W$ and $p \in \pred$ with $\ari{p} = n$. Assume $(x_{1}, \ldots, x_{n}) \in V(p,w)$. Then, there exists a nested sequent $\nsb$ in $\branch$ such that $u,w \in \lab(\nsb)$, $\rable{u}{w}$, and $ u : p(x_{1}, \ldots, x_{n})^{\inp} \in \nsb$. By the $\doms$ step of $\prove$, $x_{1} : \vb(u) \in \nsb$, $\ldots$, $x_{n} : \vb(u) \in \nsb$, which shows that $x_{1}, \ldots, x_{n} \in D(w)$ by the definition of $D$ above, and so, $V(p,w) \subseteq D(w)^{n}$.
 
 Last, we must show that $M$ satisfies the monotonicity condition (M).  Assume we have $w, u \in W$ such that $w \leq u$ and $(x_{1}, \ldots, x_{n}) \in V(p,w)$. By the first fact, there exists a nested sequent $\ns_{1}$ in $\branch$ containing the labels $w$ and $u$ such that $\rable{w}{u}$. By the second fact, there exists a nested sequent $\ns_{2}$ in $\branch$ containing the labels $v$ and $w$ such that $\rable{v}{w}$ and $p(x_{1}, \ldots, x_{n})^{\inp}$ occurs in the $v$-component of $\ns_{1}$. We know that either $\ns_{1}$ occurs above $\ns_{2}$ in $\branch$, vice-versa, or both are identical. We suppose the first case without loss of generality. Then, by \lem~\ref{lem:rable-preserved-up} we know (1) that $\rable{v}{w}$ holds for $\ns_{1}$, and (2) $p(x_{1}, \ldots, x_{n})^{\inp}$ occurs in the $v$-component of $\ns_{1}$. By (1), we have that $\rable{v}{u}$ holds, which implies that $(x_{1}, \ldots, x_{n}) \in V(p,u)$ by (2).
 
 By the argument above, we know that $M$ is indeed an \nd-model. Let us now define the assignment $\mu$ to be the identity function on $D(W)$ (mapping all other variables in $\vars \setminus D(W)$ arbitrarily). We now prove for each $\nsb \in \branch$, (1) if $w : \xi^{\inp} \in \nsb$, then $M,w,\mu \Vdash \xi$, and (2) if $w : \xi^{\outp} \in \nsb$, then $M,w,\mu \not\Vdash \xi$. We prove (1) and (2) by mutual induction on the complexity of $\xi$, and let $\nsb$ be an arbitrary nested sequent in $\branch$.
 
\begin{itemize}

\item $w : p(x_{1},\ldots,x_{n})^{\inp} \in \nsb$. By the definition of $V$, we know that $(x_{1},\ldots,x_{n}) \in V(p,w)$, from which it follows that $M,w,\mu \Vdash p(x_{1},\ldots,x_{n})$.

\item $w : p(x_{1},\ldots,x_{n})^{\outp} \in \nsb$. Suppose there exists a nested sequent $\nsc$ in $\branch$ such that (1) $u,w \in \lab(\nsb)$, (2) $\rable{u}{w}$, and (3) $p(x_{1}, \ldots, x_{n})^{\inp}$ occurs in the $u$-component of $\nsc$. Then, either $\nsc$ occurs above $\nsb$ in $\branch$, vice-versa, or $\nsc = \nsb$. By \lem~\ref{lem:rable-preserved-up}, we have that $\nsc$ is an instance of $\idfo$ in the first case, and $\nsb$ is an instance of $\idfo$ in the second and third cases. Therefore, by the $\idfo$ step of $\prove$, we would have that $\branch$ is finite, in contradiction to our assumption. By the definition of $V$, we have that $(x_{1},\ldots,x_{n}) \not\in V(p,w)$, showing that $M,w,\mu \not\Vdash p(x_{1},\ldots,x_{n})$.

\item $w : \bot^{\inp} \in \nsb$. Then, by the $\botl$ step of $\prove$ proof-search will terminate on $\branch$. This implies that $\branch$ is finite in contradiction to our assumption, from which it follows that for every nested sequent $\nsc$ in $\branch$, $u : \bot^{\inp} \not\in \nsc$ for every $u \in \lab(\nsc)$.

\item $w : \bot^{\outp} \in \nsb$. The case is trivial as $M,w,\mu \not\Vdash \bot$.

\item $w : \psi \lor \chi^{\inp} \in \nsb$. By the $\disl$ step of $\prove$, we know that a nested sequent $\nsc$ exists in $\branch$ such that either $w : \psi^{\inp} \in \nsc$ or $w : \chi^{\inp} \in \nsc$. By IH, we have that $M,w,\mu \Vdash \psi$ or $M,w,\mu \Vdash \chi$, showing that $M,w,\mu \Vdash \psi \lor \chi$.

\item $w : \psi \lor \chi^{\outp} \in \nsb$. By the $\disr$ step of $\prove$, we know that a nested sequent $\nsc$ exists in $\branch$ such that $w : \psi^{\outp}, w : \chi^{\outp} \in \nsc$. By IH, we have that $M,w,\mu \not\Vdash \psi$ and $M,w,\mu \not\Vdash \chi$, showing that $M,w,\mu \not\Vdash \psi \lor \chi$.

\item $w : \psi \land \chi^{\inp} \in \nsb$. By the $\conl$ step of $\prove$, we know that a nested sequent $\nsc$ exists in $\branch$ such that $w : \psi^{\inp}, w : \chi^{\inp} \in \nsc$. By IH, we have that $M,w,\mu \Vdash \psi$ and $M,w,\mu \Vdash \chi$, showing that $M,w,\mu \Vdash \psi \land \chi$.

\item $w : \psi \land \chi^{\outp} \in \nsb$. By the $\conr$ step of $\prove$, we know that a nested sequent $\nsc$ exists in $\branch$ such that either $w : \psi^{\outp} \in \nsc$ or $w : \chi^{\outp} \in \nsc$. By IH, we have that either $M,w,\mu \not\Vdash \psi$ or $M,w,\mu \not\Vdash \chi$, showing that $M,w,\mu \not\Vdash \psi \land \chi$.

\item $w : \psi \imp \chi^{\inp} \in \nsb$. Suppose that $M,w,\mu \Vdash \psi$ and let $w \leq u$ for an arbitrary $u$ in $W$. By the definition of $\leq$ we know that there exists a nested sequent $\nsc$ in $\branch$ such that $\rable{w}{u}$. Hence, by \lem~\ref{lem:rable-preserved-up}, $\rable{w}{u}$ will hold for every nested sequent above $\nsc$ in $\branch$. At some point during the computation of $\prove$, the $\impl$ step will be reached, showing that for some $\nsd$ in $\branch$ above $\nsb$ and $\nsc$ either $u : \psi^{\outp}$ or $u : \chi^{\inp}$. By IH, we have that either $M,u,\mu \not\Vdash \psi$ or $M,u,\mu \Vdash \chi$, showing that $M,w,\mu \Vdash \psi \imp \chi$.

\item $w : \psi \imp \chi^{\outp} \in \nsb$. By the $\impr$ step of $\prove$, there will exist a nested sequent $\nsc$ above $\nsb$ of the form $\nsc\{[\psi^{\inp},\chi^{\outp}]_{v}\}_{w}$. By IH, we have that $M,v,\mu \Vdash \psi$ and $M,v,\mu \not\Vdash \chi$, and by the definition of $\leq$, we know that $w \leq u$. Therefore, $M,w,\mu \not\Vdash \psi \imp \chi$.

\item $w : \exists x \psi^{\inp} \in \nsb$. By the $\existsl$ step of $\prove$, there will exist a nested sequent $\nsc$ above $\nsb$ of the form $\nsc\{\va,y; \psi(y/x)^{\inp}\}_{w}$ with $y$ fresh. By the definition of $D$, we know that $y \in D(w)$, and by IH, we have that $M,w,\mu[y/y] \Vdash \psi(y/x)$, showing that $M,w,\mu \Vdash \exists x \psi$.

\item $w : \exists x \psi^{\outp} \in \nsb$. Let $y$ be an arbitrary variable in $D(w)$. By the definition of $D$, there exists a nested sequent $\nsc$ in $\branch$ such that $u,w \in \lab(\nsb)$, $\rable{u}{w}$, and $y$ occurs in the signature of the $u$-component of $\nsc$. There will be a nested sequent above $\nsb$ and $\nsc$ in $\branch$ for which $\rable{u}{w}$ will hold and $y$ will occur in the signature of its $u$-component, and for which $\existsri$ will be applied bottom-up. This will yield a nested sequent $\nsd$ in $\branch$ such that $w : \psi(y/x)^{\outp} \in \nsd$. By IH, $M,w,\mu \not\Vdash \psi(y/x)$, which implies that $M,w,\mu \not\Vdash \exists x \psi$ since $y$ was arbitrarily chosen.

\item $w : \forall x \psi^{\inp} \in \nsb$. Let $y$ be an arbitrary variable in $D(u)$ and suppose that $w \leq u$ holds. By the definition of $D$, there exists a nested sequent $\nsc$ in $\branch$ such that $v,u \in \lab(\nsb)$, $\rable{v}{u}$, and $y$ occurs in the signature of the $v$-component of $\nsc$. By the definition of $\leq$, there exists a nested sequent $\nsd$ in $\branch$ such that $w,u \in \lab(\nsd)$ and $\rable{w}{u}$. Hence, there will be a nested sequent above $\nsb$, $\nsc$, and $\nsd$ in $\branch$ such that $\rable{v}{u}$, $\rable{w}{u}$, and $y$ occurs in the signature of its $v$-component, for which the $\allli$ step of $\prove$ will be applicable, and thus the nested sequent above it in $\branch$ will contain $\psi(y/x)$ in its $u$-component. By IH, we have that $M,u,\mu \Vdash \psi(y/x)$, from which it follows that $M,w,\mu \Vdash \forall x \psi$ by our assumptions.

\item $w : \forall x \psi^{\outp} \in \nsb$. By the $\allr$ step of $\prove$, there will exist a nested sequent $\nsc$ above $\nsb$ of the form $\nsc\{[y; \psi(y/x)^{\outp}]_{u}\}_{w}$ with $y$ fresh. By the definition of $D$, we know that $y \in D(u)$, and by the definition of $\leq$, we know that $w \leq u$. By IH, we have that $M,u,\mu \not\Vdash \psi(y/x)$, which implies that $M,w,\mu \not\Vdash \forall x \psi$.

\end{itemize}

 If we define $\iota$ to be an $M$-interpretation that is the identity function on the set $\lab$ of labels, then by the proof above, $M,\iota,\mu,w_{0} \not\models \vec{x}; \phi(\vec{x})$. Thus, we have shown that if a nested sequent of the form $\vec{x}; \phi(\vec{x})$ does not have a proof in $\nndl$, then it is \ndl-invalid, implying that every such \ndl-valid nested sequent is provable in $\nndl$.
\end{proof}

%%Appendix B
\section{Errata and Notes on Published Version}

\noindent
\textbf{1.} Although it is neither mentioned in this manuscript nor in the published version~\citep{LyoTF23}, the $\doms$, $\existsrii$, and $\alllii$ rules are admissible in $\ncd$ and $\ncdl$. First, the $\doms$ step of $\prove$ in the proof of cut-free completeness (\thm~\ref{thm:cut-free-comp}) can be omitted, and in the counter-model construction we define the domain $D(w) := \vars$ for each $w \in W$ of the extracted counter-model $M$. The semantic condition encoded by $\doms$, namely `$V(p,w) \subseteq D(w)^{n}$ with $\ari{p} = n$', will trivially hold, showing that $\doms$ is unneeded and admissible. Second, the $\existsrii$ and $\alllii$ rules are used in the proof of \lem~\ref{lem:fresh-delete}, which is used in the proof of cut-free completeness to (1) remove the starting variable $z$ if a proof of the input is found and (2) to ensure that all domains of the extracted counter-model are non-empty (as they contain $z$) if a proof of the input is not found. As mentioned above, in the counter-model construction we define the domain $D(w) := \vars$ for each $w \in W$ in $\ncd$ and $\ncdl$ cases, meaning all domains will be non-empty by definition, and thus, cut-free completeness can be shown for $\ncd$ and $\ncdl$ without considering the $\existsrii$ and $\alllii$ rules.\\

\noindent
\textbf{2.} I have changed the cut-free completeness theorem (\thm~\ref{thm:cut-free-comp}) in this arXiv manuscript from the one stated in the published version of the paper~\citep{LyoTF23}. In the published version of the paper, the cut-free completeness theorem reads: 
\begin{flushleft}
(1) ``Let $\mathrm{L} \in \{\ipc, \nd, \cd, \gd, \ndl, \cdl\}$. If a nested sequent $\nsa$ is $\mathrm{L}$-valid, then it is derivable in $\ncalc_{\mathrm{L}}$.''
\end{flushleft}
whereas in this manuscript it reads:
\begin{flushleft}
(2) ``Let $\mathrm{L} \in \{\ipc, \nd, \cd, \gd, \ndl, \cdl\}$ and $\vec{x}; \phi(\vec{x})$ be a nested sequent with $\vec{x}$ all free variables in $\phi(\vec{x})$. If a nested sequent $\vec{x}; \phi(\vec{x})$ is $\mathrm{L}$-valid, then it is derivable in $\ncalc_{\mathrm{L}}$.''
\end{flushleft}
I have discovered that the proof of the former claim (stated in the published version~\citep{LyoTF23}) contains an error and is incorrect, though the new proof of the latter claim (which resides in \app~\ref{app:completeness} of this manuscript) fixes this error. Below, I will explain the issue in the proof of statement (1) and clarify how the new version of cut-free completeness solves the pinpointed issue.

Let us take the nested sequent $\nsa = z; \forall x p(x) \imp p(y)$ as input to the $\prove$ algorithm (which is the same in both the proof of claim (1) and claim (2)), where $z$ is the starting variable. $\prove$ constructs the following proof in a bottom-up manner, and after the first $\impr$ inference, repeatedly applies the $\allli$ rule, adding redundant copies of $p(z)$, ad infinitum.
\begin{center}
\AxiomC{$\vdots$}
\RightLabel{$\allli$}
\UnaryInfC{$z ; [\emptyset; \forall x p(x)^{\inp}, p(z)^{\inp}, p(z)^{\inp}, p(y)^{\outp}]_{u}$}
\RightLabel{$\allli$}
\UnaryInfC{$z ; [\emptyset; \forall x p(x)^{\inp}, p(z)^{\inp}, p(y)^{\outp}]_{u}$}
\RightLabel{$\allli$}
\UnaryInfC{$z ; [\emptyset; \forall x p(x)^{\inp}, p(y)^{\outp}]_{u}$}
\RightLabel{$\impr$}
\UnaryInfC{$z ;\forall x p(x) \imp p(y)^{\outp}$}
\DisplayProof
\end{center}
The proof itself is an infinite branch $\branch$, from which we may extract a $\ndl$-model $M = (W,\leq,D,V)$ such that $W = \{w_{0},u\}$, $\leq = \{(w_{0},w_{0}),(w_{0},u),(u,u)\}$, $D(w_{0}) = D(u) = \{z\}$, $V(p,w_{0}) = \emptyset$, and $V(p,u) = \{z\}$. Moreover, as in the proof of claim (1), we define the $M$-assignment $\mu : \vars \to D(W)$ such that $\mu$ is the identity function on the elements (which are variables) in $D(W)$ and $\mu$ maps every other variable in $\vars \setminus D(W)$ arbitrarily into $D(W)$. We now make the important observation that $D(W) = \{z\}$, meaning $\mu(y) = z$. In the proof of claim (1) in the published version~\citep{LyoTF23}, we argue the following by a mutual induction on the complexity of $\xi$:
\begin{flushleft}
``For each $\nsb \in \branch$, (1) if $w : \xi^{\inp} \in \nsb$, then $M,w,\mu \Vdash \xi$, and (2) if $w : \xi^{\outp} \in \nsb$, then $M,w,\mu \not\Vdash \xi$.''
\end{flushleft}
The issue that arises is the following: although $p(y)^{\outp}$ is an output formula occurring in, e.g., the nested sequent $z ; [\emptyset; \forall x p(x)^{\inp}, p(y)^{\outp}]_{u}$ of $\branch$, we have that $M,u,\mu \Vdash p(y)$ (rather than $M,u,\mu \not\Vdash p(y)$) because $\mu(y) = z \in \{z\} = V(p,u)$. Therefore, the above claim does not hold in the proof of claim (1) in the published version~\citep{LyoTF23}. The problem is the following:
\begin{flushleft}
\textit{We cannot ignore the free variables of the nested sequent input into $\prove$ as they are relevant in defining a counter-model of the input.}
\end{flushleft}

The current manuscript fixes this issue by taking the free variables of the input into account. Notice, for instance, that the above example does not cause problems in the proof of claim (2) since $z,y;\forall x p(x) \imp p(y)^{\outp}$ would be input into $\prove$, and the following proof would be found:
\begin{center}
\AxiomC{ }
\RightLabel{$\id$}
\UnaryInfC{$z,y ; [\emptyset; \forall x p(x)^{\inp}, p(z)^{\inp}, p(y)^{\inp}, p(y)^{\outp}]_{u}$}
\RightLabel{$\allli \times 2$}
\UnaryInfC{$z,y ; [\emptyset; \forall x p(x)^{\inp}, p(y)^{\outp}]_{u}$}
\RightLabel{$\impr$}
\UnaryInfC{$z,y ;\forall x p(x) \imp p(y)^{\outp}$}
\DisplayProof
\end{center}

\noindent
\textbf{3.} An interesting consequence of formulating the statement of cut-free completeness as in this manuscript (i.e. as claim (2) in note \textbf{2} above) is that the domain shift rule $\doms$ appears admissible in $\nnd$ and $\nndl$. (NB. By note \textbf{1} above, we then have that $\doms$ is admissible in every nested calculus considered in this manuscript.) The reason being, due to the shape of the nested sequent considered in cut-free completeness, viz. $\vec{x}; \phi(\vec{x})$, it appears that one can prove the following lemma, where we let $\prove'$ be $\prove$ without the $\doms$ case:

\begin{lemma}
Let $\vec{x}; \phi(\vec{x})$ be a nested sequent with $\vec{x}$ all free variables in $\phi(\vec{x})$ and let $\prf$ be the pseudo-derivation constructed by $\prove'(\vec{x}; \phi(\vec{x}))$. For any nested sequent $\nsa$ occurring in $\prf$, if $w : \psi(\vec{y})^{\io} \in \nsa$ with $\io \in \{\inp,\outp\}$ and $\vec{y} = y_{1}, \ldots, y_{k}$ all free variables in $\psi(\vec{y})^{\io}$, then there exist $u_{1}, \ldots, u_{k} \in \lab(\nsa)$ such that $\rable{u_{i}}{w}$ and $y_{i} : \va_{i}(u_{i}) \in \nsa$ for $1 \leq i \leq k$.
\end{lemma}

The above lemma is proven by considering $\prf$ in a bottom-up manner, observing that $\vec{x}; \phi(\vec{x})$ satisfies the property mentioned in the lemma, and then checking that each inference rule preserves this property when applied bottom-up (which is straightforward to verify by inspecting the rules of $\nnd$ and $\nndl$). Then, by the definition of $D$ and $V$ in the counter-model $M$ constructed in the proof of claim (2) in note \textbf{2} above, i.e. \thm~\ref{thm:cut-free-comp} of this manuscript, it directly follows that for each $p \in \pred$ and $w \in W$ of $M$, $V(p,w) \subseteq D(w)^{n}$ with $\ari{p} = n$. Therefore, the $\doms$ rule is not needed to ensure this condition.\\

\end{document}